\documentclass[english,11pt,a4paper]{article}
\usepackage{microtype}
\usepackage{amsthm}
\usepackage{graphicx} 
\usepackage{array} 
\usepackage{colortbl,makecell}
\usepackage{pifont} 
\usepackage{cite}
\usepackage{amsmath, amssymb, amsfonts, verbatim}
\usepackage{hyphenat, epsfig, subcaption, multirow}
\usepackage{nicefrac}
\usepackage{paralist,enumitem}
\setitemize{noitemsep,topsep=0pt}
\setenumerate{noitemsep,topsep=0pt}

\usepackage{dsfont} 

\usepackage{mathtools, etoolbox}
\usepackage{thmtools}
\usepackage{thm-restate}
\usepackage{xparse}

\usepackage[linesnumbered,lined,algoruled,noend]{algorithm2e}\DontPrintSemicolon
\SetKwProg{Fn}{function}{\string:}{}
\SetKwBlock{Repeat}{repeat}{}
\SetKwFor{ForWithoutDo}{for}{}{}

\usepackage[font=small,labelfont=bf]{caption}

\usepackage{xcolor}
\definecolor{myblue}{HTML}{0088cc}
\definecolor{myorange}{HTML}{f26924}
\colorlet{paleblue}{myblue!10!white}
\colorlet{paleorange}{myorange!10!white}

\usepackage{babel}
\usepackage{csquotes}

\DeclareFontFamily{U}{mathx}{\hyphenchar\font45}
\DeclareFontShape{U}{mathx}{m}{n}{
      <5> <6> <7> <8> <9> <10>
      <10.95> <12> <14.4> <17.28> <20.74> <24.88>
      mathx10
      }{}
\DeclareSymbolFont{mathx}{U}{mathx}{m}{n}
\DeclareMathSymbol{\bigtimes}{1}{mathx}{"91}

\usepackage{tcolorbox}
\tcbuselibrary{skins,breakable}
\tcbset{enhanced jigsaw}

\usepackage{nameref}
\usepackage[linktocpage=true,colorlinks,allcolors=myblue]{hyperref}
\DeclareUrlCommand{\path}{}
\usepackage[noabbrev,nameinlink,capitalize]{cleveref}
\crefname{property}{property}{Property}
\crefname{invariant}{Invariant}{Invariants}
\creflabelformat{property}{(#1)#2#3}
\crefname{equation}{Eq}{Eq}
\crefname{observation}{Observation}{Observations}
\creflabelformat{equation}{(#1)#2#3}

\usepackage{bm}
\usepackage{url}
\usepackage{xspace}
\usepackage[mathscr]{euscript}
\usepackage{mathrsfs}

\usepackage{tikz}
\usetikzlibrary{arrows}
\usetikzlibrary{arrows.meta}
\usetikzlibrary{shapes}
\usetikzlibrary{backgrounds}
\usetikzlibrary{positioning}
\usetikzlibrary{decorations.markings}
\usetikzlibrary{patterns}
\usetikzlibrary{calc}
\usetikzlibrary{fit}
\tikzset{vertex/.style={circle, black, fill=Yellow, line width=1pt, draw, minimum width=8pt, minimum height=8pt, inner sep=0pt}}
\usepackage{wrapfig}

\usepackage[framemethod=TikZ]{mdframed}

\usepackage[margin=1in]{geometry}

\usepackage{soul}

\newtheorem{lemma}{Lemma}[section]

\newtheorem{theorem}[lemma]{Theorem}
\newtheorem{proposition}[lemma]{Proposition}
\newtheorem{corollary}[lemma]{Corollary}

\newtheorem{fact}[lemma]{Fact}

\newtheorem{definition}[lemma]{Definition}
\newtheorem{problem}{Problem}

\newtheorem*{claim*}{Claim}
\newtheorem*{proposition*}{Proposition}
\newtheorem*{lemma*}{Lemma}
\newtheorem*{problem*}{Problem}

\crefname{assumption}{Assumption}{Assumptions}
\crefname{lemma}{Lemma}{Lemmas}
\crefname{claim}{Claim}{Claims}
\crefname{enumi}{Step}{Steps}
\crefname{step}{Step}{Step}

\newtheorem{remark}[lemma]{Remark}

\renewcommand{\qed}{\nobreak \ifvmode \relax \else
      \ifdim\lastskip<1.5em \hskip-\lastskip
      \hskip1.5em plus0em minus0.5em \fi \nobreak
      \vrule height0.75em width0.5em depth0.25em\fi}

\renewcommand{\leq}{\leqslant}
\renewcommand{\geq}{\geqslant}
\renewcommand{\le}{\leqslant}
\renewcommand{\ge}{\geqslant}

\DeclareMathOperator{\Tr}{Tr}
\newcommand{\supp}{\support}

\newcommand{\As}{\mathcal{A}}
\newcommand{\Bs}{\mathcal{B}}
\newcommand{\Es}{\mathcal{E}}
\newcommand{\C}{\mathbb{C}}

\usepackage{booktabs}
\usepackage{array}
\usepackage{longtable}

\newcolumntype{L}[1]{>{\raggedright\arraybackslash}p{#1}}

\DeclarePairedDelimiter{\bracket}[]
\DeclarePairedDelimiter{\paren}()

\DeclarePairedDelimiter{\card}{\vert}{\vert}
\DeclarePairedDelimiter{\norm}{\|}{\|}
\DeclarePairedDelimiter{\abs}{|}{|}

\DeclarePairedDelimiter{\set}{\{}{\}}

\DeclareMathOperator*{\support}{supp}
\DeclareMathOperator*{\expect}{\mathbb{E}}

\DeclareMathOperator{\Trace}{Tr}
\DeclarePairedDelimiterXPP{\tr}[1]{\Trace}(){}{#1}

\DeclarePairedDelimiterXPP{\Ot}[1]{\widetilde{O}}(){}{#1}
\DeclarePairedDelimiterXPP{\Omgt}[1]{\widetilde{\Omega}}(){}{#1}
\DeclarePairedDelimiterXPP{\BigO}[1]{O}(){}{#1}

\NewDocumentCommand{\Prob}{sO{}E{_}{{}}m}{%
  {\boldsymbol{\mathbb{P}}}_{#3}
  \IfBooleanTF{#1}
  {\bracket*{#4}}
  {\bracket[#2]{#4}}
}

\NewDocumentCommand{\Exp}{sO{}E{_}{{}}m}{%
  \expect_{#3}
  \IfBooleanTF{#1}
  {\bracket*{#4}}
  {\bracket[#2]{#4}}
}

\renewcommand{\epsilon}{\varepsilon}
\newcommand{\eps}{\varepsilon}

\newcommand{\poly}{\mbox{\rm poly}}

\newenvironment{tbox}{\begin{tcolorbox}[
		enlarge top by=5pt,
		enlarge bottom by=5pt,
    boxsep=0pt,
    left=8pt,
    right=8pt,
    top=10pt,
    bottom=10pt,
    arc=2pt,
    boxrule=1pt,
    toprule=1pt,
    colback=paleblue,
    colframe=paleblue
    ]
	}
{\end{tcolorbox}}

\newtheorem{mdresult}{Result}
\crefname{mdresult}{Result}{Results}
\Crefname{mdresult}{Result}{Results}
\newenvironment{result}{\begin{tcolorbox}[
		enlarge top by=5pt,
		enlarge bottom by=5pt,
    breakable,
    boxsep=0pt,
    left=8pt,
    right=8pt,
    top=10pt,
    bottom=10pt,
    arc=2pt,
    boxrule=1pt,
    toprule=1pt,
    colback=paleorange,
    colframe=paleorange
    ]
    \begin{mdresult}
	}
{\end{mdresult}\end{tcolorbox}}

\theoremstyle{definition}
\newtheorem{mdalg}{Algorithm}
\crefname{mdalg}{Algorithm}{Algorithms}
\Crefname{mdalg}{Algorithm}{Algorithms}
\newenvironment{alg}{\begin{tcolorbox}[
		enlarge top by=5pt,
		enlarge bottom by=5pt,
    boxsep=0pt,
    left=8pt,
    right=8pt,
    top=10pt,
    bottom=10pt,
    arc=2pt,
    boxrule=1pt,
    toprule=1pt,
    colback=white,
    ]
    \begin{mdalg}
	}
{\end{mdalg}\end{tcolorbox}}

	{\noindent {}{#1}{}}{ \strut\hfill \qed \vspace{2ex}}

\newcommand{\fS}{\mathcal{S}}

\newcommand{\bN}{\mathbb{N}}
\newcommand{\bZ}{\mathbb{Z}}
\newcommand{\bR}{\mathbb{R}}

\newcommand{\ket}[1]{|#1\rangle}
\newcommand{\bra}[1]{\langle#1|}
\newcommand{\braket}[2]{\langle #1|#2\rangle}
\newcommand{\ketbra}[2]{|#1\rangle\langle#2|}
\newcommand{\braUket}[3]{\langle #1|#2|#3 \rangle}

\newcommand{\bC}{\mathbb{C}}

\usepackage{dsfont}

\newcommand{\qLOCAL}{quantum-LOCAL\xspace}

\newcommand{\maxime}[1]{\textcolor{orange}{[Maxime: #1]}}

\newcommand{\myaff}[1]{\,$\cdot$\, {\small #1}\par\medskip}

\newenvironment{myabstract}
{\list{}{\listparindent 1.5em
        \itemindent    \listparindent
        \leftmargin    0cm
        \rightmargin   0cm
        \parsep        0pt}%
    \item\relax}
{\endlist}

\newenvironment{mycover}
{\list{}{\listparindent 0pt
        \itemindent    \listparindent
        \leftmargin    0cm
        \rightmargin   1.5cm
        \parsep        0pt}%
    \raggedright
    \item\relax}
{\endlist}

\hypersetup{
    pdftitle={Quantum Advantage for Distributed Symmetry Breaking},
    pdfauthor={Maxime Flin, Longcheng Li, Jukka Suomela}
}

\begin{document}
\begin{mycover}
{\huge\bfseries Quantum Advantage for \\ Distributed Symmetry Breaking \par}
\bigskip
\bigskip

\textbf{Maxime Flin}
\myaff{Aalto University, Finland}


\textbf{Longcheng Li}
\myaff{University of Cambridge, UK}

\textbf{Jukka Suomela}
\myaff{Aalto University, Finland}
\bigskip
\end{mycover}

\begin{myabstract}
\noindent\textbf{Abstract.}
We present a distributed quantum algorithm that $3$-colors cycles in $O(1)$ rounds, with high probability. It follows that all locally checkable labeling problems (LCLs) that have round complexity $O(\log^* n)$ in the classical LOCAL model can be solved in $O(1)$ rounds in the quantum-LOCAL model, with high probability; this includes problems such as maximal independent set and maximal matching in bounded-degree graphs. This presents the first \emph{natural} examples of graph problems with an asymptotic distributed quantum advantage for the LOCAL model; all prior examples that separate LOCAL and quantum-LOCAL are artificial problems constructed merely for the sake of demonstrating quantum advantage.
\end{myabstract}

\thispagestyle{empty}
\newpage
\tableofcontents
\newpage
\pagenumbering{arabic}

\section{Introduction}

We solve some of the biggest open questions related to quantum advantage in distributed algorithms. We show that there are constant-round quantum algorithms for many natural symmetry-breaking problems, such as graph coloring, maximal independent set, and maximal matching, in bounded-degree graphs.

\subsection{Distributed Quantum Advantage in the LOCAL Model}

The most demanding setting for studying \emph{distributed quantum advantage} is the LOCAL model of distributed computing. In the classical \emph{LOCAL model} \cite{linial-1992-locality-in-distributed-graph-algorithms,peleg-2000-distributed-computing-a-locality-sensitive}, the nodes of the input graph are classical computers, edges are classical communication channels, and computation proceeds in synchronous rounds; in each round each node can send a message to each neighbor, receive a message from each neighbor, and update its state; and the \emph{running time} of the algorithm is the number of rounds until all nodes stop and announce their own local output (for example, if we are coloring a graph, each node has to announce its own color). In the \emph{quantum-LOCAL model} \cite{gavoille-kosowski-markiewicz-2009-what-can-be-observed,arfaoui-fraigniaud-2014-what-can-be-computed-without} the nodes are quantum computers and they can exchange qubits with their neighbors. Notably, in the classical version there are no limitations on the size of the message, local memory, or local computation; hence we are asking if quantum communication and computation provide any advantage over a very powerful classical model.

So far there have been primarily negative results \cite{gavoille-kosowski-markiewicz-2009-what-can-be-observed,arfaoui-fraigniaud-2014-what-can-be-computed-without,coiteux-roy-d-amore-etal-2024-no-distributed-quantum,dhar-kujawa-etal-2024-local-problems-in-trees-across-a,akbari-coiteux-roy-etal-2025-online-locality-meets,balliu-coupette-etal-2025-new-limits-on-distributed,brandt-gottlicher-2026-a-post-quantum-lower-bound-for,fraigniaud-magniez-ziccardi-2027-no-distributed-quantum,d-amore-lievonen-2026-superlogarithmic-gap-result-for}. There are, in essence, only \emph{three} problem families \cite{le-gall-nishimura-rosmanis-2019-quantum-advantage-for,balliu-brandt-etal-2025-distributed-quantum-advantage,balliu-casagrande-etal-2026-distributed-quantum} that are known to admit asymptotically faster algorithms in quantum-LOCAL in comparison with classical LOCAL, and all of them are \emph{artificial} problems, engineered just for the sake of being easy to solve with quantum algorithms. We still do not know if there are any \emph{natural} graph problems that admit quantum advantage.

The most pressing case is the seemingly simple problem of graph coloring, and in particular the task of coloring cycles with $3$ colors and more generally coloring graphs of maximum degree $\Delta = O(1)$ with $\Delta + 1$ colors. These are very well-understood problems in the classical setting, and already in the 1990s it was known that they require $\Theta(\log^* n)$ rounds \cite{cole-vishkin-1986-deterministic-coin-tossing-with,goldberg-plotkin-shannon-1988-parallel-symmetry,linial-1992-locality-in-distributed-graph-algorithms,naor-1991-a-lower-bound-on-probabilistic-algorithms-for}. Could these problems be solved faster in the quantum-LOCAL model? This was explicitly mentioned as an open question already in \cite{gavoille-kosowski-markiewicz-2009-what-can-be-observed}, and it has been recently highlighted as one of the biggest open questions in this area by e.g.\ \cite{coiteux-roy-d-amore-etal-2024-no-distributed-quantum, le-gall-2025-recent-developments-in-quantum-distributed,brandt-gottlicher-2026-a-post-quantum-lower-bound-for,coiteux-roy-flin-etal-2026-distributed-quantum}.

So far the evidence has been pointing in the direction that coloring \emph{does not} admit quantum advantage; there are several negative results that exclude e.g.\ very fast quantum algorithms \cite{gavoille-kachigar-zemor-2019-localisation-resistant,le-gall-rosmanis-2022-non-trivial-lower-bound-for-3,gur-li-2026-impossibility-of-one-way-one-round-quantum} or fast quantum algorithms that succeed with probability $1$ in anonymous networks \cite{coiteux-roy-flin-etal-2026-distributed-quantum}. Hence it has been reasonable to conjecture that if there are any natural problems that admit quantum advantage, coloring is most likely not one of them.

\subsection{Our Results}
\label{sec:our-results}

The single result that does all the heavy lifting is, informally, this:
\begin{result}\label{res:main-epsilon}
    For every $\epsilon > 0$, there is a $1$-round one-way quantum-LOCAL algorithm that colors directed cycles with $2\cdot10^4$ colors so that for each edge the local failure probability is at most~$\epsilon$.
\end{result}
Notably, the number of qubits used by the algorithm depends on $\epsilon$, but the round complexity is independent of $\epsilon$. Our algorithm uses $O(\log 1/\eps)$ qubits (see \cref{sec:communication-complexity}), and this is optimal up to a constant factor (see \cref{appx:lower-bound}). At the limit, one can get a zero-error one-round one-way quantum-LOCAL algorithm where each node uses \emph{infinitely} many qubits\footnote{This follows by taking the limit of \cref{lem:W} in \cref{sec:vectorization} in the $C^*$-algebraic formulation of quantum mechanics. We do not prove the result for the sake of conciseness as it would require introducing technical concepts about $C^*$-algebras and the related formulation of quantum mechanics (see, e.g., \cite[Chapter 2]{naaijkens-2017-quantum-spin-systems-on-infinite}).}.

In particular, if we assume that the number of nodes $n$ is known (as is commonly assumed in the context of the LOCAL model), we can plug in $\epsilon = 1/n^{10}$ and apply the union bound to obtain:
\begin{result}\label{res:main-whp}
    There is a $1$-round one-way quantum-LOCAL algorithm that colors directed cycles with $2\cdot10^4$ colors with high probability (w.h.p.).
\end{result}
The number of colors is rather large, but the standard color-reduction techniques \cite{cole-vishkin-1986-deterministic-coin-tossing-with,naor-stockmeyer-1995-what-can-be-computed-locally,kohonen-korhonen-etal-2017-distributed-colour-reduction} can be used to reduce the number of colors from $k$ to $O(\log\log k)$ in one round; for example, we obtain:
\begin{result}\label{res:3color}
    There is a $4$-round quantum-LOCAL algorithm that $3$-colors directed cycles w.h.p.
\end{result}
As soon as we have this ingredient, we can combine it with the usual classical graph coloring techniques \cite{goldberg-plotkin-shannon-1988-parallel-symmetry,panconesi-rizzi-2001-some-simple-distributed-algorithms,barenboim-elkin-2013-distributed-graph-coloring} to obtain a constant-round $(\Delta+1)$-coloring algorithm when $\Delta = O(1)$. Notably, it immediately implies fast quantum-LOCAL algorithms for many other problems in bounded-degree graphs, including the following:
\begin{result}\label{res:color}
    There are $O_\Delta(1)$-round quantum-LOCAL algorithms that compute
    \begin{itemize}
    \item a vertex coloring with ${\Delta+1}$ colors,
    \item an edge coloring with ${2\Delta-1}$ colors,
    \item a maximal independent set, and
    \item a maximal matching
    \end{itemize} 
    with high probability in graphs of maximum degree $\Delta = O(1)$.
\end{result}
All of the above problems are known to require $\Theta(\log^* n)$ rounds in the classical LOCAL model, with or without randomness \cite{linial-1992-locality-in-distributed-graph-algorithms,naor-1991-a-lower-bound-on-probabilistic-algorithms-for}, and these problems are arguably natural; they have been studied especially in the field of distributed algorithms already since the 1990s. We have hence identified the first \emph{natural} graph problems that provably admit a quantum advantage in the LOCAL model.

\subsection{Corollaries}
\label{sec:into-cor}

\paragraph{LCL Problems.}
There are more far-reaching consequences, beyond some isolated problems. Naor and Stockmeyer \cite{naor-stockmeyer-1995-what-can-be-computed-locally} introduced in their seminal Dijkstra-prize-winning work the notion of \emph{locally checkable labeling problems}, or LCLs for short. These are graph problems that can be specified by listing a \emph{finite set} of valid labelings in local neighborhoods; for any constant maximum degree $\Delta$, the above-mentioned problems such as ${(\Delta+1)}$-coloring, maximal independent set, and maximal matching are examples of LCL problems. LCLs have become the most widely studied problem family in the context of the LOCAL model, and thanks to e.g.\ \cite{
    cole-vishkin-1986-deterministic-coin-tossing-with,
    naor-1991-a-lower-bound-on-probabilistic-algorithms-for,
    linial-1992-locality-in-distributed-graph-algorithms,
    naor-stockmeyer-1995-what-can-be-computed-locally,
    brandt-fischer-etal-2016-a-lower-bound-for-the,
    fischer-ghaffari-2017-sublogarithmic-distributed,
    ghaffari-harris-kuhn-2018-on-derandomizing-local,
    balliu-hirvonen-etal-2018-new-classes-of-distributed,
    chang-pettie-2019-a-time-hierarchy-theorem-for-the,
    chang-kopelowitz-pettie-2019-an-exponential-separation,
    rozhon-ghaffari-2020-polylogarithmic-time-deterministic,
    balliu-brandt-etal-2020-how-much-does-randomness-help,
    balliu-brandt-etal-2021-almost-global-problems-in-the,
    chang-2020-the-complexity-landscape-of-distributed,
    grunau-rozhon-brandt-2022-the-landscape-of-distributed
}, we now have a good understanding of the \emph{complexity landscape} of LCL problems in the LOCAL model \cite{suomela-2020-landscape-of-locality-invited-talk}. What is interesting for us is the lower end of the complexity spectrum. First, there are \emph{gaps}: if we can solve some LCL problem with round complexity $o(\log n)$ in the deterministic LOCAL model or with round complexity $o(\log \log n)$ in the randomized LOCAL model (w.h.p.), we can also solve the same problem in $O(\log^* n)$ rounds in the LOCAL model \cite{chang-kopelowitz-pettie-2019-an-exponential-separation}. Second, the family of problems with round complexity $O(\log^* n)$ is a particularly important family of problems; there are numerous \emph{symmetry-breaking} problems for which $\Theta(\log^* n)$ is the tight round complexity.

Remarkably, Chang, Kopelowitz and Pettie \cite{chang-kopelowitz-pettie-2019-an-exponential-separation} showed that for every LCL $\Pi$ with classical complexity $O(\log^* n)$, there is a constant $k$ such that if we have an oracle that finds a distance-$k$ coloring of the input graph (with a reasonable constant number of colors), then one invocation of the oracle suffices to solve $\Pi$ in $O(1)$ rounds with classical algorithms. But if we now apply \cref{res:color} to the $k$th power of the input graph, we have a quantum algorithm that provides exactly this oracle. It follows that:
\begin{result}\label{res:lcls}
    Any LCL problem that can be solved in $O(\log^* n)$ rounds in the classical LOCAL model can be solved in $O(1)$ rounds in the quantum-LOCAL model, w.h.p.
\end{result}

The landscape of LCL problems has been studied across a wide range of models of distributed computing, ranging from models weaker than quantum-LOCAL (such as classical deterministic LOCAL) to models stronger than quantum-LOCAL (such as the randomized online-LOCAL model), and one of the main gaps in our understanding has been what happens in the $O(\log^* n)$ region in quantum-LOCAL. This is now fully resolved by \cref{res:lcls}. To give some examples of corollaries, combining our work with \cite{akbari-coiteux-roy-etal-2025-online-locality-meets} implies:
\begin{result}\label{res:lcl-classes-rooted-trees}
    In rooted trees, every LCL problem has round complexity $O(1)$ or $\Omega(\log \log \log n)$ in the quantum-LOCAL model.
\end{result}
In particular, the case of unlabeled rooted regular trees is now fully understood; combining our work with \cite{dhar-kujawa-etal-2024-local-problems-in-trees-across-a} implies:
\begin{result}\label{res:lcl-classes-rooted-regular-trees}
    In unlabeled rooted regular trees, each solvable LCL falls in one of these classes:
    \begin{enumerate}
        \item The round complexity is $O(1)$ in all models from classical LOCAL to online-LOCAL.
        \item The round complexity is $\Theta(\log^* n)$ in classical LOCAL but $O(1)$ in all models at least as strong as quantum-LOCAL or SLOCAL.
        \item The round complexity is $\Theta(\log n)$ in all models from classical LOCAL to online-LOCAL.
        \item The round complexity is $\Theta(n^{1/k})$ for some positive integer $k$ in all models from classical LOCAL to online-LOCAL.
    \end{enumerate}
\end{result}
In particular, this gives the first known separation between classical LOCAL and quantum-LOCAL in the setting of LCLs in rooted regular trees, and it also shows that quantum-LOCAL is as strong as online-LOCAL in this setting.

\paragraph{Beyond LCL Problems.}
Even beyond LCL problems, the theory of distributed graph algorithms is full of results where the round complexity of an algorithm has an additive term $O(\log^* n)$; almost always it comes from the need to break symmetry, and in many cases our work shows that it disappears in the quantum-LOCAL model. For example, the recent work \cite{boudier-kuhn-etal-2026-classification-of-local} classified local \emph{optimization} problems in the case of directed cycles, and identified four distinct complexity classes in the classical setting. Notably one of their classes consists of problems for which the round complexity is $\Theta(\log^* n)$ in both deterministic and randomized classical models and the complexity of all these problems in quantum-LOCAL was open; the same class now admits $O(1)$-round quantum algorithms. To give some concrete examples of problems, finding a good approximation of a minimum independent dominating set (minimum maximal independent set) or a good approximation of a minimum maximal matching in cycles takes $\Theta(\log^* n)$ rounds in the deterministic and randomized LOCAL models, but now we can solve them in $O(1)$ rounds in the quantum-LOCAL model.

The following result is particularly helpful when navigating between models of distributed computing; this connects the well-understood SLOCAL model \cite{ghaffari-kuhn-maus-2017-on-the-complexity-of-local} with the quantum-LOCAL model:
\begin{result}\label{res:slocal}
    In bounded-degree graphs, any graph problem that can be solved with $O(1)$ locality in the SLOCAL model can be solved in $O(1)$ rounds in the quantum-LOCAL model, w.h.p.
\end{result}

\paragraph{Beyond Bounded Degrees.}
So far we have discussed the case of $\Delta=O(1)$, but the implications go beyond the bounded-degree case. One prominent example is the task of $3$-coloring rooted trees of arbitrary maximum degree $\Delta$. In the classical LOCAL model, it is well-known that the complexity is $\Theta(\log^* n)$ rounds, independent of $\Delta$ \cite{goldberg-plotkin-shannon-1988-parallel-symmetry,linial-1992-locality-in-distributed-graph-algorithms}. Very recently, \cite{fraigniaud-magniez-ziccardi-2027-no-distributed-quantum} proved a lower bound of $\Omega(\log^* \Delta)$ in the quantum-LOCAL model. Our work implies that this is tight. To see this, we can split the tree into up to $\Delta$ collections of paths, $3$-color them in parallel with \cref{res:3color}, and this way obtain a coloring of the tree with $3^{\Delta}$ colors, and then apply fast color reduction (for simplicity, we assume here that both $n$ and $\Delta$ are known):
\begin{result}\label{res:3col-rooted-tree}
    There is a quantum-LOCAL algorithm that $3$-colors rooted trees of maximum degree $\Delta$ in $O(\log^* \Delta)$ rounds w.h.p., and this is tight.
\end{result}

\subsection{Open Questions}

In this work, we present a quantum algorithm for a very limited setting---coloring of directed cycles---and all corollaries in \crefrange{res:main-whp}{res:3col-rooted-tree} follow through classical reductions and classical post-processing. This puts limits on how far we can get with present techniques. For example, if the classical algorithm for maximal matching takes $O(\Delta + \log^* n)$ rounds in classical LOCAL \cite{panconesi-rizzi-2001-some-simple-distributed-algorithms,balliu-brandt-etal-2021-lower-bounds-for-maximal}, our approach can push it down to $O(\Delta)$ rounds in quantum-LOCAL, but not beyond that. Fundamentally different techniques are needed to improve the round complexity as a function of $\Delta$. In particular, these questions are wide open:
\begin{problem}
    Can we find a ${(\Delta+1)}$-coloring in $o(\sqrt{\Delta})$ rounds in quantum-LOCAL?
\end{problem}
\begin{problem}
    Can we find a maximal matching or a maximal independent set in $o(\Delta)$ rounds in quantum-LOCAL? The only lower bound is $\Omega(\log \Delta / \log \log \Delta)$ from \cite{kuhn-moscibroda-wattenhofer-2016-local-computation,balliu-coupette-etal-2025-new-limits-on-distributed}.
\end{problem}

The number of colors in \cref{res:main-epsilon} is large. By \cite{gur-li-2026-impossibility-of-one-way-one-round-quantum} we know that it has to be more than $4$, but the precise value is unknown:
\begin{problem}
    What is the smallest $c$ such that for any $\epsilon > 0$ we can find a $c$-coloring in directed cycles with one-way one-round quantum algorithms so that for each edge the local failure probability is at most~$\epsilon$?
\end{problem}

\subsection{Finitely Dependent Colorings}

Our work is directly connected with \emph{finitely dependent colorings} \cite{holroyd-liggett-2016-finitely-dependent-coloring,holroyd-hutchcroft-levy-2018-finitely-dependent-cycle,holroyd-hutchcroft-levy-2020-mallows-permutations-and}. We say that a distribution of colorings in a directed cycle or directed infinite path is \emph{$k$-dependent} if, for any two sets of nodes $X$ and $Y$ that are at distance more than $k$ from each other, the color assignment of nodes in $X$ is independent of the color assignment of nodes in $Y$. Moreover, we want the process to be invariant under shifts.

There is a direct connection between quantum algorithms and finitely dependent colorings \cite{akbari-coiteux-roy-etal-2025-online-locality-meets}: if a one-way one-round quantum algorithm finds a $q$-coloring in anonymous directed cycles with probability $1$, then its output is also a $1$-dependent $q$-coloring. We do not aim at probability $1$ here, but we nevertheless first define a $1$-dependent distribution of $q$-colorings, and then use that as a starting point for deriving our quantum algorithm that colors with probability $1-\epsilon$.

$1$-dependent $q$-colorings for any $q \ge 4$ are known \cite{holroyd-liggett-2016-finitely-dependent-coloring,holroyd-hutchcroft-levy-2018-finitely-dependent-cycle,holroyd-hutchcroft-levy-2020-mallows-permutations-and}. However, we suspect that the $1$-dependent $q$-coloring presented in this work is fundamentally different from the families presented in prior work; this leads to the final open question we want to highlight here:
\begin{problem}
    Prove or disprove: our $1$-dependent coloring and the family of $1$-dependent colorings from prior work \cite{holroyd-liggett-2016-finitely-dependent-coloring,holroyd-hutchcroft-levy-2018-finitely-dependent-cycle,holroyd-hutchcroft-levy-2020-mallows-permutations-and} are fundamentally different in the sense that there is no constant-locality randomized reduction between them in either direction.
\end{problem}

\subsection{Structure of This Work}

In \cref{sec:intro-techniques}, we give a high-level overview of our construction.
It has three parts: 
\begin{enumerate}
    \item constructing a 1-dependent distribution of proper $q$-colorings,
    \item vectorizing this construction and approximating it with a square to ensure positivity,
    \item constructing finite POVMs that approximate this vectorized construction.
\end{enumerate}
The steps of this construction respectively are described in detail in \cref{sec:fin-dep,sec:vectorization,sec:finite-dim}.
The definition of \qLOCAL along with some mathematical preliminaries is given in \cref{sec:prelims}.
The proof of \cref{res:main-epsilon} is given in \cref{sec:finite-dim} under \cref{cor:one-round-coloring}. The proofs of \crefrange{res:main-whp}{res:3col-rooted-tree} can be found in \cref{sec:corollaries}. We prove in \cref{appx:lower-bound} that the number of qubits used by our algorithm is optimal (up to constant factors).

\section{Technical Overview}
\label{sec:intro-techniques}

We construct first a 1-dependent distribution over proper $q$-colorings; then, we use this distribution to build a finite one-way one-round \qLOCAL algorithm
based on the energy-based framework introduced by Gur and Li \cite{gur-li-2026-impossibility-of-one-way-one-round-quantum}.
For any $\eps > 0$, the number of qubits used by the \qLOCAL algorithm can be chosen large enough to ensure that an edge is monochromatic with probability at most $\eps$.

\subsection{Infinite-Dimensional Construction}

\paragraph{Step 1: The 1-Dependent Distribution.}
We consider the operators of $\mathrm{L}(H)$ where $H = \ell_2(G)$ for some group $G$ (described later). We write $\ket{g}\in H$ for the function sending $g$ to one and every other element of $G$ to zero. A key operator is $E = \ketbra{e}{e}$, where $e$ is the neutral element of $G$. For each color $a$, we seek to construct an operator $T_a$ such that
\begin{enumerate}[label=(T\arabic*),leftmargin=1.5cm]
    \item\label[part]{intro-T1} $T_a^2 = 0$ for all $a$,
    \item\label[part]{intro-T2} $\sum_{a = 1}^q T_a = E$, and
    \item\label[part]{intro-T3} $\braUket{e}{T_{a_1} T_{a_2} \ldots T_{a_k}}{e} \geq 0$ for every finite word $w = a_1 a_2 \ldots a_k$.
\end{enumerate}
From such operators, we can define the probability of getting a word $w = w_1 \ldots w_k$ as $p(w) = \braUket{e}{T_w}{e}$ where $T_w = T_{w_1} \cdots T_{w_k}$. \ref{intro-T1} implies that the coloring is proper with probability one. The 1-dependence follows from \ref{intro-T2} as 
\[ 
p(w?w') = \sum_a p(waw') = \braUket{e}{T_w E T_{w'}}{e} = \braUket{e}{T_w}{e} \braUket{e}{T_{w'}}{e} = p(w)p(w')
\] 
by definition of $E$. It is far from obvious that such operators should exist; in fact one can easily verify that finite matrices cannot satisfy \ref{intro-T1} and \ref{intro-T2} simultaneously.

To construct those operators, we first introduce \emph{positive} and \emph{negative} colors. Indeed, a color will be a \emph{signed} generator of $G$, i.e., either $+g$ or $-g$. For each color $a$, we define a projection $P_{a}$ such that $P_{a}P_{-a} = 0$. Then we search for some operator $V$ and define $T_a = P_{+a} V P_{-a}$. We obtain \ref{intro-T1} by definition. Finding the correct $V$ directly is challenging; instead we construct a sequence $V_1, V_2, \ldots$ such that $V = \lim_{n\to\infty} V_n$ defines our finitely dependent distribution. The sequence is produced by this simple iterative process:
\[
V_{n+1} = V_n + (2/s)\paren*{ E - \sum_a P_{+a} V_n P_{-a} }
\]
so that at the limit we are guaranteed that $\sum_a P_{+a} V P_{-a} = E$ and thus obtain \ref{intro-T2}. To prove the convergence, we argue that $\norm{E - \sum_a P_{+a} V_n P_{-a}}_F$ decreases geometrically. This geometric decrease also implies that all $V_n$ are close to $V_1 = (2/s)E$, which is very positive. By leveraging this observation, one can show \ref{intro-T3}.

\paragraph{Step 2: Vectorization and Positivity.}
To realize the 1-dependent construction as a quantum algorithm, we need to
translate the operators $T_a$ into the elements of a bipartite POVM. This
requires addressing two obstacles. First, the $T_a$ act
on an infinite-dimensional space $H= \ell_2(G)$, 
whereas the desired POVM acts on a finite-dimensional bipartite space
$\mathbb{C}^N \otimes \mathbb{C}^N$, so we need an algebraic reformulation
compatible with this structure. Second, condition
\ref{intro-T3} guarantees nonnegative probabilities for finite words, but does
not directly ensure that the resulting measurement operators are positive
semidefinite, as required for a POVM.

We therefore first vectorize the operators $T_a$. In the finite
construction, we choose $V_n$ such that $\norm{E - \sum_a P_{+a} V_n P_{-a}}_F$ is
sufficiently small. Since $V_n$ is obtained from $V_0$ after only $n$ iterations
of the recurrence, it has only finitely many nonzero coefficients in the basis
$\ketbra{g}{h}$. We can thus interpret $V_n$ as a polynomial
$V_n'$ in $\bC[G \times G]$, whose indeterminates are pairs of elements of $G$, denoted by $[g,h]$.
The multiplicative structure is not preserved between the two settings:
the first uses matrix multiplication $\ket{g}\bra{h} \cdot \ket{g'}\bra{h'} =
1_{h=g'}\ket{g}\bra{h'}$, whereas the second uses group multiplication
$[g,h] \cdot [g',h'] = [gg',hh']$. Under this vectorization, the operators
$P_{+a} V_n P_{-a}$ become polynomials $K_a V_n' K_a$, which we will later
evaluate in a finite-dimensional unitary representation of $G \times G$
as part of the POVM construction.

Second, to ensure positivity, we replace $V_n'$ by $X^2\approx V_n'$, where
$X$ is some self-adjoint polynomial. Since $K_a$ is also self-adjoint,
\(K_a X^2 K_a = (X K_a)^* (X K_a),\)
so evaluation under any finite-dimensional unitary representation gives a
positive semidefinite matrix. To construct $X$, we use the fact that
$(s/2)V_n' = I + R$ where $R$ has small norm (by the way $V_n$ was chosen).
Hence we can use a Taylor approximation
\[
\sqrt{\frac{s}{2}} X = I + \frac{1}{2}R - \frac{1}{8}R^2 + \frac{1}{16}R^3 - \ldots \approx \sqrt{I + R},
\]
so that $X^2 \approx V_n'$, and then we are ready to approximate $K_a X^2 K_a$ using finite matrices.

\paragraph{Number of Colors.}
We choose $s=10^4$, which ensures convergence of both
$\{V_n\}$ in Step 1 and the Taylor series in Step 2.
Thus, $\sum_a K_a X^2 K_a$ can be made arbitrarily close to the identity
by taking the number of iterations $n$ and the number of Taylor-series terms sufficiently large.
We did not attempt to optimize the constant $10^4$, which can be reduced
with a more careful analysis.
Since each of the $s$ generators gives two signed colors, the total number
of colors is $q=2s=2\cdot 10^4$.

\subsection{Finite-Dimensional Quantum Realization}

\paragraph{Step 3: The Quantum-LOCAL Algorithm.}
The key ingredient in realizing the above infinite construction as a quantum algorithm is the
energy-based characterization of one-way one-round quantum algorithms introduced
in \cite{gur-li-2026-impossibility-of-one-way-one-round-quantum}.
This characterization reduces the problem to constructing a finite-dimensional
POVM with a fixed number of outcomes and arbitrarily small normalized total energy.

Concretely, 
we consider a one-way one-round quantum algorithm
with $q=2s=2\cdot 10^4$ colors, 
where each node $v_i$ acts as follows.
\begin{enumerate}
\item Prepare two quantum registers $\mathsf{A}_i$ and $\mathsf{B}_i$ in the
maximally entangled state
\[
\frac{1}{\sqrt{N}}\sum_{j=1}^N\ket{j}_{\mathsf{A}_i}\otimes\ket{j}_{\mathsf{B}_i}.
\]
\item Send $\mathsf{B}_i$ to its successor $v_{i+1}$, retain $\mathsf{A}_i$,
and receive $\mathsf{B}_{i-1}$ from its predecessor $v_{i-1}$.
\item Apply a $q$-outcome POVM $\{M_a\}$, identical for all nodes, to
$\mathsf{B}_{i-1}\otimes\mathsf{A}_i$ and output the measurement outcome as its
color $c_i\in[q]$.
\end{enumerate}

The problem reduces to,
for any given $\epsilon>0$, choosing a finite dimension $N$ and a POVM $\{M_a\}$ such that
$\Pr[c_i=c_{i+1}]\leq\epsilon$.
We construct the POVM by evaluating the polynomials from Step 2 at finite
matrices. Suppose that we have a sequence of unitary representations
$g\mapsto U_g^{(N)}\in\mathbb{C}^{N\times N}$ satisfying
\emph{asymptotic orthogonality}: for every fixed $g,h\in G$,
\[
\frac1N\Tr\bigl((U_g^{(N)})^\dagger U_h^{(N)}\bigr)
\longrightarrow 1_{g=h}
\quad\text{as}\quad N\to\infty.
\]

We evaluate polynomials in $\bC[G\times G]$ by the map
\[
Z=\sum_{g,h}Z(g,h)[g,h]
\quad\longmapsto\quad
Z^{(N)}:=\sum_{g,h}Z(g,h)
\bigl(U_g^{(N)}\otimes\overline{U_h^{(N)}}\bigr).
\]

This map 
is a unital $*$-homomorphism, which
preserves the algebraic structure of $\C[G\times G]$.
Asymptotic orthogonality also implies that, for every fixed polynomial $Z$,
\begin{equation}\label{eq:norm-conv}
\frac1N\|Z^{(N)}\|_F\longrightarrow\|Z\|_F
\quad\text{as}\quad N\to\infty.
\end{equation}

Applying this evaluation to the polynomials 
$K_a X^2 K_a$
from Step 2 gives $q=2s$ PSD operators
on $\mathbb{C}^N\otimes\mathbb{C}^N$,
\(
\tilde{M}_a:=K_a^{(N)}(X^{(N)})^2K_a^{(N)},
\)
for $a\in\{\pm 1,\ldots,\pm s\}$. 
Moreover, 
since $\sum_a K_aX^2K_a$ is close to the group-algebra identity $[e,e]$
in Frobenius norm,
the convergence in \eqref{eq:norm-conv} implies that 
its evaluation
$T:=\sum_a\tilde{M}_a$ is therefore close to $I_{N^2}$ in normalized Frobenius
norm, when $N$ is sufficiently large.

The operators $\{\tilde{M}_a\}$ are not yet a POVM
since their sum $T$ is not exactly the identity.
To obtain a valid POVM, we normalize these matrices as follows.
If $T$ is invertible, set
\[
M_a:=T^{-1/2}\tilde{M}_aT^{-1/2},
\qquad\text{so that}\qquad
\sum_a M_a=I_{N^2}.
\]
If $T$ is singular, replace $T^{-1/2}$ by $(T^+)^{1/2}$, where $T^+$ is the
pseudoinverse, and add the residual projector $\Pi_{\ker(T)}$ to any one outcome.
Relabeling the outcomes gives 
the POVM $\{M_a\}_{a\in[q]}$ with $q=2s$ outcomes, 
used in the quantum algorithm above.

\paragraph{Error Analysis.}
To bound the local failure probability, 
the energy-based characterization in \cite{gur-li-2026-impossibility-of-one-way-one-round-quantum}
established that 
\[
\Pr[c_i=c_{i+1}]=\frac1{N^3}\sum_{a\in[q]}\Es(M_a),
\qquad
\Es(M):=\Tr\bigl((\Tr_\As M)^\top\Tr_\Bs M\bigr),
\]
where $\Tr_\As$ and $\Tr_\Bs$ are partial traces on 
the first and second registers, respectively.
The quantity $\Es(M)$ is called the \emph{energy} of $M$.\footnote{
    The term energy comes from
    the fact that $\Es(M)$ can be expressed
    as a sum of squared norms in terms of 
    the eigenvectors of $M$; see \Cref{sec:spectral-formula} for a detailed explanation.
}
Thus it suffices to show the POVM has small total energy. 
A PSD operator $M$ has zero energy precisely when
its support lies in a subspace of the form
\[
\mathcal{V}\otimes\overline{\mathcal{V}}^{\perp}
\cong\mathrm{L}(\mathcal{V}^{\perp},\mathcal{V}).
\]
Under the (inverse) vectorization map $\ket{i}\otimes\ket{j}\mapsto\ketbra{i}{j}$, this is a
space of off-diagonal matrices: its matrices map
$\mathcal{V}^{\perp}$ into $\mathcal{V}$ and vanish on $\mathcal{V}$.

For a signed color $a=\pm j$, the matrix $K_a^{(N)}$ is the orthogonal projector
onto $\mathcal{V}_a\otimes\overline{\mathcal{V}_a}^{\perp}$, where
$\mathcal{V}_a$ is the corresponding $\pm1$ eigenspace of $U_{g_j}^{(N)}$.
Then $\tilde{M}_a=K_a^{(N)}(X^{(N)})^2K_a^{(N)}$
is supported on this subspace, so
$\Es(\tilde{M}_a)=0$. Thus the total energy is zero before normalization.
Since $T$ is already close to the identity in normalized Frobenius norm,
normalization introduces only a small total energy
(see \cref{lem:normalize-povm}). By choosing the polynomial approximation in Step 2
sufficiently accurate and then taking $N$ sufficiently large, we can make
$\Pr[c_i=c_{i+1}]$ arbitrarily small,
while keeping the number of colors $q=2\cdot 10^4$ fixed.

\paragraph{Constructing the Unitary Representations.}
We construct the required representations using the asymptotic freeness theorem
for random matrices \cite{collins-sniady-2006-integration-with-respect-to-the}.
It suffices to consider even dimensions $N=2d$.
For each $d$, let $O_1^{(d)},\ldots,O_s^{(d)}\in\mathrm{O}(d)$ be independent
Haar-random orthogonal matrices, and assign to each generator the matrix
\[
g_i\quad\longmapsto\quad U_{g_i}^{(N)}:=
\begin{pmatrix}
0 & (O_i^{(d)})^\top\\
O_i^{(d)} & 0
\end{pmatrix},
\qquad i=1,\ldots,s.
\]

These matrices are self-adjoint unitaries satisfying $(U_{g_i}^{(N)})^2=I_N$,
so the assignment extends to a unitary representation $g\mapsto U_g^{(N)}$
of $G$. Since $(U_g^{(N)})^\dagger U_h^{(N)}=U_{g^{-1}h}^{(N)}$, asymptotic
orthogonality follows from the claim that, for all $g\in G$,
\[
\frac1N\Tr(U_g^{(N)})\longrightarrow 1_{g=e}
\quad\text{as}\quad N\to\infty.
\]

This claim is immediate for $g=e$. For $g\ne e$, write
$g=g_{i_1}\cdots g_{i_k}$ in reduced form, so that $i_j\ne i_{j+1}$.
If $k$ is odd, $U_g^{(N)}$ is block off-diagonal and has trace zero.
If $k$ is even, its two diagonal blocks are products of the $O_i^{(d)}$ and
their transposes, with distinct adjacent indices. Asymptotic freeness implies
that the normalized trace of each block converges to zero almost surely.
Since $G$ is countable, these limits hold simultaneously for all $g\in G$ 
with probability one. Fixing a realization in this event gives the
required deterministic sequence of representations.

We remark that the entire POVM construction uses only real matrices,
as the polynomials in Step 2 have real coefficients and 
the group representations in Step 3 use real orthogonal matrices. 
Thus the resulting algorithm can be implemented with real-valued 
quantum computation.

\section{Preliminaries}
\label{sec:prelims}

\paragraph{Notation.}
We use $\bZ$, $\bN$, $\bR$, and $\bC$ to denote the integers, the integers greater than zero, the reals and the complex numbers respectively. For $n \geq 1$ an integer, we write $[n] = \set{1, 2, \ldots, n}$. 

For a Hilbert space $H$, we use $\ket{x}$ to denote vectors of $H$ and $\bra{x}$ for
the corresponding covector. 
Let $\mathrm{L}(H)$ denote the bounded linear operators on $H$.
Given $X \in \mathrm{L}(H)$, we write 
$X(e_i, e_j) := \bra{ e_i } X \ket{ e_j }$ for the
coefficients of $X\in \mathrm{L}(H)$, where $\{\ket{ e_i }\}_i$ is an orthonormal basis of $H$; we say $X$ is Hilbert-Schmidt if
\[
\norm{X}_F := \sqrt{\sum_{i,j} |X(e_i, e_j)|^2 } < \infty,
\]
where $\norm{X}_F$ is called the Frobenius norm of $X$.
Given $X \in \mathrm{L}(H)$, 
we denote by $X^*$ the adjoint of $X$, and we say $X$ is self-adjoint if $X = X^*$.

Given two finite-dimensional linear spaces $S,T\subseteq \mathbb{C}^N$, we
denote by $\mathrm{L}(S,T)$ the space of linear operators from $S$ to $T$, set
$\mathrm{L}(S) := \mathrm{L}(S,S)$, denote by $S^\perp$ the orthogonal
complement of $S$, and denote by $\Pi_S$ the orthogonal projector onto $S$. 
Given a matrix $X\in\mathbb{C}^{M\times N}$, we denote by 
$\bar{X}$ its conjugate, $X^\top$ its transpose, and $X^\dagger := \bar{X}^\top$ its conjugate transpose.

We refer readers to \cite{davidson-1996-c-algebras-by-example} for background on Hilbert spaces and group algebras.

\paragraph{Quantum Information.} A \emph{pure state} of a quantum register
$\mathsf{A}$ of dimension $d$ is a unit vector $\ket{\psi} \in \bC^d$. If
$\mathsf{A}$ is in state $\ket{\psi_i}$ with probability $p_i$, then it is in a
\emph{mixed state}, described by the density matrix
$\rho = \sum_i p_i \ketbra{\psi_i}{\psi_i}$. In general, a density matrix is a
PSD matrix $\rho\in\mathbb{C}^{d\times d}$ with
$\Tr(\rho) = 1$; by its spectral decomposition, every density matrix can be
viewed as a probability distribution over pure states.

Classical information is extracted from a quantum register by a quantum
measurement. On a register of dimension $d$, a measurement with $q$ outcomes is
described by a \emph{positive operator-valued measure} (POVM), i.e., PSD
matrices $P_1, P_2, \ldots, P_q\in\mathbb{C}^{d\times d}$ such that
$\sum_{i=1}^q P_i = I$. By Born's rule, measuring a register in state $\rho$
yields outcome $i\in [q]$ with probability $\tr{P_i \rho}$. POVMs generalize
projective measurements. For instance, measuring in an orthonormal basis
$\{\ket{e_i}\}_{i=1}^d$ corresponds to the POVM given by
$P_i = \ketbra{e_i}{e_i}$ for $i\in [d]$.

We refer readers to \cite{watrous-2018-the-theory-of-quantum-information} for details on the operator formalism of quantum information.

\paragraph{Models of Distributed Computing.}

Our key results are positive results in the quantum-LOCAL model, but e.g.\ to understand \cref{res:lcls} we also need to explain what we mean by the classical LOCAL model. One slightly annoying detail is that there is no single canonical definition of the LOCAL model, but fortunately for our purposes most of the details are not important. The only crucial assumption that we need to make is that the nodes know the value of $n$, which we will use to denote an \emph{upper bound} on the size of the network. For concreteness, we will give here one definition of the classical deterministic LOCAL model, and one definition of the quantum-LOCAL model. Our results can be directly generalized to other model definitions, but naturally if we e.g.\ strengthen the classical model we need to strengthen the quantum model in an analogous manner if we want to prove a result analogous to \cref{res:lcls}.

\paragraph{Classical Deterministic LOCAL Model.}

When we briefly refer to the LOCAL model in this work, we specifically mean the classical \emph{deterministic} LOCAL model \cite{linial-1992-locality-in-distributed-graph-algorithms,peleg-2000-distributed-computing-a-locality-sensitive}. This is a message-passing model for computer networks. In this model an algorithm $A$ describes computation from the perspective of a single node: how the node chooses its initial state, how it constructs messages that it sends to its neighbors, how it updates its local state based on the messages it receives, and how it determines its local output.

Computation proceeds as follows. The adversary first chooses an input graph $G = (V,E)$ and a value $n$ such that $n \ge |V|$. The adversary also chooses an assignment of \emph{unique identifiers}; there is a known constant $C$ such that each node is labeled with a distinct value from $\{1,2,\dotsc,n^C\}$. Then we apply algorithm $A$ to choose the initial states of all nodes; the initial state of a node $v$ is a function of the degree of $v$, its unique identifier, the value $n$, and possible local inputs associated with node $v$, if any.

After initialization, computation proceeds in synchronous rounds. If a node $v$ has degree $d$, then it first uses algorithm $A$ to construct a $d$-element tuple of outgoing messages, one per neighbor (ordered by \emph{port numbers}), and it then receives a $d$-element tuple of incoming messages (again ordered in the same consistent manner), and finally it uses algorithm $A$ to update its own state based on the messages it received. We will here assume that the set of possible states and the set of possible messages are finite, but the size of the set can depend on~$n$.

After some number of rounds $T(n)$, each node stops and announces its \emph{local output}, which is a function of its final state. Now we say that algorithm $A$ solves a graph problem $\Pi$ in $T(n)$ rounds if in any input graph and for any choice of unique identifiers and for any valid choice of $n$, the algorithm announces the local outputs after $T(n)$ rounds, and these outputs form a valid solution to $\Pi$. For example, if the task at hand is graph coloring, then the local output is the color of the node.

\paragraph{Quantum-LOCAL Model.}

The \emph{quantum-LOCAL model} \cite{gavoille-kosowski-markiewicz-2009-what-can-be-observed,arfaoui-fraigniaud-2014-what-can-be-computed-without} is, in brief, the natural quantum analog of the classical LOCAL model. Each node holds a finite number of finite-dimensional quantum registers (the number of registers and their dimensions can depend on~$n$). The algorithm $A$ determines how each node locally initializes its own registers. The whole network is initially in a product state; we emphasize that the nodes do \emph{not} share any entanglement before computation starts.

For communication, nodes can pass some quantum registers that they hold to their neighbors, and after receiving registers from their neighbors, they can perform arbitrary quantum operations. Finally, after $T(n)$ rounds, each node produces its classical local output as a result of measurement.

We will focus on Monte Carlo algorithms: we say that $A$ solves problem $\Pi$ in $T(n)$ rounds with probability $p$ if the following holds: all nodes stop and announce their local outputs after $T(n)$ rounds, and the probability that the local outputs form a valid solution of $\Pi$ is at least $p$. Usually we are interested in algorithms that are correct \emph{with high probability}, that is, we can choose any constant $C$ and achieve global success probability $p = 1-1/n^C$.

\paragraph{One-Way One-Round Anonymous Quantum Algorithms in Cycles.} Consider a
directed cycle with $n$ nodes $(v_i)_{i\in\mathbb{Z}_n}$. In a \emph{one-way
one-round anonymous quantum algorithm}, all nodes are initially identical and
execute the same algorithm. In the single communication round, each node $v_i$
sends a quantum message to its successor $v_{i+1}$ and receives one from its
predecessor $v_{i-1}$. The local computation and message length are unbounded.
In the high-probability setting, these algorithms are equivalent to one-way
one-round quantum-LOCAL algorithms, since each node can independently sample an
identifier uniformly from a sufficiently large set $[n^C]$ such that the
probability of the sampled identifiers being pairwise distinct is at least
$1-1/n^{C-2}$. 

\paragraph{Other models.}

We will give only brief, informal descriptions of other models that we mention in this work and cite the relevant papers for more detailed discussion.

In the \emph{deterministic SLOCAL model} \cite{ghaffari-kuhn-maus-2017-on-the-complexity-of-local}, computation proceeds in a sequential manner. The adversary chooses both the input graph $G=(V,E)$ and a permutation of nodes $\pi$. The adversary presents nodes one by one, following $\pi$. When a node $v$ is presented, the algorithm gets to see the radius-$T(n)$ neighborhood of $v$, including any previously assigned labels there, and using only this information, the algorithm has to produce the label of $v$, where the label is a pair consisting of the local output of $v$ and whatever local state the algorithm wants to store for itself.

In the \emph{deterministic online-LOCAL model} \cite{akbari-eslami-etal-2023-locality-in-online-dynamic}, the algorithm has also access to a global memory. Put otherwise, the algorithm sees at each point all nodes presented so far, and their radius-$T(n)$ neighborhoods. This is a generalization of the usual setting of online graph algorithms.

In the \emph{randomized online-LOCAL model} \cite{akbari-coiteux-roy-etal-2025-online-locality-meets} the algorithm has access to a source of random bits. We emphasize that the adversary is oblivious, i.e., it fixes the input and the permutation before the algorithm starts to flip coins. What is good to know is that randomized online-LOCAL is strong enough to simulate quantum-LOCAL \cite{akbari-coiteux-roy-etal-2025-online-locality-meets}.

\section{Infinite-Dimensional Construction}
\label{sec:inf}

In this section, we construct a 1-dependent distribution over the proper $q$-colorings of $\bZ$. The essential outcome of this section is not this distribution itself but rather its vectorized version. The key lemma that we prove here is the following (see \cref{def:group-alg} for a definition of $\bC[G \times G]$).

\begin{tbox}
\begin{lemma}
    \label{lem:W}
    For $s \geq 10^4$ and $n \geq 1$, there exists a self-adjoint $W \in \bC[G \times G]$ such that
    \[
    \norm*{[e,e] - \sum_{i} \left(K_{+i} W^2 K_{+i} + K_{-i} W^2 K_{-i}\right)}_F
    \leq (2s+1)\lambda^n \ ,
    \]
    where $K_{\pm i} = ([e,e] \pm [i,e])([e,e] \mp [e,i])/4$ is self-adjoint for all $i \in [s]$ and $\lambda$ is a real number smaller than one.
\end{lemma}
\end{tbox}

\subsection{Constructing a 1-Dependent Distribution}
\label{sec:fin-dep}

This section covers the construction of a 1-dependent distribution over proper $q$-colorings of paths. We express this in the language of operators:
we construct a Hilbert space $H$ with a distinguished vector $\ket{e}\in H$ and operators $T_1, T_2, \ldots, T_q$ over $H$, one for each color, such that
\begin{enumerate}[label=(T\arabic*),leftmargin=1.5cm,topsep=1em]
\item\label[part]{T1} $T_a^2 = 0$ for all $a$,
\item\label[part]{T2} $\sum_{a = 1}^q T_a = \ketbra{e}{e}$, and
\item\label[part]{T3} $\braUket{e}{ T_{a_1} T_{a_2} \ldots T_{a_k} }{e}  \geq 0$ for every finite word $w = a_1 a_2 \ldots a_k$.
\end{enumerate}
The fact that 1-dependent distributions over proper $q$-colorings of $\bZ$ exist is already known \cite{holroyd-liggett-2016-finitely-dependent-coloring}. Hence, we do not merely seek operators as in \ref{T1}--\ref{T3}; we seek a distribution that can be approximated well by \qLOCAL algorithms. In this section, we describe such a 1-dependent distribution.

\paragraph{Letters and Colors.}
Throughout, we call the number of colors $q$. We choose $q = 2s$ even and see colors as signed, i.e., the colors are $-s, \ldots, -2,-1,+1,+2, \ldots, +s$.
Let 
\[
G = \langle g_1, g_2, \ldots, g_s \vert g_i^2 = e\text{ for all }i\in[s]\rangle
\]
be the group with $s$ generators, all of which are involutions. To ease the notation, we will write $i$ instead of $g_i$. The elements of $G$ are words with (unsigned) letters $1, 2, \ldots, s$ such that we can remove any pair of identical consecutive letters. The empty word and neutral element of $G$ is denoted by $e$.

\paragraph{The Orthogonal Projections.}
Now we consider the Hilbert space $H = \ell_2(G) = \set{ f: G \to \bC : \sum_{g\in G} |f(g)|^2 < \infty}$ and, for all $g\in G$, let us write $\ket{g}$ for the function that maps $g$ to one and every other $h \neq g$ to zero. We define two operators on $H$ for all $g\in G$ as 
\[
L_g\ket{w} = \ket{gw} \ ,\quad\text{and}\quad
R_g\ket{w} = \ket{wg} \ .
\]
For each letter $i$, we define two orthogonal projections on orthogonal subspaces
\[
P_{+i} = (I + L_i)(I + R_i) / 4
\quad\text{and}\quad
P_{-i} = (I - L_i)(I - R_i) / 4 \ .
\]

\begin{lemma}
    \label{lem:orth-projections}
    For all $i\in [s]$, the matrices $P_{+i}$ and $P_{-i}$ are orthogonal projections and $P_{+i} P_{-i} = 0$.
\end{lemma}
\begin{proof}
    Note that $L_i^2 = I$, $R_i^2 = I$ for all $i$. Since the $L_i$ and $R_i$ permute the basis $\ket{g}$ of $H$, they are unitary matrices, hence $L_i^* = (L_i)^{-1} = L_i$ and $R_i^* = (R_i)^{-1} = R_i$. In particular, the $P_{+i}$ and $P_{-i}$ are also self-adjoint. Using that $L_iR_i = R_i L_i$, simple calculation shows that $P_{+i}^2 = P_{+i}$ and $P_{-i}^2 = P_{-i}$. Finally, since $L_i$, $R_i$ and $I$ commute, we have that
    \[
    \begin{aligned}
    P_{+i}P_{-i}
    &= \frac{1}{16}(I+L_i)(I+R_i)(I-L_i)(I-R_i) \\
    &= \frac{1}{16}(I-L_i^2)(I - R_i^2)
    = 0 \, . \qedhere
    \end{aligned}
    \]
\end{proof}

\paragraph{The Iterated Construction.}
Now we need to find an operator $V$ such that the operator $T_{a} = P_a V P_{-a}$ satisfies \ref{T2} and \ref{T3}. For the resulting distribution to be approximated by \qLOCAL algorithms, we further need that it can be approximated by operators that have finitely many non-zero coefficients in the $\ketbra{g}{h}$ basis.
To do so, we build a sequence $V_0, V_1, V_2, \ldots$ iteratively by starting from $V_0 = 0$ and using the following definition
\begin{equation}
    \label{eq:V-rec}
    V_{n+1} = V_n + \frac{2}{s}(E - S(V_n))
    \quad\text{where}\quad
    S(X) = \sum_{i = 1}^s \left(P_{+i} X P_{-i} + P_{-i} X P_{+i} \right).
\end{equation}
As we show in \cref{thm:fin-dep-V}, the sequence of $V_n$ converges to some operator $V$ such that $S(V) = E$. As such the operators $T_a = P_a V P_{-a}$ satisfy \ref{T1} and \ref{T2}. Positivity (i.e., \ref{T3}) still needs to be proven, but as it is not strictly necessary to obtain \cref{res:color}, we defer the proof to \cref{app:proof-1-dep}. The important properties of this construction are summarized in \cref{thm:fin-dep-V}, which we prove next.

\begin{tbox}
\begin{lemma}
    \label{thm:fin-dep-V}
    Suppose $s \geq 10^4$.
    For every $n \geq 1$, we have that
    \begin{enumerate}[label=(V\arabic*),leftmargin=1.5cm]
        \item\label[part]{V1} $V_n \in \fS = \operatorname{span}\set{ \ketbra{g}{h}: g, h\in G}$, and
        \item\label[part]{V2} $\norm{ E - S(V_n) }_F \leq \lambda^n$ for some $\lambda < 1$.
    \end{enumerate}
\end{lemma}
\end{tbox}

\paragraph{The Weighted Norm.}
For the positivity in \cref{sec:vec}, bounds on the Frobenius norm do not suffice. The reason is that we see $V_n$ in the group algebra $\bC[G \times G]$ which has a different multiplicative structure for which the Frobenius norm is not submultiplicative, i.e., there is no universal constant $\alpha$ for which $\norm{ X \cdot Y }_F \leq \alpha \norm{ X }_F \norm{ Y }_F$ where $X \cdot Y$ is the product in $\bC[G \times G]$. To overcome this issue, we work with a norm that dominates the Frobenius norm and for which $\norm{ X \cdot Y }_w \leq 4 \norm{ X }_w \norm{ Y }_w$ (see \cref{lem:norm-submul}).

\begin{proposition}
    \label{prop:weighted-norm}
    For $X \in \fS$ let 
    \[
    \norm{X}_w = \sum_{k,\ell \geq 0} 2^{(k + \ell)/2} \norm{ X_{k,\ell} }_F
    \quad\text{where}\quad
    X_{k,\ell} = \sum_{\stackrel{g,h\in G:}{|g|=k, |h| = \ell}} X(g,h) \ketbra{g}{h}
    \]
    We have that $\norm{ \cdot }_w$ is a norm on $\fS$ and $\norm{ X }_F \leq \norm{ X }_w$ for all $X \in \fS$.
\end{proposition}
\begin{proof}
    It is easy to verify that $\norm{ \cdot }_w$ is a norm. For all $X$, we have
    \[ 
    \norm{X}_F^2 = \sum_{k,\ell} \norm{X_{k,\ell}}_F^2 \leq \sum_{k,\ell} 2^{k + \ell} \norm{X_{k,\ell}}_F^2 \leq \norm{X}_w^2
    \] 
    using that $\paren*{ \sum_{k,\ell} a_{k,\ell} }^2 \geq \sum_{k,\ell} a_{k,\ell}^2$ with $a_{k,\ell} = 2^{(k + \ell)/2}\norm{X_{k,\ell}}_F \geq 0$.
\end{proof}
    
\paragraph{The Contraction Lemma.}
\cref{lem:contraction} is the key technical idea behind \cref{thm:fin-dep-V}. It implies that each time that we apply \cref{eq:V-rec}, we contract geometrically the norm of $R_n = E - S(V_n)$. We quickly prove \cref{thm:fin-dep-V} and the remainder of this section is dedicated to proving \cref{lem:contraction}.

\begin{lemma}[Contraction Lemma]
\label{lem:contraction}
For all $X \in \fS$ we have that
\[
    \norm*{X - \frac{2}{s}S(X)}_w \leq \lambda \norm{X}_w 
    \quad\text{where}\quad
    \lambda = \frac{3}{4} + \frac{24}{\sqrt{2s}}
    \]
\end{lemma}

\begin{proof}[Proof of \cref{thm:fin-dep-V}]
    We prove \ref{V1} by induction on $n$. Indeed, we have $V_1 = (2/s)E$ hence it has only one non-zero coefficient. For all $\ket{g}\in H$, we have that $P_{\pm i} \ket{g} = (1/4)(\ket{g} + \ket{igi} \pm \ket{ig} \pm \ket{gi})$, hence $V_{n+1} = V_n + (2/s)(E - S(V_n))$ also has only finitely many coefficients.
    
    To prove \ref{V2}, define $R_n = E - S(V_n)$ and observe that for all $n \geq 0$ we have
    \[
    \norm{R_{n+1}}_w = \norm{E - S(V_{n+1})}_w = \norm{R_n - (2/s)S(R_n)}_w \leq \lambda\norm{R_n}_w \leq \lambda^{n+1}
    \]
    where the second equality follows from \cref{eq:V-rec}, the first inequality follows from \cref{lem:contraction}, and the last uses that $\norm{R_0}_w = \norm{E}_w = 1$. As $\norm{R_n}_F \leq \norm{R_n}_w$ (by \cref{prop:weighted-norm}), we get \cref{V2}.
\end{proof}

\paragraph{Decomposing $S(X)$.}
Here, it is useful to think of $S$ as an infinite matrix with rows and columns indexed by $G^2$. We bound the contribution of the diagonal terms and off-diagonal terms of $S(X)$ separately. The norm of the diagonal term $D(X)$ is easy to bound while the norm of the off-diagonal terms $F_i(X)$ requires more work.
\begin{lemma}
    \label{lem:decompose-S}
    We have $S(X) = D(X) + (1/8)\sum_{j=1}^3 F_j(X) - (1/8)\sum_{j=4}^7 F_j(X)$ where 
    \[
    \text{for all }X\in\fS\ , \quad
    \norm{X - (2/s)D(X)}_w \leq (3/4)\norm{X}_w
    \]
    and
    \begin{align*}
    F_1(\ketbra{g}{h}) &= \sum_{i:igi\ne g}\ketbra{igi}{h},
    & F_2(\ketbra{g}{h}) &= \sum_{i:ihi\ne h}\ketbra{g}{ihi},
    & F_3(\ketbra{g}{h}) &= \sum_{i:(igi,ihi)\ne(g,h)}\ketbra{igi}{ihi},\\
    F_4(\ketbra{g}{h}) &= \sum_i\ketbra{ig}{ih},
    & F_5(\ketbra{g}{h}) &= \sum_i\ketbra{ig}{hi},
    & F_6(\ketbra{g}{h}) &= \sum_i\ketbra{gi}{ih},\\
    F_7(\ketbra{g}{h}) &= \sum_i\ketbra{gi}{hi}.
  \end{align*}
\end{lemma}
\begin{proof}
    Fix $g, h \in G$ and consider $S(\ketbra{g}{h})$. For $i \in [s]$, define $\ket{ x_i(g) } = \ket{g} + \ket{igi}$ and $\ket{ y_i(g) } = \ket{ig} + \ket{gi}$ so that $P_{\pm i}\ket{g} = (\ket{x_i(g)} \pm \ket{y_i(g)})/4$. Simple calculation shows
    \[
    P_{+i}\ketbra{g}{h}P_{-i} + P_{-i}\ketbra{g}{h}P_{+i} = \frac{\ketbra{x_i(g)}{x_i(h)} - \ketbra{y_i(g)}{y_i(h)}}{8} \ .
    \]
    When one expands this, there is a $\ketbra{g}{h}$ which does not depend on $i$ and seven other terms: 
    \begin{align*}
    S(\ketbra{g}{h}) 
    =&\quad (s/8)\ketbra{g}{h} \\
    &+ (1/8)\sum_{i \in [s]} \ketbra{igi}{h} + \ketbra{g}{ihi} + \ketbra{igi}{ihi} \\
    &- (1/8)\sum_{i \in [s]} \ketbra{gi}{hi} + \ketbra{gi}{ih} + \ketbra{ig}{hi} + \ketbra{ig}{ih}
    \end{align*}
    Note that each of the three positive terms $\ketbra{igi}{h} + \ketbra{g}{ihi} + \ketbra{igi}{ihi}$ can result in $\ketbra{g}{h}$ for the correct values of $i$. We aggregate those in the diagonal operator $D(X)$, i.e., we set
    \[
    D(\ketbra{g}{h}) = d_{g,h}\ketbra{g}{h}
    \quad\text{where}\quad 
    d_{g,h} = \frac{s}8 + \frac18 \sum_i 1( igi = g ) + 1( ihi = h ) + 1( igi = g \wedge ihi = h ) \ .
    \]
    The terms of $\ketbra{y_i(g)}{y_i(h)}$ either increase or reduce the length of the word by one, hence cannot result in $\ketbra{g}{h}$. We distribute the seven remaining terms in the operators $F_1, \ldots, F_7$. Note that the sums in $F_1$, $F_2$ and $F_3$ exclude the $i$'s for which we get $\ketbra{g}{h}$, as those are already accounted for in $D(X)$.

    To bound $\norm{X - (2/s)D(X)}_w$, observe that each term in the sum contributes at most $1/8$ and there are at most $3s$ many. Hence $d_{g,h} \in [s/8, s/2]$ and therefore for all $k,\ell \geq 0$,
    \[
    \norm{ (X - (2/s)D(X))_{k,\ell} }_F^2 
    \leq \sum_{g,h: |g|=k, |h|=\ell} |1 - (2/s)d_{g,h}|^2 |X(g,h)|^2 \leq (3/4)^2\norm{X_{k,\ell}}^2_F
    \]
    and thus
    \begin{align*}
    \norm{ X - (2/s)D(X) }_w 
    &= \sum_{k,\ell \geq 0} 2^{(k + \ell)/2} \norm{ (X - (2/s)D(X))_{k,\ell} }_F\\
    &\leq (3/4)\sum_{k,\ell \geq 0} 2^{(k + \ell)/2} \norm{ X_{k, \ell} }_F
    = (3/4)\norm{X}_w \ . \qedhere
    \end{align*}
\end{proof}

\paragraph{Bounding the Off-Diagonal Terms.}
Next, we prove that each of the Frobenius norms of off-diagonal terms $F_i(X)$ increases only by a factor $O(\sqrt{s})$ compared to that of $X$. The proof is identical for every $F_j$ and relies on \cref{lem:row-col-bound}. Again, see $F_j$ as an infinite matrix whose rows and columns are indexed by $\Gamma = G \times G$. The idea is that, by ordering the elements of $\Gamma$ by length (and breaking ties arbitrarily), we can decompose each $F_j$ into two matrices whose rows and columns sum to at most $s$ or $2$. The reason we can order $\Gamma$ in such a way is because each $F_i$ is defined by involutions $\phi_{j, 1}, \phi_{j, 2}, \ldots, \phi_{j, s}$ on $\Gamma$. In \cref{lem:row-col-bound}, we provide a generic way of bounding each $\norm{F_j(X)}_F$ by analyzing the involutions corresponding to $F_j$. As this is a relatively standard proof, we defer it to the appendix.

\begin{restatable}{proposition}{PropRowColBound}
    \label{lem:row-col-bound}
    Let $\Gamma = G \times G$ and $\phi_1, \phi_2, \ldots, \phi_s : \Gamma \to \Gamma$ be such that $\phi_i(\phi_i(g)) = g$ for all $g\in \Gamma$ and $i\in [s]$. Let $\prec$ be the total order on $\Gamma$ and for all $x\in \Gamma$ define
    \[
    d^-(x) = \card{\set{ i : \phi_i(x) \prec x }}
    \quad\text{and}\quad
    d^+(x) = \card{\set{ i : x \prec \phi_i(x) }} \ .
    \]
    Then the map $Y : \fS \to \fS$ defined by
    \[ 
    Y(\ketbra{g}{h}) = \sum_{i=1}^s 1( \phi_i(g,h) \neq (g,h) ) e_{\phi_i(g,h)}
    \quad\text{where}\quad e_{g,h} = \ketbra{g}{h} \ .
    \] 
    is such that
    $\norm*{Y(X)}_F \leq 2\sqrt{ \sup_{x\in \Gamma} d^-(x) } \sqrt{ \sup_{x\in \Gamma} d^+(x) }\norm{X}_F$ for all $X \in \fS$.
\end{restatable}

We now use \cref{lem:row-col-bound} to prove bounds on the \emph{weighted} norm of each $F_j(X)$. Note that \cref{lem:row-col-bound} uses the Frobenius norm, and it does not hold in general for the weighted norm. In \cref{lem:bound-norm-F}, we use the fact that our involutions change the length of $|g| + |h|$ by at most two when $j \neq 3$, and by at most four when $j = 3$, and only lose a constant factor compared to \cref{lem:row-col-bound} because of the weighted norm.

\begin{lemma}
    \label{lem:bound-norm-F}
    We have $\norm{F_j(X)}_w \leq 6\sqrt{2s}\norm{X}_w$ for every $j \neq 3$ and $\norm{F_3(X)}_w \leq 12\sqrt{2s}\norm{X}_w$.
\end{lemma}
\begin{proof}
    Let us first describe the involutions $\phi_{j,1}, \ldots, \phi_{j,s} : \Gamma \to \Gamma$ corresponding to each of the $F_j$: for each $i\in [s]$, define
    \begin{align*}
        \phi_{1,i}(g, h) &= (igi, h) \ ,\quad
        \phi_{2,i}(g, h) = (g, ihi) \ ,\quad
        \phi_{3,i}(g, h) = (igi, ihi) \ ,\\
        \phi_{4,i}(g, h) &= (ig, ih) \ ,\quad
        \phi_{5,i}(g, h) = (ig, hi) \ ,\quad
        \phi_{6,i}(g, h) = (gi, ih) \ ,\quad\text{and}\quad\\
        \phi_{7,i}(g, h) &= (gi, hi) 
    \end{align*}
    such that $F_j(\ketbra{g}{h}) = \sum_{i: \phi_{j,i}(g,h) \neq (g,h)} e_{\phi_{j,i}(g,h)}$ for every $j\in[7]$. Note that we need bounds using the weighted norm while \cref{lem:row-col-bound} uses the Frobenius norm.
    Using the triangle inequality, we get $\norm{ F_j(X) }_w \leq \sum_{k,\ell\geq 0} \norm{F_j(X_{k,\ell})}_w$. For all $j$, we have
    \begin{align*}
        \norm{ F_j(X_{k,\ell}) }_w 
        = \sum_{k', \ell' \geq 0} 2^{(k' + \ell')/2} \norm{ F_j(X_{k,\ell})_{k', \ell'} }_F
        \leq \sum_{(a,b) \in \Delta_j} 2^{(k + a + \ell + b)/2} \norm{ F_j(X_{k,\ell})_{k + a,\ell + b} }_F
    \end{align*}
    where $\Delta_j$ is the set of pairs $(a,b)$ such that there exists $i\in [s]$ for which $\phi_{j,i}(g,h) = (g', h')$ with $|g'|-|g|=a$ and $|h'| - |h| = b$. In words, they are all the pairs of length change that some $\phi_{j,i}$ can cause. For completeness, we interpret $F_j(X_{k,\ell})_{k + a,\ell + b}$ as zero when the indices are negative.
    Using Cauchy--Schwarz, we get that
    \[
        \norm{ F_j(X_{k,\ell}) }_w
        \leq 2^{(k+\ell)/2} \paren*{ \sum_{(a,b) \in \Delta_j} 2^{a + b} }^{1/2} \paren*{ \sum_{a,b = -2}^2 \norm{ F_j(X_{k,\ell})_{k + a,\ell + b} }_F^2 }^{1/2}
    \]
    We bound the sum over $2^{a+b}$ first. Note that $\Delta_1 = \set{ (+2, 0), (0,0), (-2,0)}$, $\Delta_2$ is the symmetric, $\Delta_3 = \set{ (a,b): a,b\in \set{-2,0,+2}}$, and $\Delta_j = \set{ (a,b): a,b\in \set{+1,-1}}$ for $j = 4,5,6,7$. One can then easily verify that $\paren*{ \sum_{a,b \in \Delta_j} 2^{a + b} }^{1/2}$ is less than $3$ for $j \neq 3$, less than $6$ for $j=3$.

    For the second sum, observe that each term $\norm{ F_j(X_{k,\ell})_{k + a,\ell + b} }_F^2$, when expanded, corresponds to different coefficients of $F_j(X_{k,\ell})$, hence that $\sum_{(a,b)\in\Delta_j} \norm{ F_j(X_{k,\ell})_{k + a,\ell + b} }_F^2 = \norm{F_j(X_{k,\ell})}_F^2$. Overall, this means that for $j \neq 3$
    \[
        \norm{ F_j(X) }_w 
        \leq \sum_{k,\ell\geq 0} \norm{F_j(X_{k,\ell})}_w
        \leq 3\sum_{k,\ell\geq 0} 2^{(k+\ell)/2} \norm{F_j(X_{k,\ell})}_F \ .
    \]
    As \cref{lem:row-col-bound} implies that $\norm{F(X_{k,\ell})}_F \leq 2\sqrt{d^-_j d^+_j} \norm{X_{k,\ell}}_F$, it follows that for $j \neq 3$
    \[
        \norm{ F_j(X_{k,\ell}) }_w 
        \leq 6\sqrt{d^-_j d^+_j} \norm{ X_{k,\ell} }_w \ ,
    \]
    where $d^\pm_j = \sup_{g,h\in G} d_j^\pm(g,h)$ where $d_j^\pm(g,h)$ is defined as in \cref{lem:row-col-bound} with respect to $F_j$. The total order $\prec$ on $G \times G$ is defined such that $(\ketbra{g}{h}) \prec (g',h')$ if $|g| + |h| < |g'| + |h'|$ and ties are broken arbitrarily. For $j = 3$ the same reasoning applies after replacing $3$ by $6$ and we get that $\norm{ F_3(X_{k,\ell}) }_w 
        \leq 12\sqrt{d^-_3 d^+_3} \norm{ X_{k,\ell} }_w$.
    We argue that $d_j^- d_j^+ \leq 2s$ for each $j\in[7]$, which implies the claimed bound.
    
    Let us begin with $j = 1$. Hence, consider $d_1^-(g, h) = |\set{ i: \phi_{1,i}(g,h) \prec (g,h) }| = |\set{ i: igi \prec g }|$. If $g = e$ then $igi = e = g$ so $d^-(g,h) = 0$. If $|g| = 1$ either $g=i$ and $igi = iii = i = g$ or $g \neq i$ and $|igi| > |g|$; either way $d(g,h) = 0$. Finally, if $|g|\geq2$ then $g = xwy$ for some word $w$ (possibly empty). To have $|igi| \leq |g|$ we need $i = x$ or $i=y$. So there are at most two choices of $i$ that work; hence $d(g) \leq 2$ for all such $g$. On the other hand, we trivially have that $d_1^+(g,h) \leq s$. The case of $d_2^-d_2^+ \leq 2s$ is entirely symmetric.

    For $j = 3$, if $g = e$ or $h = e$, the same argument as for $j=1$ applies. Suppose $g, h \neq e$. If $i\in[s]$ is such that $(igi, ihi) \prec (g,h)$, then we have $|igi| + |ihi| \leq |g| + |h|$. Note that $|igi| \leq |g|+2$ and that the difference between both sides of the inequality is always even. Call $a_i, b_i \geq 0$ the numbers such that $|igi| = |g| + 2 - 2a_i$ and $|ihi| = |h| + 2 - 2b_i$. We have that $\sum_i a_i \leq 2$ and $\sum_i b_i \leq 2$ because $a_i \neq 0$ (resp.\ $b_i \neq 0$) only if the first or last letter of $g$ (resp.\ $h$) is $i$. The inequality then becomes $|igi| + |ihi| = |g| + |h| + 4 - 2(a_i + b_i) \leq |g| + |h|$. Summing over all $i$ for which $(igi, ihi) \prec (g,h)$, we get that $4d^-_3(g,h) \leq 2\sum_i (a_i + b_i) \leq 8$; hence, $d^{-}_3(g,h) \leq 2$.

    Consider now the case of $j=4$. We have $d_4^-(g,h) = |\set{i: (ig, ih) \prec (g,h)}|$. If $(ig, ih) \prec (g,h)$, then we must have $|ig| + |ih| \leq |g| + |h|$. If $g = h = e$, there is no such $i$. If exactly one of the two is $e$, say $g = e$ and $h \neq e$, the only letter $i$ for which $|ig| + |ih| = 1 + |ih| \leq |h|$ is the leftmost letter of $h$. If $g,h \neq e$, the only letters $i$ for which $|ig| + |ih| \leq |g| + |h|$ must be the leftmost letter of either $g$ or $h$. Hence $d_4^- \leq 2$ and trivially $d^+_4 \leq s$. The cases $j=5,6,7$ are symmetric where the letters that matter are some combinations of leftmost and rightmost letters of $g$ and $h$.
\end{proof}

\paragraph{Concluding the Proof of \cref{lem:contraction}.}
By combining the bound on $D(X)$ from \cref{lem:decompose-S} and the bounds on the norms of $F_j(X)$ from \cref{lem:bound-norm-F}, we get that
\begin{align*}
\norm{X - (2/s)S(X)}_w
&\leq \norm{X - (2/s)D(X)}_w + \frac{1}{4s}\sum_{j=1}^7 \norm{F_j(X)}_w \\
&\leq \paren*{ 3/4 + \frac{(6 \times 6 + 12)\sqrt{2s}}{4s} }\norm{X}_w 
= \lambda \norm{X}_w
\end{align*}
where the last equality is by definition of $\lambda$. Note that for $s \geq 10^4$, we have $\lambda < 1$.

\subsection{Vectorization and Positivity}
\label{sec:vectorization}
\label{sec:vec}

In this section, we consider the $V_n$ constructed in \cref{sec:fin-dep} and reformulate it differently so that (1) it can easily be approximated using finite matrices and (2) we ensure positivity of the resulting POVMs by construction.

\paragraph{Vectorization.} By \cref{thm:fin-dep-V}, we know that $V_n$ is in the span of the $\ketbra{g}{h}$ for $g,h\in G$. In particular, it is a \emph{finite} linear combination of these rank-one operators. This means that we can interpret $V_n \in \fS$ as a polynomial whose indeterminates are elements of $G \times G$. We emphasize that the multiplication of those indeterminates respects the group structure of $G \times G$. The purpose of this step is to provide an algebraic object whose multiplicative structure corresponds to that of $\mathrm{L}(\bC^N \otimes \bC^N)$. However, as we change the multiplicative structure, property \ref{T3} does not transfer automatically.

\begin{definition}
    \label{def:group-alg}
    For a group $G$, we denote by $\bC[G \times G]$ the set of polynomials whose coefficients are complex numbers and indeterminates are pairs $[g,h]$ with $g,h\in G$ such that $[g,h] \cdot [g', h'] = [gg', hh']$ for all $g,g',h,h' \in G$.
    For an operator $Z \in \fS = \operatorname{span}\set{ \ketbra{g}{h} }$, denote 
    \[
        Z' = \sum_{g,h} Z(g, h) [g,h] \in \bC[G \times G] \ .
        \]
\end{definition}

Note that $E' = [e,e]$ which is the neutral multiplicative element of $\bC[G \times G]$. Note that $L_i$ and $R_i$ do \emph{not} belong to $\fS$ hence $L_i'$ and $R_i'$ are not well-defined. However, their actions in the basis $\ketbra{g}{h}$ can be replicated by the appropriate left and right multiplications in $\bC[G \times G]$. Indeed, the multiplication on the left by $L_i$ acts in $\bC[G \times G]$ as the multiplication on the left by $[i,e]$, but the multiplication on the \emph{left} by $R_i$ acts in $\bC[G \times G]$ as the multiplication on the \emph{right} by $[i,e]$ because
\begin{align*}
(L_i\ketbra{g}{h})' &= (\ketbra{ig}{h})' = [ig,h] = [i,e] \cdot [g,h] \\
(R_i\ketbra{g}{h})' &= (\ketbra{gi}{h})' = [gi,h] = [g,h] \cdot [i, e] \ .
\end{align*}
As such, the operator $T_{\pm i}^{(n)} = P_{\pm i} V_n P_{\mp i}$ changes slightly when seen in $\bC[G \times G]$. This is captured by \cref{lem:vec-projections}.

\begin{lemma}
    \label{lem:vec-projections}
    For all $i\in[s]$ we have that
    \[
    (P_{\pm i} V_n P_{\mp i})' = K_{\pm i} V_n' K_{\pm i}
    \quad\text{where}\quad
    K_{\pm i} = K_{\pm i}^* = ([e,e] \pm [i,e])([e,e] \mp [e,i])/4 \ . 
    \]
\end{lemma}
\begin{proof}
    The adjoint in $\bC[G \times G]$ is defined as $X^* = \sum_{g,h} \overline{X(g,h)}[g^{-1}, h^{-1}]$ hence it should be clear that $K_{\pm i}^* = K_{\pm i}$ as $i^{-1} = i$ for each $i\in[s]$. Fix $g,h \in G$; we have 
    \begin{align*}
    (P_{\pm i}\ketbra{g}{h})' &= \frac{1}{4}(\ketbra{g}{h} + \ketbra{igi}{h} \pm \ketbra{ig}{h} \pm \ketbra{gi}{h})' = \frac{1}{4}([e,e] \pm [i,e])[g,h]([e,e] \pm [i,e])\\
    (\ketbra{g}{h}P_{\pm i})' &= \frac{1}{4}(\ketbra{g}{h} + \ketbra{g}{ihi} \pm \ketbra{g}{ih} \pm \ketbra{g}{hi})' = \frac{1}{4}([e,e] \pm [e,i])[g,h]([e,e] \pm [e,i]) \ .
    \end{align*}
    As $([e,e] + [i,e])$ and $([e,e] - [e,i])$ commute, we get that $(P_{\pm i} \ketbra{g}{h} P_{\mp i})' = K_{\pm i} [g,h] K_{\pm i}$ for all $g,h\in G$.
    Since $(X + \lambda Y)' = X' + \lambda Y'$, the lemma follows by linearity
    \[
    (P_{\pm i} V_n P_{\mp i})' 
    = \sum_{g,h \in G} V_n(g,h) (P_{\pm i} \ketbra{g}{h}P_{\mp i})'
    = \sum_{g,h \in G} V_n(g,h) K_{\pm i} [g,h] K_{\pm i} 
    = K_{\pm i} V_n' K_{\pm i} \ . \qedhere
    \]
\end{proof}

\paragraph{Preserving \ref{T1} and \ref{T2}.}
At this stage, it may not be clear what we have gained or lost from the vectorization step. Property \ref{T1} guaranteed that the output coloring was always proper. This aspect is preserved by the vectorization. Indeed, in \cref{sec:finite-dim}, we replace each $[g,h]$ by $U_g \otimes \bar{U}_h$ where the $U_g$ and $U_h$ are finite matrices. From the definition of $K_{\pm i}$, one then directly gets that it has zero energy, hence that the output coloring would always be proper. Recall, however, that the $K_{\pm i} V_n' K_{\pm i}$ do \emph{not} sum to the identity, hence do not define a proper measurement yet.

On the other hand, \ref{T2} (from \cref{thm:fin-dep-V}--\ref{V2}) ensures that the sum of the $K_{\pm i} V_n' K_{\pm i}$ can be made arbitrarily close to the identity. To make this precise, let us define $\norm{X}_F = \paren*{ \sum_{g,h} |X(g,h)|^2 }^{1/2}$ for all $X = \sum_{g,h} X(g,h)[g,h] \in \bC[G \times G]$ so that $\norm{Z}_F = \norm{Z'}_F$ for all $Z\in \fS$. In particular, by taking $Z = E - S(V_n)$ we have that
\[
\norm{[e,e]- S'(V_n')}_F \leq \lambda^n 
\quad\text{where}\quad
S'(X) = \sum_{i\in[s]} K_{+i} X K_{+i} + K_{-i} X K_{-i} \ .
\]
Overall, the vectorization preserved \ref{T1} and \ref{T2} more or less automatically. However, this is \emph{not} the case for \ref{T3} as we significantly changed the multiplicative structure. The remainder of this section is dedicated to enforcing positivity.

\paragraph{Positivity.}
Instead of proving that $K_{\pm i} V_n' K_{\pm i}$ is positive, we search for some $W_n \in \bC[G \times G]$ such that $W_n^2 \approx V_n'$ and use $K_{\pm i} W_n^2 K_{\pm i}$ in \cref{sec:finite-dim}. Doing so introduces some additional error, which we make desirably small as $n$ grows. The idea is to use the following Taylor approximation $\sqrt{1 + x} = 1 + x/2 - x^2/8 + x^3/16 - \ldots$ that is accurate for $x$ close to zero, and where $1$ is replaced by $[e,e]$ and $x$ by $(s/2)V_n' - [e,e]$ which we know to be small.

For the Taylor approximation to work, we need (1) a sub-multiplicative norm $\norm{ \cdot }_w$ and (2) that $\norm{ (s/2)V_n' - [e,e] }_w$ is small. Let us now prove that our weighted norm has the desired submultiplicativity. We emphasize that henceforth all products are taken in $\bC[G \times G]$.

\begin{restatable}{lemma}{NormSubmul}
    \label{lem:norm-submul}
    For every $X, Y\in \bC[G \times G]$ we have that $\norm{ X \cdot Y }_w \leq 4\norm{ X }_w \norm{ Y }_w$.
\end{restatable}

The proof follows closely that of Haagerup in \cite[Lemma 1.3]{haagerup-1979-an-example-of-a-non-nuclear-c-algebra}, where the most notable difference is that \cite{haagerup-1979-an-example-of-a-non-nuclear-c-algebra} works with a free group on $s$ generators while in our group $G$ each generator is also an involution. In \cite[Lemma 2]{picardello-1986-positive-definite-functions-and-l-p}, Picardello extends the convolution estimate of \cite[Lemma 1.3]{haagerup-1979-an-example-of-a-non-nuclear-c-algebra} to our group $G$. We give a proof of \cref{lem:norm-submul} for our specific use case in \cref{appx:norm-submul} for completeness.

\begin{proposition}
    \label{prop:sqrt-taylor}
    Let $Y \in \bC[G \times G]$ and $C = Y - [e,e]$. If $\norm{C}_w < 1/4$, for any $k \geq 0$ and
    \[
    X = \sum_{j=0}^k \binom{1/2}{j} C^j
    \quad\text{we have that}\quad
    \norm{ Y - X^2 }_w \to 0 
    \quad\text{as}\quad k \to \infty
    \]
    where $\binom{\alpha}{j} = \alpha(\alpha-1)(\alpha-2)\ldots(\alpha-j+1)/j!$.
\end{proposition}

One can prove \cref{prop:sqrt-taylor} by expanding $X^2$ for a finite $k$ and proving by induction that $\norm{ Y - X^2 }_w \leq (4\norm{ C }_w)^k$. We omit the proof for conciseness.

\begin{tbox}
\begin{lemma}
    \label{lem:positive-squares}
    For $s \geq 10^4$ and $n \geq 1$, there exists a self-adjoint $W_n \in \bC[G\times G]$ such that
    \[
    \norm{V_n' - W_n^2}_F \leq \lambda^n \ .
    \]
\end{lemma}
\end{tbox}
\begin{proof}
    By \cref{prop:sqrt-taylor} with $Y_n = (s/2)V_n'$ and $C_n = Y_n - [e,e]$, if $\norm{C_n}_w < 1/4$, for a large enough $k \geq 0$ the sum $X_n = \sum_{j=0}^k \binom{1/2}{j} C_n^j$ is such that
    $\norm{ (s/2)V_n' - X_n^2 }_F \leq \norm{ (s/2)V_n' - X_n^2 }_w \leq (s/2) \cdot \lambda^n$. We henceforth argue that $\norm{C_n}_w \leq 1/4$, and thus we can take $W_n = \sqrt{2/s}X_n$ and get the claimed bound as
    \[
    \norm{V_n' - W_n^2}_F = \norm{V_n' - (2/s)X_n^2}_F 
    = (2/s)\norm{(s/2)V_n' - X_n^2}_F \leq \lambda^n
    \]
    The fact that $W_n$ is self-adjoint follows directly from $V_n'$ being self adjoint. Prove that $V_n'$ is self-adjoint by induction, using that $V_0' = 0$ is self-adjoint and that $V_{n+1}' = V_n' + (2/s)([e,e]- S'(V_n'))$ preserves self-adjointness because the $K_{\pm i}$ are self-adjoint (\cref{lem:vec-projections}).
    
    It remains to prove that $\norm{C_n}_w < 1/4$. Using the recursive definition of $V_n$ and $C_n$ (\cref{eq:V-rec}), we can rewrite $C_n$ as follows
    \begin{align*}
    \norm{C_n}_w = \norm{ (s/2)V_n' - [e,e]}_w
    &= \norm{ (s/2)V_{n-1}' - S'(V_{n-1}') }_w \tag{using \cref{eq:V-rec}} \\
    &= \norm{ C_{n-1} + [e,e]- (2/s)S'(C_{n-1} + [e,e]) }_w \tag{by definition of $C_{n-1}$}  \\
    &\leq \norm{ C_{n-1} - (2/s)S'(C_{n-1}) }_w + \norm{ [e,e]- (2/s)S'([e,e]) }_w \tag{triangle inequality}\\
    &\leq \lambda\norm{C_{n-1}}_w + \norm{ [e,e]- (2/s)S'([e,e]) }_w \ . \tag{by \cref{lem:contraction}}
    \end{align*}
    We have that $S'([e,e]) = (s/2)[e,e]- (1/2)\sum_i [i,i]$ hence
    \[
    \norm{[e,e]- (2/s)S'([e,e])}_w = (1/s)\norm*{\sum_i [i,i]}_w \leq \frac{2}{\sqrt{s}} \ .
    \]
    This means that 
    \[
    \norm{ C_n }_w \leq \lambda^n\norm{C_1}_w + \frac{2}{\sqrt{s}}\sum_{j=0}^n \lambda^j
    \leq \frac{2}{(1 - \lambda)\sqrt{s}}
    \]
    where the last inequality uses that $C_1 = 0$ and $\lambda < 1$ for $s \geq 10^4$. By expanding the definition of $\lambda$ from \cref{lem:contraction}, we get that
    \[
    \norm{ C_n }_w \leq \frac{8\sqrt{2}}{\sqrt{2s} - 96} < 1/4 \ ,
    \] 
    where the last inequality holds for $s \geq 10^4$.
\end{proof}

We can now prove the main result of \cref{sec:inf}.

\begin{proof}[Proof of \cref{lem:W}]
    Left or right multiplication by a group element in $G\times G$ permutes the
    coefficients of a polynomial and therefore preserves its Frobenius norm.
    Expanding the two factors $K_{\sigma i}$ gives
    \[
    \|K_{\sigma i}ZK_{\sigma i}\|_F\leq
    \frac{1}{16}
    \sum_{\substack{
    \eta,\xi \in \{[e,e],-\sigma[e,g_i],\\
    \sigma[g_i,e],-[g_i,g_i]\}
    }}
    \| \eta Z \xi \|_F = \|Z\|_F,
    \]
    and thus
    \[
    \|S'(Z)\|_F\leq \sum_{i\in[s]}\sum_{\sigma\in\{-1,1\}}\|K_{\sigma i}ZK_{\sigma i}\|_F \leq 2s\|Z\|_F.
    \]
    Since $s\geq 10^4$, \Cref{thm:fin-dep-V,lem:positive-squares}
    give self-adjoint polynomials $V_n',W_n\in\C[G\times G]$ such that
    \[
    \|S'(V_n')-[e,e]\|_F\leq \lambda^n,
    \qquad
    \|V_n'-W_n^2\|_F\leq \lambda^n.
    \]
    Consequently,
    \[
    \|S'(W_n^2)-[e,e]\|_F
    \leq\|S'(V_n')-[e,e]\|_F+\|S'(W_n^2-V_n')\|_F
    \leq(2s+1)\lambda^n. \qedhere
    \]
\end{proof}

\section{Finite-Dimensional Quantum Realization}
\label{sec:finite-dim}

We turn the polynomial construction of \Cref{sec:vectorization}
into a one-way one-round anonymous quantum algorithm for coloring directed cycles,
by the energy-based characterization introduced in \cite{gur-li-2026-impossibility-of-one-way-one-round-quantum}.
First, we recall the definition of the energy of finite-dimensional PSD operators.

\begin{definition}[Energy] \label{def:energy}

    Let $\As=\Bs=\C^N$. Given PSD $M\in\mathrm{L}(\As\otimes\Bs)$,
    define its energy by
    \[
        \Es(M) := \Tr(AB)
        \qquad\text{where}\qquad A := \Tr_\As(M)^\top,\quad B:= \Tr_\Bs(M).
    \]
\end{definition}

The following key lemma relates the total energy of a POVM to the probability
of a monochromatic edge in a one-way one-round quantum algorithm.

\begin{lemma}[Corollary of {\cite[Theorem~6.1]{gur-li-2026-impossibility-of-one-way-one-round-quantum}}]
\label{lem:energy-to-coloring}
Fix an integer $q\geq 2$. Suppose that, for every $\epsilon>0$, there exist a dimension $N\in\mathbb{N}^+$ and a POVM $\{M_a\}_{a\in[q]}$ on $\C^N\otimes\C^N$ such that
\[
    \frac{1}{N^3}\sum_{a\in[q]} \Es(M_a) \leq \epsilon.
\]
Then, for every $\epsilon>0$, there is a one-way one-round anonymous quantum
algorithm 
that $q$-colors directed cycles, so that each edge 
is monochromatic with probability at most $\epsilon$.
\end{lemma}

It therefore suffices to construct POVMs with arbitrarily small normalized
total energy. The following theorem provides such POVMs with $q=2\cdot 10^4$ outcomes.

\begin{tbox}
\begin{theorem} \label{thm:main-povm}
Fix $q=2 \cdot 10^4$. For every $\epsilon>0$, there are a dimension
$N\in\mathbb{N}^+$ and a $q$-outcome POVM
$\{M_a\}_{a\in[q]}$ on $\C^N\otimes\C^N$ such that
\[
\frac{1}{N^3}\sum_{a\in[q]}\Es(M_a)\leq\epsilon.
\]
\end{theorem}
\end{tbox}

Combining \Cref{thm:main-povm,lem:energy-to-coloring} gives the following
corollary, which implies \Cref{res:main-epsilon}.
The remaining algorithmic results follow from the reductions in \Cref{sec:corollaries}.

\begin{corollary}\label{cor:one-round-coloring}
For every $\epsilon>0$, there is a one-way one-round anonymous quantum algorithm
that colors directed cycles with $q=2 \cdot 10^4$ colors, 
so that each edge is monochromatic with probability at most $\epsilon$.
\end{corollary}

\cref{alg:coloring} details the implementation of the one-way one-round anonymous quantum coloring algorithm obtained from \cref{thm:main-povm}.

\begin{alg}\label{alg:coloring}
Our $q$-coloring algorithm at node $v_i$ with $q=2 \cdot 10^4$ and a given $\epsilon>0$ error.

\begin{enumerate}[leftmargin=*,itemsep=0.5em,topsep=0.5em]
    \item Deterministically choose a dimension $N$ and a POVM
    $\mathcal{M}:=\{M_a\}_{a\in[q]}$ such that 
    \[
    \frac{1}{N^3}\sum_{a\in[q]}\Es(M_a)\leq\epsilon,
    \]
    as guaranteed by \Cref{thm:main-povm}.
    \item Prepare a bipartite maximally entangled state
    \[
    \ket{\Phi_N}_{\mathsf{A}_i\mathsf{B}_i}
    :=\frac1{\sqrt{N}}\sum_{j\in[N]}
    \ket{j}_{\mathsf{A}_i}\otimes\ket{j}_{\mathsf{B}_i}
    \]
    on two registers $\mathsf{A}_i$ and $\mathsf{B}_i$.

    \item In one communication round, send $\mathsf{B}_i$ to the successor
    $v_{i+1}$, keep $\mathsf{A}_i$, and receive $\mathsf{B}_{i-1}$ from the
    predecessor $v_{i-1}$.

    \item Apply $\mathcal{M}$ to the joint register
    $\mathsf{B}_{i-1}\otimes\mathsf{A}_i$, and output the
    measurement outcome $c_i\in[q]$ as the color of $v_i$.
\end{enumerate}
\end{alg}

The remainder of this section is to prove \Cref{thm:main-povm}. 

\subsection{Zero-Energy Operators and Normalization}

We use two helper lemmas about zero-energy operators.
The first characterizes their supports.

\begin{lemma}[Corollary of {\cite[Lemma~5.3]{gur-li-2026-impossibility-of-one-way-one-round-quantum}}] \label{lem:zero-energy}
    For a PSD operator $M\in\mathrm{L}(\C^N\otimes\C^N)$, the following are equivalent:
    \begin{enumerate}
        \item
        $M$ has zero energy, i.e., $\Es(M)=0$.
        \item There exists an orthogonal decomposition $\mathcal{V}\oplus \mathcal{V}^\perp = \C^N$
        such that
        \[
        \supp(M) \subseteq
        \mathcal{V} \otimes \bar{\mathcal{V}}^\perp.
        \]
    \end{enumerate}
\end{lemma}

The second normalizes a collection of zero-energy operators whose sum is close to
the identity into a POVM with small total energy.
We defer its proof to \Cref{appx:proof-normalize-povm}.

\begin{lemma} \label{lem:normalize-povm}
Let $\{\tilde{M}_a\}_{a\in[q]}$ be zero-energy PSD operators
in $\mathrm{L}(\C^N\otimes\C^N)$ such that
\[
\left\|
\sum_{a\in[q]} \tilde{M}_a - I
\right\|_F^2 \leq \epsilon.
\]
Then there exists a $q$-outcome POVM
$\{M_a\}_{a \in[q]}$ on the same space such that
\[
\sum_{a\in[q]} \Es(M_a)\leq 7\epsilon N.
\]
\end{lemma}

\subsection{Finite-Dimensional Approximation}
\label{sec:finite-approx}

We now evaluate the polynomials from \Cref{sec:vectorization} at
finite-dimensional matrices to obtain 
the POVMs with arbitrarily small total energy.
The following matrix families give unitary representations of the group $G$ that
approximate the orthogonality of distinct group elements in normalized Frobenius inner product.

\begin{lemma} \label{lem:matrix-families}
For every integer $s\geq 1$, there exists
a sequence of matrix families 
\[
\left(
\{U_i^{(N)}
\in \C^{N\times N}\}_{i\in[s]}\right)_{N\in\mathbb{N}^+}
\]
satisfying the following two conditions:
\begin{enumerate}[label=(U\arabic*),leftmargin=1.5cm]
    \item\label[part]{U1} For every $N\in\mathbb{N}^+$ and $i\in[s]$,
    $(U_i^{(N)})^\dagger=U_i^{(N)}$ and $(U_i^{(N)})^2=I_N$.
    \item\label[part]{U2} For $g=g_{i_1}\cdots g_{i_k}\in G$, write
    \[
    U_g^{(N)}:=U_{i_1}^{(N)}\cdots U_{i_k}^{(N)},
    \qquad U_e^{(N)}:=I_N.
    \]
    Then for every fixed $g,h\in G$,
    \[
    \frac1N\Tr\left((U_g^{(N)})^\dagger U_h^{(N)}\right)
    \longrightarrow 1_{g=h}
    \quad\text{as}\quad N\to\infty.
    \]
\end{enumerate}
\end{lemma}

We defer the proof of \Cref{lem:matrix-families} to \Cref{sec:matrix-families}.
Fix such a sequence. By \ref{U1}, $g\mapsto U_g^{(N)}$ is a unitary
representation of $G$, since it respects the defining relations $g_i^2=e$.
Define a linear map $\pi_N:\C[G\times G]\to\mathrm{L}(\C^N\otimes\C^N)$ by
\[
\pi_N([g,h]):=U_g^{(N)}\otimes\overline{{U}_h^{(N)}}.
\]

The map $\pi_N$ allows us to realize any polynomial in $\C[G\times G]$ with
finite-dimensional matrices. 
We require the following two properties of $\pi_N$.
The first one shows that $\pi_N$ preserves the algebraic structure of
$\C[G\times G]$.

\begin{proposition}\label{prop:finite-evaluation}
The map $\pi_N$ is a unital $*$-homomorphism,
i.e., it preserves
\begin{enumerate}
    \item multiplication: $\pi_N(XY)=\pi_N(X)\pi_N(Y)$ for all $X,Y\in\C[G\times G]$;
    \item adjoint: $\pi_N(X^*)=\pi_N(X)^\dagger$ for all $X\in\C[G\times G]$;\footnote{
    Recall that the adjoint of a group algebra element \(X\in\mathbb C[G\times G]\) is defined
    by \(X^*:=\sum_{g,h}\overline{X_{g,h}}[g^{-1},h^{-1}]\). For
    finite-dimensional matrices, we denote the adjoint by \(\dagger\), following
    standard notation in quantum computing. } 
    \item identity:
    $\pi_N([e,e])=I_{N^2}$.
\end{enumerate}
\end{proposition}

\begin{proof}
Since $\pi_N$ is linear, 
it suffices to check these properties on the group-algebra basis.
For any $(g,h),(g',h')\in G\times G$, we have
\begin{align*}
\pi_N([g,h])\pi_N([g',h'])
=U_g^{(N)}U_{g'}^{(N)}\otimes
  \overline{U_h^{(N)}U_{h'}^{(N)}}
=U_{gg'}^{(N)}\otimes\overline{U_{hh'}^{(N)}}
=\pi_N([gg',hh']).
\end{align*}
Since $(U_g^{(N)})^\dagger=U_{g^{-1}}^{(N)}$, we also have
\[
\pi_N([g,h])^\dagger
=U_{g^{-1}}^{(N)}\otimes\overline{U_{h^{-1}}^{(N)}}
=\pi_N([g^{-1},h^{-1}]).
\]
The identity element $(e,e)$ maps to $I_N\otimes I_N=I_{N^2}$.
\end{proof}

The next proposition shows that, for every fixed polynomial, the normalized Frobenius
norm of its evaluation under $\pi_N$ converges to the Frobenius norm of the polynomial.

\begin{proposition}\label{lem:norm-converge}
For every fixed polynomial $Z\in\C[G\times G]$,
\[
\frac1N\|\pi_N(Z)\|_F\longrightarrow\|Z\|_F
\quad\text{as}\quad N\to\infty.
\]
\end{proposition}

\begin{proof}
Write $Z=\sum_{g,h}Z_{g,h}[g,h]$ and set
\[
\tau_N(g,g'):=\frac1N\Tr\left((U_g^{(N)})^\dagger U_{g'}^{(N)}\right).
\]
Expanding the squared Frobenius norm gives
\begin{align*}
\frac1{N^2}\|\pi_N(Z)\|_F^2
&=\sum_{g,h}\sum_{g',h'}
  \overline{Z_{g,h}}Z_{g',h'}\,
  \tau_N(g,g')\overline{\tau_N(h,h')}\\
&\longrightarrow
  \sum_{g,h}\sum_{g',h'}
  \overline{Z_{g,h}}Z_{g',h'}\,1_{g=g'}1_{h=h'}
=\sum_{g,h}|Z_{g,h}|^2=\|Z\|_F^2,
\end{align*}
where the limit follows from \ref{U2}, since $Z\in\mathbb{C}[G\times G]$ has only finitely many nonzero
coefficients. Taking square roots proves the claim.
\end{proof}

Then we prove that evaluating the polynomial construction from \Cref{sec:vectorization} under
$\pi_N$ yields finite-dimensional zero-energy PSD operators whose sum
approximates the identity.

\begin{lemma}\label{lem:finite-approx}
Fix $s\geq 10^4$. For every $\epsilon>0$, there are a dimension
$N\in\mathbb{N}^+$ and zero-energy PSD operators
$\{M_a\}_{a\in[2s]}$ on $\C^N\otimes\C^N$ such that
\[
\frac1N\left\|\sum_{a\in[2s]}M_a-I\right\|_F\leq\epsilon.
\]
\end{lemma}

\begin{proof}
We first recall from \Cref{sec:vectorization} that 
\[
K_{\sigma i}=\frac14([e,e]+\sigma[g_i,e])([e,e]-\sigma[e,g_i])
\quad\text{and}\quad
S'(Z)=\sum_{i\in[s]}\sum_{\sigma\in\{-1,1\}}K_{\sigma i}ZK_{\sigma i}.
\]
Since $s \geq 10^4$, by \cref{lem:W}, 
for every $n\geq 1$,
there exists a self-adjoint $X$ such that $\norm{ [e,e] - S'(X^* X) }_F \leq (2s+1)\lambda^n$.
By \(\lambda\in(0,1)\), we can choose a positive integer 
$n \ge \log_{\lambda}\frac{\epsilon/2}{2s+1}$ so that 
\[
\|S'(X^* X)-[e,e]\|_F\leq \epsilon/2.
\]
Applying \Cref{lem:norm-converge} to the fixed polynomial
$S'(X^* X)-[e,e]$ gives a sufficiently large dimension $N$ such that
\[
\frac1N\left\|\pi_N(S'(X^* X))-I\right\|_F\leq\epsilon.
\]

Define $2s$ finite-dimensional operators
\[
M_{\sigma i}:=\pi_N(K_{\sigma i}X^* XK_{\sigma i}) \in \mathrm{L}(\C^N\otimes\C^N),
\qquad i\in[s],\quad\sigma\in\{-1,1\}.
\]
By linearity of $\pi_N$, their sum is
\[
T:=\sum_{i\in[s]}\sum_{\sigma\in\{-1,1\}}M_{\sigma i}
=\pi_N\left(\sum_{i\in[s]}\sum_{\sigma\in\{-1,1\}}K_{\sigma i}X^* XK_{\sigma i}\right)
=\pi_N(S'(X^* X)),
\]
so they satisfy the required bound $\frac{1}{N} \|T-I\|_F\leq\epsilon$.

Furthermore, since $\pi_N$ preserves multiplication and adjoints by \Cref{prop:finite-evaluation},
and $K_{\sigma i}$ is self-adjoint, we have
\[
M_{\sigma i}
=\bigl(\pi_N(X)\pi_N(K_{\sigma i})\bigr)^\dagger
 \bigl(\pi_N(X)\pi_N(K_{\sigma i})\bigr),
\]
which proves that $M_{\sigma i}$ is PSD.

To show they have zero energy, let
$\Pi_{\sigma i}:=(I_N+\sigma U_i^{(N)})/2$ and
$\mathcal{V}_{\sigma i}:=\supp(\Pi_{\sigma i})$.
By \ref{U1},
\[
\Pi_{\sigma i}^2
=\frac{I_N+\sigma^2 (U_i^{(N)})^2+2\sigma U_i^{(N)}}{4}
=\Pi_{\sigma i},\qquad
\Pi_{\sigma i}^\dagger = \frac{I_N+\sigma (U_i^{(N)})^\dagger}{2} = \Pi_{\sigma i},
\]
so $\Pi_{\sigma i}$ is an orthogonal projector. Moreover,
\[
\pi_N(K_{\sigma i})
=\frac{I_N+\sigma U_i^{(N)}}2\otimes
 \overline{\left(\frac{I_N-\sigma U_i^{(N)}}2\right)}
=\Pi_{\sigma i}\otimes\overline{(I_N-\Pi_{\sigma i})}.
\]
Thus $\pi_N(K_{\sigma i})$ is the orthogonal projector onto
$\mathcal{V}_{\sigma i}\otimes\overline{\mathcal{V}_{\sigma i}}^{\perp}$.
Since $M_{\sigma i}$ is supported on this subspace,
\Cref{lem:zero-energy} gives $\Es(M_{\sigma i})=0$.
Relabeling these $2s$ operators by $a\in[2s]$ completes the proof.
\end{proof}

Finally, we apply \Cref{lem:normalize-povm} to normalize these operators to prove
\Cref{thm:main-povm}.

\begin{proof}[Proof of \Cref{thm:main-povm}]
Fix $s= 10^4$.
Applying \Cref{lem:finite-approx} with error $\sqrt{\epsilon/7}$,
we obtain a collection of zero-energy PSD operators $\{\tilde{M}_a\}_{a\in[2s]}$ satisfying
\[
\left\|\sum_{a\in[2s]}\tilde{M}_a-I\right\|_F^2\leq\frac{\epsilon N^2}{7}.
\]
Then by \Cref{lem:normalize-povm}, there exists a POVM
$\{M_a\}_{a\in[2s]}$ with $2s=2\cdot 10^4$ outcomes such that
\[
\sum_{a\in[2s]}\Es(M_a)
\leq 7N \cdot \frac{\epsilon N^2}{7}=\epsilon N^3. 
\]
Thus $\frac{1}{N^3}\sum_{a\in[2s]}\Es(M_a)\leq\epsilon$, as required.
\end{proof}

\subsection{Constructing the Matrix Families}
\label{sec:matrix-families}

The required matrices of \Cref{lem:matrix-families} can be constructed from
independent Haar-random orthogonal matrices. We use the following consequence of
the asymptotic freeness theorem by Collins and \'{S}niady
\cite{collins-sniady-2006-integration-with-respect-to-the}, which states
that reduced products of independent Haar-random orthogonal matrices have
vanishing normalized traces.

\begin{lemma}[Corollary of {\cite[Theorem~5.2]{collins-sniady-2006-integration-with-respect-to-the}}]\label{lem:haar}
Fix $s\geq 1$. For each $d\geq 1$, let
$O_1^{(d)},\ldots,O_s^{(d)}\in\mathrm{O}(d)$ be independent
Haar-random matrices, all defined on a common probability space.
For \emph{fixed} $k\geq 1$, indices $i_1,\ldots,i_k\in[s]$ with
$i_j\neq i_{j+1}$ for $j\in[k-1]$, and signs $\sigma_1,\ldots,\sigma_k\in\{-1,1\}$,
we have
\[
\frac1d\Tr\left(
(O_{i_1}^{(d)})^{\sigma_1}\cdots
(O_{i_k}^{(d)})^{\sigma_k}
\right)\longrightarrow 0
\quad\text{as}\quad d\to\infty,
\]
with probability one.
\end{lemma}

The above lemma gives a probability-one event for each 
fixed choice of length, indices and signs. 
Since there are only countably many such choices,
we can strengthen the lemma to hold simultaneously for all choices.

\begin{corollary}\label{cor:haar-simultaneous}
Fix $s\geq 1$. For each $d\geq 1$, let
$O_1^{(d)},\ldots,O_s^{(d)}\in\mathrm{O}(d)$ be independent Haar-random matrices.
With probability one, we have
\[
\frac1d\Tr\left(
(O_{i_1}^{(d)})^{\sigma_1}\cdots
(O_{i_k}^{(d)})^{\sigma_k}
\right)\longrightarrow 0
\quad\text{as}\quad d\to\infty,
\]
\emph{simultaneously} for every $k\geq 1$, indices $i_1,\ldots,i_k\in[s]$
with $i_j\neq i_{j+1}$ for $1\leq j<k$, and signs
$\sigma_1,\ldots,\sigma_k\in\{-1,1\}$.
\end{corollary}

\begin{proof}
For each admissible choice of $k$, $\{i_j\}_{j\in[k]}$ and $\{\sigma_j\}_{j\in[k]}$, \Cref{lem:haar}
gives an event of probability one on which the corresponding limit holds.
There are only countably many such choices, so the intersection of these
events also has probability one.
\end{proof}

The following construction now proves \Cref{lem:matrix-families}.

\begin{lemma}\label{lem:matrix-families-realization}
Fix an integer $s\geq 1$. For each $d\geq 1$, let
$O_1^{(d)},\ldots,O_s^{(d)}\in\mathrm{O}(d)$ be independent Haar-random matrices.
For each $i\in[s]$, define
\[
U_i^{(2d)}:=
\begin{pmatrix}
0 & (O_i^{(d)})^\top\\
O_i^{(d)} & 0
\end{pmatrix},
\qquad
U_i^{(2d+1)}:=U_i^{(2d)}\oplus (1),
\]
and set $U_i^{(1)}:=(1)$.
With probability one, the matrix families
$\left(\{U_i^{(N)}\}_{i\in[s]}\right)_{N\in\mathbb{N}^+}$
satisfy conditions \ref{U1} and \ref{U2}. In particular, deterministic families
with these properties exist.
\end{lemma}

\begin{proof}
We first prove \ref{U1}, which holds deterministically.
Each $U_i^{(N)}$ is real symmetric and hence self-adjoint.
Since $O_i^{(d)}$ is orthogonal,
\[
(U_i^{(2d)})^2
=\begin{pmatrix}
(O_i^{(d)})^\top O_i^{(d)} & 0\\
0 & O_i^{(d)}(O_i^{(d)})^\top
\end{pmatrix}
=I_{2d}.
\]
One can verify that $(U_i^{(2d+1)})^2=I_{2d+1}$ also holds.
Thus \ref{U1} holds for all $N$, and
$g\mapsto U_g^{(N)}$ is a unitary representation of $G$.

To prove \ref{U2}, fix a realization in the probability-one event of
\Cref{cor:haar-simultaneous}. We claim that, for every $g\in G$,
\[
\frac1N\Tr(U_g^{(N)})\longrightarrow 1_{g=e}
\quad\text{as}\quad N\to\infty.
\]
For $g=e$, the normalized trace is exactly $1$.
Otherwise, write $g=g_{i_1}\cdots g_{i_k}$ in reduced form, so that
$k\geq 1$ and $i_j\neq i_{j+1}$ for $1\leq j<k$.
We first consider even dimension $N=2d$.

\begin{enumerate}
\item If $k$ is odd, then $U_g^{(2d)}$ is block off-diagonal and has trace zero.
\item If $k$ is even, then
    \[
    U_g^{(2d)}=
    \begin{pmatrix}
    D_1^{(d)} & 0\\
    0 & D_2^{(d)}
    \end{pmatrix},
    \quad\text{where}\quad
    \begin{array}{l}
    D_1^{(d)}:=(O_{i_1}^{(d)})^\top O_{i_2}^{(d)}\cdots
             (O_{i_{k-1}}^{(d)})^\top O_{i_k}^{(d)},\\[1ex]
    D_2^{(d)}:=O_{i_1}^{(d)}(O_{i_2}^{(d)})^\top\cdots
             O_{i_{k-1}}^{(d)}(O_{i_k}^{(d)})^\top.
    \end{array}
    \]
    Since $(O_i^{(d)})^\top=(O_i^{(d)})^{-1}$ and adjacent indices differ,
    \Cref{cor:haar-simultaneous} applies to both $D_1^{(d)}$ and $D_2^{(d)}$. Hence
    \[
    \frac1{2d}\Tr(U_g^{(2d)})
    =\frac12\left(\frac1d\Tr(D_1^{(d)})
                +\frac1d\Tr(D_2^{(d)})\right)
    \longrightarrow 0.
    \]
\end{enumerate}

For odd dimensions $N=2d+1$, we have
$U_g^{(2d+1)}=U_g^{(2d)}\oplus(1)$, so
\[
\frac1{2d+1}\Tr(U_g^{(2d+1)})
=\frac{2d}{2d+1}\left(\frac1{2d}\Tr(U_g^{(2d)})\right)
 +\frac1{2d+1}
\longrightarrow 1_{g=e}.
\]
Thus the claimed limit holds as $N\to\infty$.
Finally, since $g\mapsto U_g^{(N)}$ is a unitary representation,
for every $g,h\in G$ we have
\[
\frac1N\Tr\left((U_g^{(N)})^\dagger U_h^{(N)}\right)
=\frac1N\Tr(U_{g^{-1}h}^{(N)})
\longrightarrow 1_{g^{-1}h=e}=1_{g=h}.
\]
Therefore \ref{U2} holds on the same probability-one event.
Any realization in this event gives the required deterministic families.
\end{proof}

\begin{remark}
    The POVMs 
    for \Cref{thm:main-povm} can be constructed with only real 
    matrices, since we use random orthogonal matrices,
    and the required polynomials all have real coefficients.
\end{remark}

\subsection{Communication Complexity} \label{sec:communication-complexity}

The construction in \Cref{cor:one-round-coloring} guarantees a
finite-dimensional algorithm but does not give an explicit dimension
bound. 
In fact, a quantitative refinement yields
$N=\operatorname{poly}(1/\epsilon)$.
Consequently, each node sends $O(\log(1/\epsilon))$ qubits,
and this is optimal up to a constant factor
(see \Cref{appx:lower-bound}).

\begin{theorem}
For every $\epsilon>0$, there is a one-way one-round anonymous quantum algorithm
which sends $O(\log(1/\epsilon))$ qubits per node and 
colors directed cycles with $q=2 \cdot 10^4$ colors, 
so that each edge is monochromatic with probability at most $\epsilon$.
\end{theorem}

The proof is almost identical to that of \Cref{thm:main-povm}, so we give only a brief sketch.
The key ingredient is a quantitative version of the Haar-random
convergence in \Cref{lem:haar}.

\begin{lemma}[Quantitative version of \Cref{lem:haar}]
\label{lem:haar-quantitative}
Fix $s\geq 1$. For each $d\geq 1$, let
$O_1^{(d)},\ldots,O_s^{(d)}\in\mathrm{O}(d)$ be independent
Haar-random matrices.
For every $k\geq 1$, indices $i_1,\ldots,i_k\in[s]$ with
$i_j\neq i_{j+1}$ for $j\in[k-1]$, and signs
$\sigma_1,\ldots,\sigma_k\in\{-1,1\}$, if $d=\Omega(k^4)$, then
\[
    \mathbb{E}\left|
        \frac1d\Tr\left(
            (O_{i_1}^{(d)})^{\sigma_1}\cdots
            (O_{i_k}^{(d)})^{\sigma_k}
        \right)
    \right|
    \leq \frac{O(k^8)}{d}.
\]
\end{lemma}

Now we explain how to obtain the claimed dimension bound.
Fix $s=10^4$ as before. Choose the iteration parameter
$n=O(\log(1/\epsilon))$, and truncate the square-root polynomial
$W_n$ by retaining only terms $[g,h]$ with $|g|+|h|\leq R$.
A weighted-norm estimation shows that this truncation changes
the residual $S'(W_n^2)-1$ by at most $2^{-\Omega(R)}$
in Frobenius norm.
Thus, choosing $n$ and $R=O(\log(1/\epsilon))$ sufficiently
large gives a self-adjoint polynomial $X$ whose residual
$Z:=S'(X^2)-1$ satisfies $\|Z\|_F^2\leq\epsilon/14$.
Moreover, the maximum total word length in $Z$ is
$\ell=O(\log(1/\epsilon))$.
Since $s$ is fixed, the number $m$ of nonzero coefficients of $Z$
is bounded by
\[
    m\leq \exp(O(\ell))=\operatorname{poly}(1/\epsilon).
\]

We then proceed to evaluate $Z$ using the map $\pi_N$ associated with the
matrix families in \Cref{lem:matrix-families-realization} as before.
By \Cref{lem:haar-quantitative} and Cauchy--Schwarz,
\[
    \mathbb{E}\left[
        \frac1{N^2}\|\pi_N(Z)\|_F^2
    \right]
    \leq
    \|Z\|_F^2
    \left(1+O\left(\frac{m\ell^8}{N}\right)\right),
\]
for $N=\Omega(\ell^4)$.
Since $m\ell^8=\poly(1/\epsilon)$, we can choose
$N=\poly(1/\epsilon)$ large enough so that the factor in parentheses is at most $2$.
The expectation on the left-hand side is then at most
$2\|Z\|_F^2\leq\epsilon/7$.
Hence there exists a realization of the Haar-random matrices
satisfying the same bound.

Fix such a realization and applying \Cref{lem:normalize-povm} gives a POVM $\{M_a\}$ with
normalized total energy at most $\epsilon$. 
By \Cref{lem:energy-to-coloring}, the resulting algorithm has
collision probability at most $\epsilon$ on each edge.
Finally, each node transmits a single register of dimension
$N=\poly(1/\epsilon)$, which can be encoded using
$\lceil\log_2N\rceil=O(\log(1/\epsilon))$ qubits.

\section{Corollaries}\label{sec:corollaries}

We have now established \cref{res:main-epsilon}. Now we will describe how \crefrange{res:main-whp}{res:3col-rooted-tree} follow. We state the results in a bit more precise manner here.

\subsection{Coloring Cycles and Paths}

Let us first establish \cref{res:main-whp}:
\begin{theorem}\label{thm:oneway-whp}
    Assume that all nodes know the value of $n$. Then for any $C \ge 1$ there is a $1$-round one-way quantum algorithm that colors directed cycles with $2\cdot10^4$ colors with probability $1-1/n^C$.
\end{theorem}
\begin{proof}
    Set $\epsilon = 1/n^{C+1}$ in \cref{res:main-epsilon}. Then for each edge the probability that it is monochromatic is at most $1/n^{C+1}$. By the union bound, the probability that there is at least one monochromatic edge is at most $n \cdot 1/n^{C+1}$.
\end{proof}

Then let us switch from one-way algorithms in directed cycles to regular 1-round quantum algorithms in undirected cycles:
\begin{theorem}\label{thm:one-round-undirected-whp}
    Assume that all nodes know the value of $n$. Then for any $C \ge 1$ there is a $1$-round quantum-LOCAL algorithm that colors undirected cycles with $2\cdot10^4+2$ colors with probability $1-1/n^C$.
\end{theorem}
\begin{proof}
    Recall that the algorithm of \cref{thm:oneway-whp} consists of these steps:
    \begin{enumerate}
        \item Node $v$ prepares a pair of entangled registers $(A_v,B_v)$.
        \item Node $v$ sends $B_v$ to its successor $v+1$.
        \item Node $v$ receives $B_{v-1}$ from its predecessor $v-1$.
        \item Node $v$ measures $(B_{v-1}, A_v)$.
    \end{enumerate}
    This, of course, assumes that there is a unique well-defined predecessor and successor. In the undirected setting, each node has two neighbors, and it can refer to these with port numbers $1$ and $2$, but the port numbering can be arbitrary; it does not reveal any globally consistent orientation.
    
    However, unique identifiers will help. If each node compares its identifier with those of its two neighbors, it can distinguish between the following two cases:
    \begin{enumerate}
        \item The node is a local minimum or a local maximum w.r.t.\ unique identifiers.
        \item The node is in the middle of an increasing path of unique identifiers, and the orientation in which the identifiers grow gives us a consistent direction for the path fragment between the nearest local minimum and local maximum.
    \end{enumerate}
    So the idea is simple: we use the two extra colors for local minima and local maxima, and apply the algorithm for directed paths for the oriented segments between two such extrema.

    This gives a trivial $2$-round algorithm: first exchange identifiers to discover the identifiers of your two neighbors, and then you also know in which direction to send messages for \cref{thm:oneway-whp}.

    But we can do better; we can solve the task in $1$ round as follows. Each node creates \emph{two} pairs of entangled registers, $(A_v^1,B_v^1)$ and $(A_v^2,B_v^2)$, and then for each port $p$, it sends in the same round both its identifier $v$ and register $B_v^p$ to port $p$. Then when it receives the messages, it knows if it is a local minimum or maximum, and if so, uses the special colors. Otherwise it follows the one-way algorithm; now it knows which of the ports $p$ points in the direction of increasing identifiers. Then node $v$ defines that $B_{v-1}$ is the register it received from the neighbor with the smaller identifier (``predecessor''), $A_v = A_v^p$ is the register it kept for itself, while $B_v = B_v^p$ is the register it sent to the neighbor with the larger identifier (``successor'').

    Now if an edge $\{u,v\}$ is such that at least one of the nodes is a local minimum or maximum, it will be trivially properly colored. Otherwise, it is in the middle of an increasing path of identifiers, and a proper coloring follows from the correctness of the one-way algorithm.
\end{proof}

Then \cref{res:3color} follows easily:
\begin{theorem}\label{thm:3color}
    Assume that all nodes know the value of $n$. Then for any $C \ge 1$ there is a $4$-round quantum-LOCAL algorithm that $3$-colors undirected cycles and paths with probability $1-1/n^C$.
\end{theorem}
\begin{proof}
    In the first round, apply \cref{thm:one-round-undirected-whp}. Then apply the $3$-round color reduction algorithm from \cite{kohonen-korhonen-etal-2017-distributed-colour-reduction} that reduces the number of colors from $10^{100}$ to $3$ in three rounds.

    This covers the case of cycles. If we are in a path, each endpoint of a path can pretend that it is in the middle of a path; it is sufficient to simulate $4$ additional nodes.
\end{proof}

\subsection{Coloring Rooted Trees}

It turns out that it will be convenient to establish next \cref{res:3col-rooted-tree}:
\begin{theorem}\label{thm:3-col-rooted-tree}
    Assume that all nodes know the values of $n$ and $\Delta$. Then for any $C \ge 1$ there is an $O(\log^* \Delta)$-round quantum-LOCAL algorithm that $3$-colors rooted trees of maximum degree $\Delta$ with probability $1-1/n^C$.
\end{theorem}
\begin{proof}
    We will follow the convention that the edges are oriented towards the root (so each node has at most $1$ outgoing edge).
    
    First, we partition the edges into $\Delta$ classes, as follows: for each node, the $i$th incoming edge is assigned to class $i$. Note that for each $i$, the edges of class $i$ form a collection of (directed) paths; each node belongs to up to $\Delta$ classes, but each edge belongs to exactly one class.

    Then, in parallel for each class $i$, we apply \cref{thm:3color} to $3$-color the collection of paths (there are at most $\Delta n \le n^2$ applications of \cref{thm:3color}, so we can choose error probability $1/n^{C+2}$ for each application, so that they all succeed with probability at least $1-1/n^C$). We label each node $v$ with a color vector $(c_i)_i$, where $c_i$ is its color for class $i$ (or $1$ if it did not participate in this class). Note that these vectors form a proper coloring of the tree, with $3^{\Delta}$ colors (edge $\{u,v\}$ is properly colored, since it belongs to some class $i$, and then component $i$ of the color vectors differs between $u$ and $v$).
    
    The rest uses standard techniques \cite[Chapter~3]{barenboim-elkin-2013-distributed-graph-coloring}: For example, we can use the color reduction technique of \cite{naor-stockmeyer-1995-what-can-be-computed-locally} to reduce the number of colors from $3^{\Delta}$ to $4$, in $O(\log^* 3^{\Delta}) = O(\log^* \Delta)$ rounds. Then we can shift colors down, so that for each node all children have the same color. In particular, for nodes of color $4$ there is at least one color from $\{1,2,3\}$ that is not used by any of their neighbors; nodes of color $4$ can hence safely recolor themselves.
\end{proof}

\subsection{Coloring Graphs, with Simple Applications}

Now that we can color rooted trees, we can also use standard techniques to color graphs. Let us first recover Linial's $O(\Delta^2)$-coloring algorithm:
\begin{theorem}\label{thm:delta-squared-coloring}
    Assume that all nodes know the values of $n$ and $\Delta$. Then for any $C \ge 1$ there is an $O(\log^* \Delta)$-round quantum-LOCAL algorithm that finds an $O(\Delta^2)$-coloring in graphs of maximum degree $\Delta$ with probability $1-1/n^C$.
\end{theorem}
\begin{proof}
    Orient the edges from smaller to larger unique identifier. Divide the edges into up to $\Delta$ classes so that for each node its $i$th outgoing edge belongs to class $i$. Each class forms a collection of rooted trees. Apply \cref{thm:3-col-rooted-tree} in parallel for each class (again with an appropriately chosen error probability). Put together, the colorings of the forests form a $3^{\Delta}$-coloring of the original graph. Then use Linial's \cite{linial-1992-locality-in-distributed-graph-algorithms} color reduction technique (based on cover-free families) to reduce the number of colors first down to $O(\Delta^2)$ in $O(\log^* 3^{\Delta}) = O(\log^* \Delta)$ rounds.
\end{proof}

We could combine this primitive with state-of-the-art distributed graph coloring algorithms \cite{maus-tonoyan-2022-linial-for-lists,barenboim-elkin-goldenberg-2022-locally-iterative}, and obtain an algorithm for coloring graphs of maximum degree $\Delta$ in $\tilde{O}(\sqrt{\Delta})$ rounds. However, we will settle for the following much simpler bound that suffices for us; overall, on top of \cref{res:main-epsilon}, we only need textbook-level ingredients \cite[Chapter 4]{hirvonen-suomela-2020-distributed-algorithms-2020}:

\begin{theorem}\label{thm:delta+1-coloring}
    Assume that all nodes know the values of $n$ and $\Delta$. Then for any $C \ge 1$ there is an $O(\Delta)$-round quantum-LOCAL algorithm that finds a ${(\Delta+1)}$-coloring in graphs of maximum degree $\Delta$ with probability $1-1/n^C$.
\end{theorem}
\begin{proof}
    Apply \cref{thm:delta-squared-coloring}. Then apply additive-group coloring from \cite{barenboim-elkin-goldenberg-2022-locally-iterative} to reduce the number of colors to $O(\Delta)$ in $O(\Delta)$ rounds, and finally do trivial color reduction to push it down to $\Delta+1$ in $O(\Delta)$ rounds.
\end{proof}
\begin{theorem}
    Assume that all nodes know the values of $n$ and $\Delta$. Then for any $C \ge 1$ there is an $O(\Delta)$-round quantum-LOCAL algorithm that finds a maximal independent set in graphs of maximum degree $\Delta$ with probability $1-1/n^C$.
\end{theorem}
\begin{proof}
    Apply \cref{thm:delta+1-coloring}, and then proceed greedily by color classes.
\end{proof}
\begin{theorem}
    Assume that all nodes know the values of $n$ and $\Delta$. Then for any $C \ge 1$ there are $O(\Delta)$-round quantum-LOCAL algorithms that find a maximal matching and a ${(2\Delta-1)}$-edge coloring in graphs of maximum degree $\Delta$ with probability $1-1/n^C$.
\end{theorem}
\begin{proof}
    Apply the algorithms of \cite{panconesi-rizzi-2001-some-simple-distributed-algorithms}, replacing the part where rooted trees are colored with \cref{thm:3-col-rooted-tree}.
\end{proof}
These establish \cref{res:color}.

\subsection{LCL Problems}

Let us then turn our attention to LCL problems. We will use $R$ to denote the ``checking radius'' of the LCL.

Let us first establish \cref{res:lcls}; it will be convenient to directly write it in the form that applies to all $o(\log n)$-round algorithms:
\begin{theorem}\label{thm:lcls}
    Let $\Pi$ be an LCL problem, and assume that there is an algorithm $A$ that solves $\Pi$ in $T(n) = o(\log n)$ rounds in the classical deterministic LOCAL model. Assume that all nodes know the value of $n$. Then for any $C \ge 1$ there is an $O(1)$-round quantum-LOCAL algorithm that solves $\Pi$ with probability $1-1/n^C$.
\end{theorem}
\begin{proof}
    This is a direct application of \cite{chang-kopelowitz-pettie-2019-an-exponential-separation} combined with \cref{thm:delta-squared-coloring}. In essence, we choose an appropriately large constant $N$, and lie to $A$ that the number of nodes is $N$. We only need to argue that if $A$ fails under this lie, $A$ will also fail in some graph that honestly has $N$ nodes.

    The only issue is that $A$ can assume that in a graph with $N$ nodes, the unique identifiers are polynomial in $N$, while in our original input graph $G$ the unique identifiers can be much larger, polynomial in $n$. We will use fast graph coloring algorithms to produce fake unique identifiers that are sufficiently small.

    Fix $\Pi$; this also fixes the maximum degree $\Delta$ of our input graph $G$. Let $r = 2T(N)+2R$. We apply the result of \cref{thm:delta-squared-coloring} to $G^r$; note that $G^r$ has maximum degree at most $\Delta^r$, and our choice of $N$ then ensures that $T(N) \ll \log_{\Delta} N$ and therefore $\Delta^r \ll N$. Hence \cref{thm:delta-squared-coloring} yields a coloring with polynomial-in-$N$ colors in $O(\log^* \Delta^r) = O(1)$ rounds. These are appropriate fake identifiers that we can give to $A$ when we lie that the number of nodes is $N$. If $A$ fails to solve $\Pi$ locally in some local neighborhood of $G$, then we can also construct a genuine $N$-sized instance $G'$ with the same unique identifier assignment, and our deterministic algorithm $A$ fails also in $G'$.
\end{proof}

Now \cref{res:lcl-classes-rooted-trees} follows directly from \cref{thm:lcls} and \cite{akbari-coiteux-roy-etal-2025-online-locality-meets}, and \cref{res:lcl-classes-rooted-regular-trees} follows directly from \cref{thm:lcls} and \cite{dhar-kujawa-etal-2024-local-problems-in-trees-across-a}.

\subsection{SLOCAL Model}

Recall that the SLOCAL model \cite{ghaffari-kuhn-maus-2017-on-the-complexity-of-local} is the sequential counterpart of the LOCAL model. In the LOCAL model all nodes make decisions simultaneously in parallel, but in the SLOCAL model they make decisions in some adversarially chosen sequential order. This ordering trivially breaks symmetry, and hence many problems that cannot be solved in constant time in the LOCAL model admit $O(1)$-locality algorithms in the SLOCAL model.

Our \cref{res:slocal} shows that in bounded-degree graphs, constant-round quantum-LOCAL algorithms are at least as strong as constant-locality SLOCAL algorithms; a bit more formally, we have:

\begin{theorem}\label{thm:slocal}
    Fix a constant $\Delta$. Let $\Pi$ be any graph problem in graphs of maximum degree $\Delta$, and assume that $A$ is a deterministic SLOCAL-model algorithm with locality $T=O(1)$ that solves $\Pi$. Assume that all nodes know the value of $n$. Then for any $C \ge 1$ there is an $O(1)$-round quantum-LOCAL algorithm that solves $\Pi$ with probability $1-1/n^C$.
\end{theorem}
\begin{proof}
    Let $r = 2T+2$. Use \cref{thm:delta-squared-coloring} to find a coloring of $G^r$. Then proceed by color classes; in step $i$ we apply $A$ in parallel to all nodes of color class $i$. The radius-$T$ neighborhoods of the nodes of each color class are disjoint, and hence the result of such a parallel application of $A$ is indistinguishable from an adversary that applies $A$ first to all nodes of color class $1$ in some arbitrary order, then to all nodes of color class $2$ in some arbitrary order, etc. Each step can be simulated in $O(T)$ rounds, and hence overall the round complexity will be bounded by a constant (that depends on $T$ and~$\Delta$).
\end{proof}

We also note that \cref{thm:lcls} can be derived as a corollary of \cref{thm:slocal} and the well-known connections between the LOCAL and SLOCAL models \cite{chang-kopelowitz-pettie-2019-an-exponential-separation,ghaffari-kuhn-maus-2017-on-the-complexity-of-local}. Indeed, all of these are known to represent the same class of LCL problems \cite{chang-kopelowitz-pettie-2019-an-exponential-separation}:
\begin{enumerate}
    \item LCLs solvable in classical randomized LOCAL in $o(\log \log n)$ rounds.
    \item LCLs solvable in classical deterministic LOCAL in $o(\log n)$ rounds.
    \item LCLs solvable in classical deterministic LOCAL in $O(\log^* n)$ rounds.
    \item LCLs solvable in deterministic SLOCAL with $O(1)$ locality.
\end{enumerate}
Now we know that this class of problems admits an $O(1)$-round quantum-LOCAL algorithm.

\subsection{From Monte Carlo to Las Vegas}

We have focused on Monte Carlo algorithms in this work. We briefly note that essentially all of our results can be also turned into Las Vegas algorithms, using the standard idea that we can verify the coloring locally and stop early if a node is locally happy, and nodes that are stopping later can ensure that their color choice is compatible with the choices that their neighbors that stopped have already made earlier; see e.g.\ \cite{korman-sereni-viennot-2011-toward-more-localized-local}.

Let us use \cref{thm:3color} as a concrete example. We can proceed as follows. We have two color palettes, the final colors $\{1,2,3\}$ and \emph{tentative} colors $\{4,5,6\}$. We repeat the following steps:
\begin{enumerate}
    \item \label{step:undecided} All nodes that are still undecided apply \cref{thm:3color} to produce a coloring with $\{4,5,6\}$ in $4$ rounds (this might fail).
    \item Each node then \emph{verifies} if the color it got is different from the colors of its neighbors. If so, it switches to the \emph{finalization} phase.
    \item Nodes in the finalization phase will do $3$ steps of greedy color reduction. First nodes of color $6$ recolor themselves with the smallest color from $\{1,2,3\}$ that is not used by their neighbors. Then we do the same for nodes of color $5$, and nodes of color $4$. After this, the node has decided its final color and can stop.
    \item Nodes that did not make it to the finalization phase remain undecided and start again from step~\ref{step:undecided}.
\end{enumerate}

\section{AI Methodology}\label{sec:ai}

\textbf{Our main result was discovered using OpenAI GPT-6 in Codex.} Our starting point was the \emph{lower-bound} result in \cite{gur-li-2026-impossibility-of-one-way-one-round-quantum}, and we tried to use the LLM to prove a \emph{stronger} lower-bound result; the first step would be to extend the lower bound from $4$ colors to an arbitrarily large number of colors. To our surprise, GPT-6 showed that this is \emph{not} possible: there is, actually, an algorithm.

To the best of our knowledge, \cite{gur-li-2026-impossibility-of-one-way-one-round-quantum} was the critical human-supplied ingredient that enabled the breakthrough for the AI tools. We have tried to study the same question extensively (also with the latest models such as OpenAI GPT-6 and Claude Fable 5.1), but there was no tangible progress until \cite{gur-li-2026-impossibility-of-one-way-one-round-quantum} appeared.

Once \cite{gur-li-2026-impossibility-of-one-way-one-round-quantum} was available, the task was apparently easy for the LLMs. The AI agent proved the equivalents of \cref{res:main-epsilon} and \cref{res:main-whp} in a few hours, and in regular OpenAI API pricing the cost would have been in the ballpark of 50 euros (much less than that using the monthly ChatGPT subscription).

Beyond the use of AI for initial discovery, we have used AI tools extensively to ensure the correctness of the result---in particular, many versions of \cref{res:main-epsilon}, as well as the corresponding negative result for $4$-coloring from \cite{gur-li-2026-impossibility-of-one-way-one-round-quantum}, were formalized in Lean 4 with the help of GPT-6 and Codex. We have also used AI tools to e.g.\ better understand the proof, explore alternative proof techniques, simplify the proof, explore connections with finitely-dependent colorings. In addition to Codex and GPT-6, we have also used Claude Code and Fable 5.1.

While we used AI tools extensively to discover the proof, \textbf{this paper is entirely written by humans, for humans}. In writing, AI was only used to correct spelling errors and grammar mistakes and to verify the proofs that we wrote. To ensure a clean separation between AI slop and human-authored material, we had two Git repositories, one where the AI agents worked, and one where this manuscript was written.

\section*{Acknowledgements}

This work was supported in part by the Research Council of Finland, Grant 363558. Longcheng Li's work was supported by ERC Starting Grant 101163189.
We would like to thank Tom Gur, Marc-Olivier Renou and Xavier Coiteux-Roy for discussions. For AI usage, see \cref{sec:ai}.

\bibliographystyle{alphaurl}
\bibliography{da}

@inproceedings{akbari-coiteux-roy-etal-2025-online-locality-meets,
  author = {Amirreza Akbari and Xavier Coiteux-Roy and Francesco d'Amore and Fran{\c{c}}ois {Le Gall} and Henrik Lievonen and Darya Melnyk and Augusto Modanese and Shreyas Pai and Marc-Olivier Renou and V{\'a}clav Rozho{\v{n}} and Jukka Suomela},
  editor = {Michal Kouck{\'y} and Nikhil Bansal},
  title = {Online Locality Meets Distributed Quantum Computing},
  booktitle = {Proceedings of the 57th Annual ACM Symposium on Theory of Computing, STOC 2025, Prague, Czechia, June 23-27, 2025},
  pages = {1295--1306},
  year = {2025},
  publisher = {ACM},
  doi = {10.1145/3717823.3718211}
}

@inproceedings{akbari-eslami-etal-2023-locality-in-online-dynamic,
  author = {Amirreza Akbari and Navid Eslami and Henrik Lievonen and Darya Melnyk and Joona S{\"a}rkij{\"a}rvi and Jukka Suomela},
  editor = {Kousha Etessami and Uriel Feige and Gabriele Puppis},
  title = {Locality in Online, Dynamic, Sequential, and Distributed Graph Algorithms},
  booktitle = {50th International Colloquium on Automata, Languages, and Programming, ICALP 2023, Paderborn, Germany, July 10-14, 2023},
  series = {LIPIcs},
  volume = {261},
  pages = {10:1--10:20},
  year = {2023},
  publisher = {Schloss Dagstuhl - Leibniz-Zentrum f{\"u}r Informatik},
  doi = {10.4230/LIPIcs.ICALP.2023.10}
}

@article{arfaoui-fraigniaud-2014-what-can-be-computed-without,
  author = {Heger Arfaoui and Pierre Fraigniaud},
  title = {What can be computed without communications?},
  journal = {SIGACT News},
  volume = {45},
  number = {3},
  pages = {82--104},
  year = {2014},
  doi = {10.1145/2670418.2670440}
}

@inproceedings{balliu-brandt-etal-2020-how-much-does-randomness-help,
  author = {Alkida Balliu and Sebastian Brandt and Dennis Olivetti and Jukka Suomela},
  editor = {Yuval Emek and Christian Cachin},
  title = {How much does randomness help with locally checkable problems?},
  booktitle = {PODC '20: ACM Symposium on Principles of Distributed Computing, Virtual Event, Italy, August 3-7, 2020},
  pages = {299--308},
  year = {2020},
  publisher = {ACM},
  doi = {10.1145/3382734.3405715}
}

@article{balliu-brandt-etal-2021-almost-global-problems-in-the,
  author = {Alkida Balliu and Sebastian Brandt and Dennis Olivetti and Jukka Suomela},
  title = {Almost global problems in the {LOCAL} model},
  journal = {Distributed Computing},
  volume = {34},
  number = {4},
  pages = {259--281},
  year = {2021},
  doi = {10.1007/s00446-020-00375-2}
}

@article{balliu-brandt-etal-2021-lower-bounds-for-maximal,
  author = {Alkida Balliu and Sebastian Brandt and Juho Hirvonen and Dennis Olivetti and Mika{\"e}l Rabie and Jukka Suomela},
  title = {Lower Bounds for Maximal Matchings and Maximal Independent Sets},
  journal = {Journal of the ACM},
  volume = {68},
  number = {5},
  pages = {39:1--39:30},
  year = {2021},
  doi = {10.1145/3461458}
}

@inproceedings{balliu-brandt-etal-2025-distributed-quantum-advantage,
  author = {Alkida Balliu and Sebastian Brandt and Xavier Coiteux-Roy and Francesco d'Amore and Massimo Equi and Fran{\c{c}}ois {Le Gall} and Henrik Lievonen and Augusto Modanese and Dennis Olivetti and Marc-Olivier Renou and Jukka Suomela and Lucas Tendick and Isadora Veeren},
  editor = {Michal Kouck{\'y} and Nikhil Bansal},
  title = {Distributed Quantum Advantage for Local Problems},
  booktitle = {Proceedings of the 57th Annual ACM Symposium on Theory of Computing, STOC 2025, Prague, Czechia, June 23-27, 2025},
  pages = {451--462},
  year = {2025},
  publisher = {ACM},
  doi = {10.1145/3717823.3718233}
}

@inproceedings{balliu-casagrande-etal-2026-distributed-quantum,
  author = {Alkida Balliu and Filippo Casagrande and Francesco d'Amore and Massimo Equi and Barbara Keller and Henrik Lievonen and Dennis Olivetti and Gustav Schmid and Jukka Suomela},
  editor = {Kasper Green Larsen and Barna Saha},
  title = {Distributed Quantum Advantage in Locally Checkable Labeling Problems},
  booktitle = {Proceedings of the 2026 Annual ACM-SIAM Symposium on Discrete Algorithms, SODA 2026, Vancouver, BC, Canada, January 11-14, 2026},
  pages = {1268--1308},
  year = {2026},
  publisher = {SIAM},
  doi = {10.1137/1.9781611978971.49}
}

@inproceedings{balliu-coupette-etal-2025-new-limits-on-distributed,
  author = {Alkida Balliu and Corinna Coupette and Antonio Cruciani and Francesco d'Amore and Massimo Equi and Henrik Lievonen and Augusto Modanese and Dennis Olivetti and Jukka Suomela},
  editor = {Dariusz R. Kowalski},
  title = {New Limits on Distributed Quantum Advantage: Dequantizing Linear Programs},
  booktitle = {39th International Symposium on Distributed Computing, DISC 2025, Berlin, Germany, October 27-31, 2025},
  series = {LIPIcs},
  volume = {356},
  pages = {11:1--11:22},
  year = {2025},
  publisher = {Schloss Dagstuhl - Leibniz-Zentrum f{\"u}r Informatik},
  doi = {10.4230/LIPIcs.DISC.2025.11}
}

@inproceedings{balliu-hirvonen-etal-2018-new-classes-of-distributed,
  author = {Alkida Balliu and Juho Hirvonen and Janne H. Korhonen and Tuomo Lempi{\"a}inen and Dennis Olivetti and Jukka Suomela},
  editor = {Ilias Diakonikolas and David Kempe and Monika Henzinger},
  title = {New classes of distributed time complexity},
  booktitle = {Proceedings of the 50th Annual ACM SIGACT Symposium on Theory of Computing, STOC 2018, Los Angeles, CA, USA, June 25-29, 2018},
  pages = {1307--1318},
  year = {2018},
  publisher = {ACM},
  doi = {10.1145/3188745.3188860}
}

@book{barenboim-elkin-2013-distributed-graph-coloring,
  author = {Leonid Barenboim and Michael Elkin},
  title = {Distributed Graph Coloring: Fundamentals and Recent Developments},
  series = {Synthesis Lectures on Distributed Computing Theory},
  pages = {1--171},
  year = {2013},
  publisher = {Morgan \& Claypool Publishers},
  isbn = {9781627050180},
  doi = {10.2200/S00520ED1V01Y201307DCT011}
}

@article{barenboim-elkin-goldenberg-2022-locally-iterative,
  author = {Leonid Barenboim and Michael Elkin and Uri Goldenberg},
  title = {Locally-iterative Distributed {$(\Delta+1)$}-coloring and Applications},
  journal = {Journal of the ACM},
  volume = {69},
  number = {1},
  pages = {5:1--5:26},
  year = {2022},
  doi = {10.1145/3486625}
}

@inproceedings{boudier-kuhn-etal-2026-classification-of-local,
  author = {Thomas Boudier and Fabian Kuhn and Augusto Modanese and Ronja Stimpert and Jukka Suomela},
  editor = {Sayan Bhattacharya and Danupon Nanongkai and Michael Benedikt and Gabriele Puppis},
  title = {Classification of Local Optimization Problems in Directed Cycles},
  booktitle = {53rd International Colloquium on Automata, Languages, and Programming, ICALP 2026, Royal Holloway, University of London, Egham, United Kingdom, July 7-10, 2026},
  series = {LIPIcs},
  volume = {374},
  pages = {42:1--42:23},
  year = {2026},
  publisher = {Schloss Dagstuhl - Leibniz-Zentrum f{\"u}r Informatik},
  doi = {10.4230/LIPIcs.ICALP.2026.42}
}

@inproceedings{brandt-fischer-etal-2016-a-lower-bound-for-the,
  author = {Sebastian Brandt and Orr Fischer and Juho Hirvonen and Barbara Keller and Tuomo Lempi{\"a}inen and Joel Rybicki and Jukka Suomela and Jara Uitto},
  editor = {Daniel Wichs and Yishay Mansour},
  title = {A lower bound for the distributed {L}ov{\'a}sz local lemma},
  booktitle = {Proceedings of the 48th Annual ACM SIGACT Symposium on Theory of Computing, STOC 2016, Cambridge, MA, USA, June 18-21, 2016},
  pages = {479--488},
  year = {2016},
  publisher = {ACM},
  doi = {10.1145/2897518.2897570}
}

@inproceedings{brandt-gottlicher-2026-a-post-quantum-lower-bound-for,
  author = {Sebastian Brandt and Tim G{\"o}ttlicher},
  editor = {Kasper Green Larsen and Barna Saha},
  title = {A Post-Quantum Lower Bound for the Distributed {L}ov{\'a}sz Local Lemma},
  booktitle = {Proceedings of the 2026 Annual ACM-SIAM Symposium on Discrete Algorithms, SODA 2026, Vancouver, BC, Canada, January 11-14, 2026},
  pages = {3298--3312},
  year = {2026},
  publisher = {SIAM},
  doi = {10.1137/1.9781611978971.119}
}

@inproceedings{chang-2020-the-complexity-landscape-of-distributed,
  author = {Yi-Jun Chang},
  editor = {Hagit Attiya},
  title = {The Complexity Landscape of Distributed Locally Checkable Problems on Trees},
  booktitle = {34th International Symposium on Distributed Computing, DISC 2020, Virtual Conference, October 12-16, 2020},
  series = {LIPIcs},
  volume = {179},
  pages = {18:1--18:17},
  year = {2020},
  publisher = {Schloss Dagstuhl - Leibniz-Zentrum f{\"u}r Informatik},
  doi = {10.4230/LIPIcs.DISC.2020.18}
}

@article{chang-kopelowitz-pettie-2019-an-exponential-separation,
  author = {Yi-Jun Chang and Tsvi Kopelowitz and Seth Pettie},
  title = {An Exponential Separation between Randomized and Deterministic Complexity in the {LOCAL} Model},
  journal = {SIAM Journal on Computing},
  volume = {48},
  number = {1},
  pages = {122--143},
  year = {2019},
  doi = {10.1137/17M1117537}
}

@article{chang-pettie-2019-a-time-hierarchy-theorem-for-the,
  author = {Yi-Jun Chang and Seth Pettie},
  title = {A Time Hierarchy Theorem for the {LOCAL} Model},
  journal = {SIAM Journal on Computing},
  volume = {48},
  number = {1},
  pages = {33--69},
  year = {2019},
  doi = {10.1137/17M1157957}
}

@inproceedings{coiteux-roy-d-amore-etal-2024-no-distributed-quantum,
  author = {Xavier Coiteux-Roy and Francesco d'Amore and Rishikesh Gajjala and Fabian Kuhn and Fran{\c{c}}ois {Le Gall} and Henrik Lievonen and Augusto Modanese and Marc-Olivier Renou and Gustav Schmid and Jukka Suomela},
  editor = {Bojan Mohar and Igor Shinkar and Ryan O'Donnell},
  title = {No Distributed Quantum Advantage for Approximate Graph Coloring},
  booktitle = {Proceedings of the 56th Annual ACM Symposium on Theory of Computing, STOC 2024, Vancouver, BC, Canada, June 24-28, 2024},
  pages = {1901--1910},
  year = {2024},
  publisher = {ACM},
  doi = {10.1145/3618260.3649679}
}

@article{coiteux-roy-flin-etal-2026-distributed-quantum,
  author = {Xavier Coiteux-Roy and Maxime Flin and Carlos de Gois and Marc-Olivier Renou and Jukka Suomela and Isadora Veeren},
  title = {Distributed Quantum Algorithms Cannot Color Cycles with Probability 1},
  journal = {CoRR},
  volume = {abs/2608.11720},
  year = {2026},
  eprinttype = {arXiv},
  eprint = {2608.11720},
  doi = {10.48550/arXiv.2608.11720}
}

@article{cole-vishkin-1986-deterministic-coin-tossing-with,
  author = {Richard Cole and Uzi Vishkin},
  title = {Deterministic Coin Tossing with Applications to Optimal Parallel List Ranking},
  journal = {Information and Control},
  volume = {70},
  number = {1},
  pages = {32--53},
  year = {1986},
  doi = {10.1016/S0019-9958(86)80023-7}
}

@article{collins-sniady-2006-integration-with-respect-to-the,
  author = {Beno{\^\i}t Collins and Piotr {\'S}niady},
  title = {Integration with Respect to the {Haar} Measure on Unitary, Orthogonal and Symplectic Group},
  journal = {Communications in Mathematical Physics},
  number = {3},
  volume = {264},
  pages = {773--795},
  year = {2006},
  doi = {10.1007/s00220-006-1554-3}
}

@article{d-amore-lievonen-2026-superlogarithmic-gap-result-for,
  author = {Francesco d'Amore and Henrik Lievonen},
  title = {Superlogarithmic Gap Result for {LCL}s on Trees in Quantum-{LOCAL}},
  journal = {CoRR},
  volume = {abs/2608.16854},
  year = {2026},
  eprinttype = {arXiv},
  eprint = {2608.16854},
  doi = {10.48550/arXiv.2608.16854}
}

@book{davidson-1996-c-algebras-by-example,
  author = {Kenneth R. Davidson},
  title = {{$C^*$}-Algebras by Example},
  publisher = {American Mathematical Society},
  year = {1996},
  doi = {10.1090/fim/006},
  series = {Fields Institute Monographs},
  volume = {6}
}

@inproceedings{dhar-kujawa-etal-2024-local-problems-in-trees-across-a,
  author = {Anubhav Dhar and Eli Kujawa and Henrik Lievonen and Augusto Modanese and Mikail M{\"u}ft{\"u}o{\u{g}}lu and Jan Studen{\'y} and Jukka Suomela},
  editor = {Silvia Bonomi and Letterio Galletta and Etienne Rivi{\`e}re and Valerio Schiavoni},
  title = {Local Problems in Trees Across a Wide Range of Distributed Models},
  booktitle = {28th International Conference on Principles of Distributed Systems, OPODIS 2024, Lucca, Italy, December 11-13, 2024},
  series = {LIPIcs},
  volume = {324},
  pages = {27:1--27:17},
  year = {2024},
  publisher = {Schloss Dagstuhl - Leibniz-Zentrum f{\"u}r Informatik},
  doi = {10.4230/LIPIcs.OPODIS.2024.27}
}

@inproceedings{fischer-ghaffari-2017-sublogarithmic-distributed,
  author = {Manuela Fischer and Mohsen Ghaffari},
  editor = {Andr{\'e}a W. Richa},
  title = {Sublogarithmic Distributed Algorithms for {L}ov{\'a}sz Local Lemma, and the Complexity Hierarchy},
  booktitle = {31st International Symposium on Distributed Computing, DISC 2017, Vienna, Austria, October 16-20, 2017},
  series = {LIPIcs},
  volume = {91},
  pages = {18:1--18:16},
  year = {2017},
  publisher = {Schloss Dagstuhl - Leibniz-Zentrum f{\"u}r Informatik},
  doi = {10.4230/LIPIcs.DISC.2017.18}
}

@inproceedings{fraigniaud-magniez-ziccardi-2027-no-distributed-quantum,
  author = {Pierre Fraigniaud and Fr{\'e}d{\'e}ric Magniez and Isabella Ziccardi},
  title = {No Distributed Quantum Advantage for 3-Coloring Rooted Trees and 2-Coloring Even Cycles},
  year = {2027},
  eprinttype = {arXiv},
  eprint = {2607.04852},
  doi = {10.48550/arXiv.2607.04852},
  booktitle = {Proceedings of the 38th Annual ACM-SIAM Symposium on Discrete Algorithms, {SODA} 2027, Philadelphia, Pennsylvania, USA}
}

@inproceedings{gavoille-kachigar-zemor-2019-localisation-resistant,
  author = {Cyril Gavoille and Ghazal Kachigar and Gilles Z{\'e}mor},
  editor = {Robert Mercas and Daniel Reidenbach},
  title = {Localisation-Resistant Random Words with Small Alphabets},
  booktitle = {Combinatorics on Words - 12th International Conference, WORDS 2019, Loughborough, UK, September 9-13, 2019, Proceedings},
  series = {Lecture Notes in Computer Science},
  volume = {11682},
  pages = {193--206},
  year = {2019},
  publisher = {Springer},
  doi = {10.1007/978-3-030-28796-2_15}
}

@inproceedings{gavoille-kosowski-markiewicz-2009-what-can-be-observed,
  author = {Cyril Gavoille and Adrian Kosowski and Marcin Markiewicz},
  editor = {Idit Keidar},
  title = {What Can Be Observed Locally?},
  booktitle = {Distributed Computing, 23rd International Symposium, DISC 2009, Elche, Spain, September 23-25, 2009. Proceedings},
  series = {Lecture Notes in Computer Science},
  volume = {5805},
  pages = {243--257},
  year = {2009},
  publisher = {Springer},
  doi = {10.1007/978-3-642-04355-0_26}
}

@inproceedings{ghaffari-harris-kuhn-2018-on-derandomizing-local,
  author = {Mohsen Ghaffari and David G. Harris and Fabian Kuhn},
  editor = {Mikkel Thorup},
  title = {On Derandomizing Local Distributed Algorithms},
  booktitle = {59th IEEE Annual Symposium on Foundations of Computer Science, FOCS 2018, Paris, France, October 7-9, 2018},
  pages = {662--673},
  year = {2018},
  publisher = {IEEE Computer Society},
  doi = {10.1109/FOCS.2018.00069}
}

@inproceedings{ghaffari-kuhn-maus-2017-on-the-complexity-of-local,
  author = {Mohsen Ghaffari and Fabian Kuhn and Yannic Maus},
  editor = {Hamed Hatami and Pierre McKenzie and Valerie King},
  title = {On the complexity of local distributed graph problems},
  booktitle = {Proceedings of the 49th Annual ACM SIGACT Symposium on Theory of Computing, STOC 2017, Montreal, QC, Canada, June 19-23, 2017},
  pages = {784--797},
  year = {2017},
  publisher = {ACM},
  doi = {10.1145/3055399.3055471}
}

@article{goldberg-plotkin-shannon-1988-parallel-symmetry,
  author = {Andrew V. Goldberg and Serge A. Plotkin and Gregory E. Shannon},
  title = {Parallel Symmetry-Breaking in Sparse Graphs},
  journal = {SIAM Journal on Discrete Mathematics},
  volume = {1},
  number = {4},
  pages = {434--446},
  year = {1988},
  doi = {10.1137/0401044}
}

@inproceedings{grunau-rozhon-brandt-2022-the-landscape-of-distributed,
  author = {Christoph Grunau and V{\'a}clav Rozho{\v{n}} and Sebastian Brandt},
  editor = {Alessia Milani and Philipp Woelfel},
  title = {The Landscape of Distributed Complexities on Trees and Beyond},
  booktitle = {PODC '22: ACM Symposium on Principles of Distributed Computing, Salerno, Italy, July 25 - 29, 2022},
  pages = {37--47},
  year = {2022},
  publisher = {ACM},
  doi = {10.1145/3519270.3538452}
}

@article{gur-li-2026-impossibility-of-one-way-one-round-quantum,
  author = {Tom Gur and Longcheng Li},
  title = {Impossibility of One-Way One-Round Quantum 4-Coloring via Matrix-Space Stability},
  journal = {CoRR},
  volume = {abs/2609.09091},
  year = {2026},
  eprinttype = {arXiv},
  eprint = {2609.09091},
  doi = {10.48550/arXiv.2609.09091}
}

@article{haagerup-1979-an-example-of-a-non-nuclear-c-algebra,
  author = {Uffe Haagerup},
  title = {An example of a non nuclear {$C^*$}-algebra, which has the metric approximation property},
  journal = {Inventiones Mathematicae},
  number = {3},
  volume = {50},
  pages = {279--293},
  year = {1979},
  doi = {10.1007/bf01410082}
}

@misc{hirvonen-suomela-2020-distributed-algorithms-2020,
  _nodoi = {true},
  author = {Juho Hirvonen and Jukka Suomela},
  title = {Distributed Algorithms 2020},
  url = {https://jukkasuomela.fi/da2020/},
  year = {2020}
}

@article{holroyd-hutchcroft-levy-2018-finitely-dependent-cycle,
  author = {Alexander E. Holroyd and Tom Hutchcroft and Avi Levy},
  title = {{Finitely dependent cycle coloring}},
  journal = {Electronic Communications in Probability},
  volume = {23},
  year = {2018},
  doi = {10.1214/18-ecp118}
}

@article{holroyd-hutchcroft-levy-2020-mallows-permutations-and,
  author = {Alexander E. Holroyd and Tom Hutchcroft and Avi Levy},
  title = {{Mallows permutations and finite dependence}},
  journal = {The Annals of Probability},
  number = {1},
  volume = {48},
  year = {2020},
  doi = {10.1214/19-aop1363}
}

@article{holroyd-liggett-2016-finitely-dependent-coloring,
  author = {Alexander E. Holroyd and Thomas M. Liggett},
  title = {{Finitely Dependent Coloring}},
  journal = {Forum of Mathematics, Pi},
  eid = {e9},
  volume = {4},
  year = {2016},
  doi = {10.1017/fmp.2016.7}
}

@article{kohonen-korhonen-etal-2017-distributed-colour-reduction,
  author = {Jukka Kohonen and Janne H. Korhonen and Christopher Purcell and Jukka Suomela and Przemys{\l}aw Uzna{\'n}ski},
  title = {Distributed Colour Reduction Revisited},
  journal = {CoRR},
  volume = {abs/1709.00901},
  year = {2017},
  eprinttype = {arXiv},
  eprint = {1709.00901},
  doi = {10.48550/arXiv.1709.00901}
}

@inproceedings{korman-sereni-viennot-2011-toward-more-localized-local,
  author = {Amos Korman and Jean-S{\'e}bastien Sereni and Laurent Viennot},
  editor = {Cyril Gavoille and Pierre Fraigniaud},
  title = {Toward more localized local algorithms: removing assumptions concerning global knowledge},
  booktitle = {Proceedings of the 30th Annual ACM Symposium on Principles of Distributed Computing, PODC 2011, San Jose, CA, USA, June 6-8, 2011},
  pages = {49--58},
  year = {2011},
  publisher = {ACM},
  doi = {10.1145/1993806.1993814}
}

@article{kuhn-moscibroda-wattenhofer-2016-local-computation,
  author = {Fabian Kuhn and Thomas Moscibroda and Roger Wattenhofer},
  title = {Local Computation: Lower and Upper Bounds},
  journal = {Journal of the ACM},
  volume = {63},
  number = {2},
  pages = {17:1--17:44},
  year = {2016},
  doi = {10.1145/2742012}
}

@incollection{le-gall-2025-recent-developments-in-quantum-distributed,
  author = {Fran{\c{c}}ois {Le Gall}},
  title = {{Recent Developments in Quantum Distributed Algorithms}},
  booktitle = {Algorithmic Foundations for Social Advancement},
  publisher = {Springer Nature Singapore},
  pages = {263--273},
  year = {2025},
  doi = {10.1007/978-981-96-0668-9_17}
}

@inproceedings{le-gall-nishimura-rosmanis-2019-quantum-advantage-for,
  author = {Fran{\c{c}}ois {Le Gall} and Harumichi Nishimura and Ansis Rosmanis},
  editor = {Rolf Niedermeier and Christophe Paul},
  title = {Quantum Advantage for the {LOCAL} Model in Distributed Computing},
  booktitle = {36th International Symposium on Theoretical Aspects of Computer Science, STACS 2019, Berlin, Germany, March 13-16, 2019},
  series = {LIPIcs},
  volume = {126},
  pages = {49:1--49:14},
  year = {2019},
  publisher = {Schloss Dagstuhl - Leibniz-Zentrum f{\"u}r Informatik},
  doi = {10.4230/LIPIcs.STACS.2019.49}
}

@article{le-gall-rosmanis-2022-non-trivial-lower-bound-for-3,
  author = {Fran{\c{c}}ois {Le Gall} and Ansis Rosmanis},
  title = {Non-trivial lower bound for 3-coloring the ring in the quantum {LOCAL} model},
  journal = {CoRR},
  volume = {abs/2212.02768},
  year = {2022},
  eprinttype = {arXiv},
  eprint = {2212.02768},
  doi = {10.48550/arXiv.2212.02768}
}

@article{linial-1992-locality-in-distributed-graph-algorithms,
  author = {Nathan Linial},
  title = {Locality in Distributed Graph Algorithms},
  journal = {SIAM Journal on Computing},
  volume = {21},
  number = {1},
  pages = {193--201},
  year = {1992},
  doi = {10.1137/0221015}
}

@article{maus-tonoyan-2022-linial-for-lists,
  author = {Yannic Maus and Tigran Tonoyan},
  title = {Linial for lists},
  journal = {Distributed Computing},
  volume = {35},
  number = {6},
  pages = {533--546},
  year = {2022},
  doi = {10.1007/s00446-022-00424-y}
}

@book{naaijkens-2017-quantum-spin-systems-on-infinite,
  author = {Pieter Naaijkens},
  title = {Quantum Spin Systems on Infinite Lattices: A Concise Introduction},
  publisher = {Springer International Publishing},
  year = {2017},
  doi = {10.1007/978-3-319-51458-1},
  series = {Lecture Notes in Physics},
  volume = {933}
}

@article{naor-1991-a-lower-bound-on-probabilistic-algorithms-for,
  author = {Moni Naor},
  title = {A Lower Bound on Probabilistic Algorithms for Distributive Ring Coloring},
  journal = {SIAM Journal on Discrete Mathematics},
  volume = {4},
  number = {3},
  pages = {409--412},
  year = {1991},
  doi = {10.1137/0404036}
}

@article{naor-stockmeyer-1995-what-can-be-computed-locally,
  author = {Moni Naor and Larry J. Stockmeyer},
  title = {What Can be Computed Locally?},
  journal = {SIAM Journal on Computing},
  volume = {24},
  number = {6},
  pages = {1259--1277},
  year = {1995},
  doi = {10.1137/S0097539793254571}
}

@article{panconesi-rizzi-2001-some-simple-distributed-algorithms,
  author = {Alessandro Panconesi and Romeo Rizzi},
  title = {Some simple distributed algorithms for sparse networks},
  journal = {Distributed Computing},
  volume = {14},
  number = {2},
  pages = {97--100},
  year = {2001},
  doi = {10.1007/PL00008932}
}

@book{peleg-2000-distributed-computing-a-locality-sensitive,
  author = {David Peleg},
  title = {Distributed Computing: A Locality-Sensitive Approach},
  publisher = {Society for Industrial and Applied Mathematics},
  year = {2000},
  doi = {10.1137/1.9780898719772}
}

@article{picardello-1986-positive-definite-functions-and-l-p,
  author = {Massimo A. Picardello},
  title = {Positive definite functions and {$L^p$} convolution operators on amalgams},
  journal = {Pacific Journal of Mathematics},
  number = {1},
  volume = {123},
  pages = {209--221},
  year = {1986},
  doi = {10.2140/pjm.1986.123.209}
}

@inproceedings{rozhon-ghaffari-2020-polylogarithmic-time-deterministic,
  author = {V{\'a}clav Rozho{\v{n}} and Mohsen Ghaffari},
  editor = {Konstantin Makarychev and Yury Makarychev and Madhur Tulsiani and Gautam Kamath and Julia Chuzhoy},
  title = {Polylogarithmic-time deterministic network decomposition and distributed derandomization},
  booktitle = {Proceedings of the 52nd Annual ACM SIGACT Symposium on Theory of Computing, STOC 2020, Chicago, IL, USA, June 22-26, 2020},
  pages = {350--363},
  year = {2020},
  publisher = {ACM},
  doi = {10.1145/3357713.3384298}
}

@inproceedings{suomela-2020-landscape-of-locality-invited-talk,
  author = {Jukka Suomela},
  editor = {Susanne Albers},
  title = {Landscape of Locality (Invited Talk)},
  booktitle = {17th Scandinavian Symposium and Workshops on Algorithm Theory, SWAT 2020, T{\'o}rshavn, Faroe Islands, June 22-24, 2020},
  series = {LIPIcs},
  volume = {162},
  pages = {2:1--2:1},
  year = {2020},
  publisher = {Schloss Dagstuhl - Leibniz-Zentrum f{\"u}r Informatik},
  doi = {10.4230/LIPIcs.SWAT.2020.2}
}

@book{watrous-2018-the-theory-of-quantum-information,
  author = {John Watrous},
  title = {The Theory of Quantum Information},
  publisher = {Cambridge University Press},
  year = {2018},
  doi = {10.1017/9781316848142}
}

\appendix
\section{Deferred Proofs}

\subsection{Proof of \texorpdfstring{\cref{lem:row-col-bound}}{lem:row-col-bound}}
\PropRowColBound*
\begin{proof}
    Decompose $Y(g,h)$ as
    \[
    Y^<(\ketbra{g}{h}) = \sum_{i: \phi_i(g,h) \prec (g,h)} e_{\phi_i(g,h)}
    \quad\text{and}\quad
    Y^>(\ketbra{g}{h}) = \sum_{i: (g,h) \prec \phi_i(g,h)} e_{\phi_i(g,h)}
    \]
    so that $Y(X) = Y^<(X) + Y^>(X)$ for all $X\in \fS$. We bound $\norm{ Y^<(X) }_F$ and the bound on $\norm{Y^>(X)}_F$ is obtained by an identical argument after flipping the ordering. We write $d^- = \sup_{x\in \Gamma} d^-(x)$ and $d^+ = \sup_{x\in \Gamma} d^+(x)$ for succinctness.

    To bound $\norm{Y^<(X)}_F$, let us first consider the coefficients of $Y^<(X)$ individually: for $x\in \Gamma$, note that
    \[
    [Y^<(X)](x) = \sum_{g,h} X(g,h) \sum_{i: \phi_i(g,h) \prec (g,h)} 1(x = \phi_i(g,h)) = \sum_{i: x \prec \phi_i(x)} X(\phi_i(x))
    \]
    where we use that $x = \phi_i(g,h)$ iff $\phi_i(x) = (g,h)$ because each $\phi_i$ is an involution. We therefore have that
    \begin{align*}
        \norm{Y^<(X)}_F^2 
        = \sum_{x\in \Gamma} |[Y^<(X)](x)|^2 
        = \sum_{x\in \Gamma} \abs*{ \sum_{i: x \prec \phi_i(x)} X(\phi_i(x)) }^2
        \leq d^+ \sum_{x\in \Gamma} \sum_{i: x \prec \phi_i(x)} \abs{X(\phi_i(x))}^2 
    \end{align*}
    where the inequality uses Cauchy-Schwarz and the fact that the sum is over at most $d^+$ terms. Now we make the change of variable $y = \phi_i(x)$ for each $i$ and obtain that 
    \begin{align*}
        \norm{ Y^<(X) }_F^2 
        \leq d^+ \sum_{y \in \Gamma} \sum_{i: \phi_i(y) \prec y} \abs{X(y)}^2 
        \leq d^+ d^- \sum_{y\in \Gamma} \abs{X(y)}^2 = d^+ d^- \norm{X}_F^2 \ ,
    \end{align*}
    where the second inequality is by definition of $d^-$. This concludes the proof since 
    \[
    \norm{Y(X)}_F \leq \norm{Y^<(X)}_F + \norm{Y^>(X)}_F 
    \leq 2\sqrt{d^+d^-}\norm{X}_F \ . \qedhere
    \]
\end{proof}

\subsection{Proof of \texorpdfstring{\Cref{lem:norm-submul}}{lem:norm-submul}}
\label[appendix]{appx:norm-submul}
\NormSubmul*
\begin{proof}
    By the triangle inequality, we have that $\norm{X \cdot Y}_w \leq \sum_{k,\ell, p, q \geq 0} \norm{ X_{k,\ell} \cdot Y_{p,q} }_w$ hence we focus on bounding the norm of each $X_{k,\ell} \cdot Y_{p,q}$. We expand and reorder the coefficients of this product in terms where all the indeterminates have the same length:
    \begin{align*}
    X_{k,\ell} \cdot Y_{p,q} 
    &= \sum_{g,g',h,h'} X_{k,\ell}(g,h) Y_{p,q}(g',h')[gg',hh']
    = \sum_{a,b \geq 0} Z_{a,b}
    \shortintertext{where}
    Z_{a,b} &= \sum_{\stackrel{|gg'| = k + p - 2a}{|hh'| = \ell + q - 2b}} 
        X_{k,\ell}(g,h) Y_{p,q}(g',h')[gg',hh'] 
    \end{align*}
    for all $0 \leq a \leq \min(k,p)$ and $0 \leq b \leq \min(\ell,q)$.
    All the indeterminates $[gg',hh']$ in $Z_{a,b}$ have $|gg'| = k + p - 2a$ and $|hh'| = \ell + q - 2b$ hence $Z_{a,b} = (Z_{a,b})_{k + p - 2a, \ell + q - 2b}$ meaning that $\norm{ Z_{a,b} }_w = 2^{(k+\ell)/2} 2^{(p+q)/2} 2^{-(a+b)} \norm{ Z_{a,b} }_F$.
    
    Let us now bound $\sum_{a,b \geq 0} \norm{ Z_{a,b} }_F$. To bound this sum, we express the coefficients of $Z_{a,b}$ explicitly. All the terms that contribute to the coefficient $[gg', hh']$ are of the form $X_{k,\ell}(gr^{-1}, ht^{-1}) Y_{p,q}(rg', th')$ with $|r| = a$ and $|t| = b$, where $r$ and $t$ are the parts that got erased in $gg'$ and $hh'$. In other words, for all $g,g',h,h'$ such that $|g| = k - a$, $|g'| = p - a$, $|h| = \ell - b$, $|h'| = q - b$, $|gg'| = |g| + |g'|$, and $|hh'| = |h| + |h'|$, we have that
    \[
    Z_{a,b}(gg', hh') = \sum_{r,t} X_{k,\ell}(gr^{-1}, ht^{-1}) Y_{p,q}(rg', th')
    \]
    where the sum is over the $r,t$ with $|r| = a$ and $|t| = b$.
    By applying Cauchy-Schwarz, we get
    \[
    |Z_{a,b}(gg', hh')|^2 \leq 
        \paren*{ \sum_{r,t} |X_{k,\ell}(gr^{-1}, ht^{-1})|^2 }
        \paren*{ \sum_{r,t} |Y_{p,q}(rg', th')|^2 }
    \]
    where the sum is over $|r| = a$ and $|t| = b$. First, observe that $\norm{ Z_{a,b} }_F^2 = \sum_{g,g',h,h'} |Z_{a,b}(gg', hh')|^2$ where the sum is over all $|g| = k-a$, $|g'| = p-a$, $|h| = \ell - b$, $|h'| = q - b$, $|gg'| = |g| + |g'|$, and $|hh'| = |h| + |h'|$. Indeed, each non-zero coefficient $Z_{a,b}(\alpha,\beta)$ is such that $|\alpha| = k + p - 2a$ and $|\beta| = \ell + q - 2b$, which means there is a unique way to write $\alpha = gg'$ and $\beta = hh'$ because for the given length there cannot be any cancellations in $gg'$ or $hh'$, i.e., $|\alpha| = |g| + |g'|$ and $|\beta| = |h| + |h'|$.
    Using that in the upper bound above the first factor depends only on $g,h$ and the second factor only on $g', h'$, we get
    \begin{align*}
    \norm{ Z_{a,b} }_F^2
    &\leq \paren*{\sum_{g,h} \sum_{r,t} |X_{k,\ell}(gr^{-1}, ht^{-1})|^2 }
         \paren*{\sum_{g',h'} \sum_{r,t} |Y_{p,q}(rg', th')|^2 }
    \end{align*}
    Finally, we claim that those two sums are respectively $\norm{ X_{k,\ell} }_F^2$ and $\norm{ Y_{p,q} }_F^2$. Indeed, each non-zero coefficient $X_{k,\ell}(\alpha,\beta)$ has $|\alpha| = k$ and $|\beta| = \ell$, which means that there is a unique decomposition $\alpha = gr^{-1}$ and $\beta = ht^{-1}$ with $|g| = k - a$, $|r| = a$, $|h| = \ell - b$, and $|t| = b$ because there cannot be any cancellations in $gr^{-1}$ or $ht^{-1}$, i.e., $|\alpha| = |g| + |r|$ and $|\beta| = |h| + |t|$. The same argument applies to the second sum with the $Y_{p,q}(\alpha,\beta)$.
    
    Overall, we have that
    \begin{align*}
    \norm{ X_{k,\ell} \cdot Y_{p,q} }_w 
    &\leq 2^{(k+\ell)/2} 2^{(p+q)/2} \sum_{a,b \geq 0} 2^{-(a+b)} \norm{ Z_{a,b} }_F \\
    &\leq 2^{(k+\ell)/2} 2^{(p+q)/2} \paren*{ \sum_{a,b \geq 0} 2^{-(a+b)} }
         \norm{ X_{k,\ell} }_F \norm{ Y_{p,q} }_F \\
    &= 2^{(k+\ell)/2} 2^{(p+q)/2} \cdot 4 \cdot
        \norm{ X_{k,\ell} }_F \norm{ Y_{p,q} }_F
    \end{align*}
    Hence 
    \[
    \norm{X \cdot Y}_w 
    \leq \sum_{k,\ell,p,q \geq 0} \norm{ X_{k,\ell} \cdot Y_{p,q} }_w 
    \leq 4 \cdot \sum_{k,\ell,p,q \geq 0} 2^{(k+\ell)/2} 2^{(p+q)/2} \norm{ X_{k,\ell} }_F \norm{ Y_{p, q} }_F
    = 4\norm{ X }_w \norm{ Y }_w
    \] 
    which concludes the proof.
\end{proof}

\subsection{Proof of \texorpdfstring{\Cref{lem:normalize-povm}}{lem:normalize-povm}}
  \label[appendix]{appx:proof-normalize-povm}

We first prove a proposition bounding the energy of an operator
in terms of its distance to a zero-energy operator.

\begin{proposition} \label{lem:energy-close}
Let $Z\in\mathrm{L}(\C^N\otimes\C^N)$ be PSD with $\Es(Z)=0$.
Then, for any PSD operator $M\in\mathrm{L}(\C^N\otimes\C^N)$ satisfying $0\preceq M\preceq I$
and any factorization $M=\Gamma^\dagger \Gamma$ with $\Gamma\in\mathrm{L}(\C^N\otimes\C^N)$, we have
\[
\Es(M) \leq 4N \|\Gamma - \sqrt{Z}\|_F^2.
\]
In particular, taking $\Gamma=\sqrt{M}$ gives $\Es(M) \leq 4N \|\sqrt{M} - \sqrt{Z}\|_F^2$.
\end{proposition}

\begin{proof}
    By \Cref{lem:zero-energy}, $\Es(Z)=0$ gives an orthogonal decomposition
    $\mathcal{V}\oplus \mathcal{V}^\perp = \C^N$ such that
    \[\supp(Z) \subseteq \mathcal{V} \otimes \bar{\mathcal{V}}^\perp.\]

    Let $Q:=\Pi_{\mathcal{V}^\perp}\otimes I$ denote the orthogonal projector onto the subspace
    $\mathcal{V}^\perp\otimes \C^N$. Since $\mathcal{V}^\perp\otimes \C^N$ is orthogonal to
    $\mathcal{V}\otimes \bar{\mathcal{V}}^\perp$ and $\supp(\sqrt{Z}) = \supp(Z) \subseteq \mathcal{V}\otimes \bar{\mathcal{V}}^\perp$, we have
    \(
    \sqrt{Z}Q=0.
    \)
    Then the leakage of $M$ into $\mathcal{V}^\perp\otimes \C^N$ is bounded by
    \begin{align*}
        \alpha
        :=\Tr(QM)
        =\Tr(Q\Gamma^\dagger \Gamma Q)
        =\|\Gamma Q\|_F^2
        =\|(\Gamma-\sqrt{Z})Q\|_F^2
        \leq\|\Gamma-\sqrt{Z}\|_F^2.
    \end{align*}

    Similarly, let $R:=I\otimes \Pi_{\bar{\mathcal{V}}}$ denote the orthogonal projector
    onto $\C^N\otimes \bar{\mathcal{V}}$. The leakage of $M$ into this subspace is bounded by
    \[
    \beta := \Tr(RM) \leq \|\Gamma-\sqrt{Z}\|_F^2.
    \]

    We now relate these leakage bounds to the energy of $M$. Define
    \(A:=\Tr_\As(M)^\top\) and \(B:=\Tr_\Bs(M)\).
    Since \(0\preceq M\preceq I_N\otimes I_N\), we have \(0\preceq A,B\preceq NI_N\).
    By the triangle inequality,
    \begin{align*}
        \sqrt{\Es(M)}
        =\sqrt{\Tr(AB)}=\|B^{1/2}A^{1/2}\|_F
        &\leq
        \|B^{1/2}\Pi_{\mathcal{V}^\perp}A^{1/2}\|_F
        +
        \|B^{1/2}\Pi_{\mathcal{V}}A^{1/2}\|_F.
    \end{align*}

    The first term can be bounded as
    \[
    \|B^{1/2}\Pi_{\mathcal{V}^\perp} A^{1/2}\|_F
    = \sqrt{\Tr(\Pi_{\mathcal{V}^\perp}A\Pi_{\mathcal{V}^\perp}\cdot
    \Pi_{\mathcal{V}^\perp} B \Pi_{\mathcal{V}^\perp})}
    \leq \sqrt{N \Tr(\Pi_{\mathcal{V}^\perp} B)}
    = \sqrt{N\alpha},
    \]
    where the inequality is by
    $\Pi_{\mathcal{V}^\perp}A\Pi_{\mathcal{V}^\perp}\preceq 
    \Pi_{\mathcal{V}^\perp}(NI)\Pi_{\mathcal{V}^\perp}=
    N\Pi_{\mathcal{V}^\perp}$,
    and the final equality is by
    \[
    \alpha
    = \Tr((\Pi_{\mathcal{V}^\perp}\otimes I)M)
    = \Tr(\Pi_{\mathcal{V}^\perp}\Tr_\Bs(M)) = \Tr(\Pi_{\mathcal{V}^\perp}B).
    \]

    Similarly, using $B\preceq NI_N$, we obtain
    \[
    \|B^{1/2}\Pi_{\mathcal{V}}A^{1/2}\|_F
    \leq \sqrt{N\Tr(\Pi_{\mathcal{V}} A)}
    = \sqrt{N\beta},
    \]
    where
    \[
    \beta
    = \Tr((I\otimes\Pi_{\bar{\mathcal{V}}})M)
    = \Tr(\Pi_{\bar{\mathcal{V}}}\Tr_\As(M))
    = \Tr(\Pi_{\mathcal{V}} A).
    \]
    Consequently,
    \[
    \Es(M)
    \leq
    N(\sqrt{\alpha}+\sqrt{\beta})^2
    \leq
    2N(\alpha+\beta)
    \leq
    4N\|\Gamma-\sqrt{Z}\|_F^2.
    \]
    In particular, we can take $\Gamma=\sqrt{M}$.
\end{proof}

Then we prove \Cref{lem:normalize-povm}.

\begin{proof}[Proof of \Cref{lem:normalize-povm}]
Let $S:=\sum_{a=1}^{q}\tilde{M}_a$, and $S^+$ be the pseudoinverse of $S$.
Define $\{M_a\}_{a\in[q]}$ by
\[
M_a :=
\begin{cases}
R_1+Q & a=1\\
R_a & 2\leq a\leq q
\end{cases},
\quad\text{where}\quad 
R_a := \sqrt{S^+} \tilde{M}_a \sqrt{S^+},
\quad
Q := \Pi_{\ker(S)}.
\]
Then $\{M_a\}_{a\in[q]}$ is a POVM
 since $M_a \succeq 0$ and
\[
\sum_{a=1}^{q}M_a
=
\sqrt{S^{+}}
\left(\sum_{a=1}^{q}\tilde{M}_a\right)
\sqrt{S^{+}}
+Q
=
\sqrt{S^{+}}S\sqrt{S^{+}}
+\Pi_{\ker(S)}
=I.
\]

First, we bound the energy of $R_a$ for each $a\in[q]$.
We can factor $R_a$ as
\[
R_a = (\sqrt{\tilde{M}_a}\sqrt{S^+})^\dagger(\sqrt{\tilde{M}_a}\sqrt{S^+}).
\]
Since $\Es(\tilde{M}_a)=0$, \Cref{lem:energy-close} gives
\[
\Es(R_a) \leq 4N \delta_a,
\qquad\text{where}\qquad
\delta_a := \|\sqrt{\tilde{M}_a}\sqrt{S^+}-\sqrt{\tilde{M}_a}\|_F^2.
\]
Let $\{\lambda_j\}_{j=1}^{N^2}$ denote the eigenvalues of $S$.
Summing the $\delta_a$ gives
\begin{align*}
\sum_{a=1}^{q}\delta_a
=
\sum_{a=1}^{q}
\left\|
\sqrt{\tilde{M}_a}(\sqrt{S^+}-I)
\right\|_F^2
&=
\Tr\!\left[
(\sqrt{S^+}-I)
\left(\sum_{a=1}^{q}\tilde{M}_a\right)
(\sqrt{S^+}-I)
\right]\\
&= \Tr\!\left[
(\sqrt{S^+}-I)S(\sqrt{S^+}-I)
\right]
=
\sum_{\lambda_j>0}(\sqrt{\lambda_j}-1)^2.
\end{align*}
Since \((\sqrt{\lambda}-1)^2\leq(\lambda-1)^2\) for all $\lambda\geq 0$, we have
\begin{equation*}
\sum_{a=1}^{q}\delta_a
\leq
\sum_{j=1}^{N^2}(\lambda_j-1)^2
=
\|S-I\|_F^2
\leq\epsilon.
\end{equation*}
Thus
\begin{equation}\label{eq:energy-bound-i}
\sum_{a=1}^{q}\Es(R_a)
\leq 4N\sum_{a=1}^{q}\delta_a
\leq 4N\epsilon.
\end{equation}

Next, we bound the difference $\Es(M_1)-\Es(R_1)$.
For $X\in\{R_1,Q\}$, write
$A_X:=\Tr_\As(X)^\top$ and $B_X:=\Tr_\Bs(X)$.
Since $M_1=R_1+Q$, we have
\[
\Es(M_1)-\Es(R_1)
=\Tr(A_QB_Q)+\Tr(A_{R_1}B_Q)+\Tr(A_QB_{R_1}).
\]
For each $X\in\{R_1,Q\}$, the bound $0\preceq X\preceq I$ gives
$0\preceq A_X,B_X\preceq NI_N$.
Thus the three terms can be bounded by
\[
\Tr(A_QB_Q), \Tr(A_{R_1}B_Q)\leq N\Tr(B_Q)=N\Tr(Q),
\quad
\Tr(A_QB_{R_1})\leq N\Tr(A_Q)=N\Tr(Q).
\]
Moreover,
\[
\Tr(Q)=\dim(\ker S)
=\sum_{\lambda_j=0}1
\leq\sum_{j=1}^{N^2}(\lambda_j-1)^2
=\|S-I\|_F^2
\leq\epsilon.
\]
Therefore $\Es(M_1)-\Es(R_1)\leq 3\epsilon N$.
Combining this with \eqref{eq:energy-bound-i}, we obtain
\[
\sum_{a=1}^{q}\Es(M_a)
\leq \sum_{a=1}^{q}\Es(R_a)+3\epsilon N
\leq 4\epsilon N+3\epsilon N
=7\epsilon N. \qedhere
\]
\end{proof}

\section{Positivity of the 1-Dependent Distribution}
\label{app:proof-1-dep}

In this section, we prove that the construction from \cref{sec:fin-dep} indeed defines a 1-dependent distribution over proper $q$-colorings of $\bZ$. 
We give a proof \emph{using more colors than \cref{res:main-epsilon}}. This proof uses the contraction rate of \cref{lem:contraction}, which was proven for the weighted norm and can be improved if we use only Frobenius norms.

The contraction lemma (\cref{lem:contraction}) shows that $\norm{R_{n+1}}_w \leq \lambda\norm{R_n}_w$ where $R_n = E - S(V_n)$ and $\lambda \in (0,1)$, which implies that the $V_1, V_2, \ldots$ form a Cauchy sequence as for every $n \leq m$,
\begin{equation}
\label{eq:cauchy}
\norm{V_m - V_n}_F 
\leq \sum_{i = n}^{m-1} \norm{V_{i+1}-V_i}_F 
\leq \frac{2}{s}\sum_{i=n}^{m-1} \norm{R_i}_w 
\leq \frac{2}{s}\sum_{i=n}^{m-1} \lambda^{i-n} \norm{R_n}_w
\leq \frac{2\norm{R_n}_w}{(1 - \lambda)s} 
\leq \frac{2 \lambda^{n}}{(1 - \lambda)s} 
\end{equation}
which goes to zero as $n \to \infty$.
We can see each of the $V_n$ as an element from the complete space $\ell_2(G \times G)$, hence there exists some $V \in \ell_2(G \times G)$ that is the limit of the $V_n$. We prove the following theorem.

\begin{theorem}
    For $s \geq 2 \cdot 10^6$, 
    the $T_{\pm i} = P_{\pm i} V P_{\mp i}$ for $i\in[s]$ satisfy \ref{T1}--\ref{T3}.
\end{theorem}

\begin{proof}
The proof of \ref{T1} is direct from $P_{-i} P_{+i} = 0$ (\cref{lem:orth-projections}) as $T_{\pm i}^2 = P_{\pm i} V P_{\mp i} P_{\pm i} V P_{\mp i} = 0$. \ref{T2} states that $E = S(V)$. By \ref{V2} in \cref{thm:fin-dep-V}, we have that $\norm{E - S(V_n)}_F \leq \lambda^n \to 0$ as $n \to \infty$, and since $S : \ell_2(G \times G) \to \ell_2(G \times G)$ is a continuous map, we obtain that $S(V) = E$. The remainder of this proof is dedicated to proving \ref{T3}. In fact, we prove that all the $T_{\pm i}^{(n)} = P_{\pm i} V_n P_{\mp i}$ satisfy \ref{T3}, hence so do the $T_{\pm i}$ by continuity of the scalar product. That is, for every finite word $w$, one can prove that $\norm{ T_w^{(n)} - T_w }_{\operatorname{op}} \to 0$ because $\norm{ V_n - V }_F \to 0$ and the $P_{\pm i}$ are orthogonal projections. We omit this step for the sake of succinctness. 

Recall that $V_1 = (2/s)E$ and $R_1 = E - (2/s)S(E)$ and that $S(E) = (s/2)E - (1/2)\sum_i \ketbra{i}{i}$; hence $\norm{R_1}_w = (1/s)\norm*{\sum_i \ketbra{i}{i}}_w \leq 2/\sqrt{s}$. As the $V_n$ form a Cauchy sequence, we get that every $V_n$ is close to $V_1$. Plugging $\norm{R_1}_w$ into \cref{eq:cauchy}, we get a concrete bound
\[
\norm*{ V_n - V_1 } \leq \frac{4}{(1 - \lambda)s^{3/2}} = \gamma \ .
\]

Consider the case of $T_a^{(1)}$ first. For each color $a = \sigma i$ with $i\in[s]$ and $\sigma \in \set{+1,-1}$, define $\ket{u_a} = \ket{e} + \sigma \ket{i}$ and $\ket{v_a} = \ket{e} - \sigma \ket{i}$ so that $P_{a}\ket{e} = \ket{u_a}/2$ and $P_{-a}\ket{e} = \ket{v_a}/2$.
Consider words of single letters. Then $T_{a}^{(1)} = (2/s)P_{a} E P_{-a} = (1/q)\ketbra{u_a}{v_a}$. For any pairs of colors $a,b$, we have
\[
\braket{v_a}{u_b} = \braket{u_b}{v_a} = \begin{cases}
    0 &\text{if } a = b \\
    2 &\text{if } a = -b \\
    1 &\text{if } a \notin \set{ b, -b }
\end{cases} \ .
\]
and so it is easy to see that $p_1(a) = \braUket{e}{T^{(1)}_a}{e} = 1/q$ for all colors $a$, and for all words $w = w_1 w_2 \ldots w_k$ we have
\[
\braUket{e}{T_w^{(1)}}{e} = \braUket{e}{T_{w_1}^{(1)} T_{w_2}^{(1)} \ldots T_{w_k}^{(1)}}{e}
= (1/q)^{k} \braket{e}{u_{w_1}} \braket{v_{w_1}}{u_{w_2}} \ldots \braket{v_{w_{k-1}}}{u_{w_k}} \braket{v_{w_k}}{e} \geq 0
\]

Fix some $n \geq 1$. We prove by induction on $|w|$ that for all proper colorings $w$ ending in $a$
\[
p_n(w) = \braUket{e}{T_w^{(n)}}{e} > 0
\quad\text{and}\quad
\norm*{ \frac{\bra{e}T_w^{(n)}}{p_n(w)} - \bra{v_a} }_2 \leq \eta := 4q\gamma
= \frac{32}{(1-\lambda)\sqrt{s}}\ .
\]
We define $\Delta_a = P_{+a}(V_n - V_1)P_{-a}$ so that 
$T_a^{(n)} = (1/q) \ketbra{u_a}{v_a} + \Delta_a$
and as the difference between $V_n$ and $V_1$ is small, we argue it will have little impact on positivity. For the base case, $w = a = \pm i$ for some $i\in[s]$. Observe that
\[
p_n(a) 
= \braUket{e}{T_a^{(n)}}{e} = p_1(a) + \braUket{e}{\Delta_a}{e} 
\in [p_1(a) \pm \gamma]
\]
where we use that
$
|\braUket{e}{\Delta_a}{e}| 
\leq \norm{V_n - V_1}_{\operatorname{op}} \leq \norm{V_n - V_1}_{F} \leq \gamma
$. In particular $p_n(a) > 0$. Then we have that $\bra{e}T_a^{(1)} = p_1(a)\bra{v_a}$, hence $\bra{e}T_a^{(n)} = p_1(a) \bra{v_a} + \bra{e}\Delta_a$, which means that 
\[
\norm*{ \frac{\bra{e}T_a^{(n)}}{p_n(a)} - \bra{v_a} }_2 
= \frac{ \norm*{ \bra{e}{T_a^{(n)}} - p_n(a) \bra{v_a} }_2 }{p_n(a)}
\leq \frac{\abs{p_1(a) - p_n(a)}\norm{v_a}_2 + \norm*{ \bra{e}\Delta_a }_2}{p_n(a)}
\leq \frac{q(1 + \sqrt{2})\gamma}{1 - q\gamma}
\]
where the last inequality uses the derivation on $p_n(a)$ and $\norm*{ \bra{e}\Delta_a }_2 \leq \norm{ \Delta_a }_{\operatorname{op}} \leq \norm{\Delta_a}_F \leq \gamma$.

For the induction step, consider a proper coloring $w$ which ends in $a$. We prove that $p_n(wb) > 0$ and $\norm*{\frac{\bra{e}T_{wb}^{(n)}}{p_n(wb)} - \bra{v_b}} \leq \eta$ for all $b \neq a$. Let $\bra{ \ell }$ be so that $\bra{e}T_{w}^{(n)} = p_n(w)\bra{\ell}$. We have
\begin{equation}
\label{eq:pn-ratio}
p_n(wb) 
= \braUket{e}{T_{wb}^{(n)}}{e}
= \braUket{e}{T_{w}^{(n)} \cdot T_b^{(n)}}{e}
= p_n(w) \braUket{\ell}{T_b^{(n)}}{e}
= p_n(w) \paren*{ p_1(b)\braket{\ell}{u_b} + \braUket{\ell}{\Delta_b}{e} } \ .
\end{equation}
We have $\abs{\braUket{\ell}{\Delta_b}{e}} \leq \norm{\Delta_b}_F \norm{\ell}_2 /\sqrt{2} \leq (\sqrt{2} + \eta)\gamma/\sqrt{2}$ because $\norm{P_{-b}\ket{e}}_2 = 1/\sqrt{2}$, $\norm{v_a} = \sqrt{2}$ and $\norm{\ell - v_a}_2 \leq \eta$ by the induction hypothesis. We have
$\braket{\ell}{u_b} = \braket{v_a}{u_b} + \braket{\ell - v_a}{u_b}$ hence
\[
\abs{ \braket{\ell}{u_b} - \braket{v_a}{u_b} } 
\leq \norm{\ell - v_a}_2\norm{u_b}_2
\leq \eta\sqrt{2} \ .
\]
Observe that when $a \neq b$, the scalar product $\braket{v_a}{u_b}$ is either one or two.
Overall, this means that
\[
p_n(wb) = \braUket{e}{T_{wb}^{(n)}}{e}
\geq p_n(w)\paren*{ 1/q - \eta\sqrt{2}/q - (\sqrt{2} + \eta)\gamma/\sqrt{2} } > 0
\]
which is positive because $p_n(w) > 0$ and $\gamma = \Theta(s^{-3/2}) \ll 1/q$, $\eta/q = 4\gamma$, and $\eta\gamma = 4q\gamma^2 = \Theta(s^{-2}) \ll 1/q$ for our choice of $s$. For the second part of the induction hypothesis, note that $\bra{e}T_{wb}^{(n)} = p_n(w)\bra{\ell}T_b^{(n)} = p_n(w)p_1(b)\braket{\ell}{u_b}\bra{v_b} + p_n(w)\bra{\ell}\Delta_b$. In particular,
\begin{align*}
\norm*{ \frac{\bra{e}T_{wb}^{(n)}}{p_n(wb)} - \bra{v_b} }_2 
&= \frac{ \norm*{ \bra{e}T_{wb}^{(n)} - p_n(wb)\bra{v_b} }_2 }{p_n(wb)}\\
&\leq \frac{ p_n(w) }{p_n(wb)} \paren*{ \abs*{ p_1(b)\braket{\ell}{u_b} - \frac{p_n(wb)}{p_n(w)} } \cdot \norm{ \bra{v_b} }_2 + \norm{ \bra{\ell}\Delta_b  }_2 } \ .
\end{align*}
Notice that \cref{eq:pn-ratio} rewrites as 
\[
\frac{ p_n(wb) }{p_n(w)}
= p_1(b)\braket{\ell}{u_b} + \braUket{\ell}{\Delta_b}{e}
\]
hence
\begin{align*}
\norm*{ \frac{\bra{e}T_{wb}^{(n)}}{p_n(wb)} - \bra{v_b} }_2 
&\leq \frac{ p_n(w) }{p_n(wb)} \paren*{ \abs*{ \braUket{\ell}{\Delta_b}{e} } \cdot \norm{ \bra{v_b} }_2 + \norm{ \bra{\ell}\Delta_b }_2 } \\
&\leq \frac{1}{p_1(b)} \frac{ \abs*{ \braUket{\ell}{\Delta_b}{e} } \cdot \norm{ \bra{v_b} }_2 + \norm{ \bra{\ell}\Delta_b }_2 }{\braket{\ell}{u_b} + \braUket{\ell}{\Delta_b}{e}/p_1(b) } \\
&\leq q \frac{ (\sqrt{2} + \eta)\gamma + (\sqrt{2} + \eta)\gamma }{1 - \eta\sqrt{2} - q(\sqrt{2} + \eta)\gamma/\sqrt{2} }\\
&= q \frac{ 2\sqrt{2}\gamma + 2\gamma\eta }{ 1 - \eta\sqrt{2} - q\gamma - q\eta\gamma/\sqrt{2} }
\leq \eta
\end{align*}
where the last inequality uses our value of $\gamma = \gamma(s)$ and that $s \geq 2 \cdot 10^6$.
\end{proof}

\section{Lower Bound on the Number of Qubits Used}
\label[appendix]{appx:lower-bound}

This section proves that our construction 
uses the optimal number of qubits up to a constant factor,
via the weighted energy introduced in \Cref{appx:weighted-energy}.

\begin{theorem} \label{thm:qubit-lower-bound}
    Fix an integer $q\ge 2$. Any one-way one-round anonymous quantum
    $q$-coloring algorithm that colors the directed cycle such that each
    edge is monochromatic with probability at most $\eps$ requires each node
    to send $\Omega(\log(1/\eps))$ qubits.
\end{theorem}

The lower bound extends to quantum-LOCAL, 
showing that $\Omega(\log n)$ qubits are necessary.

\begin{corollary} \label{cor:qubit-lower-bound-whp}
    Fix an integer $q\geq 2$. Any one-way one-round quantum-LOCAL algorithm
    that $q$-colors directed $n$-cycles with probability at least
    $1-1/\poly(n)$ requires each node to send $\Omega(\log n)$ qubits.
\end{corollary}

\begin{proof}
    Suppose for contradiction that there is such an algorithm $\mathcal{A}$
    that succeeds with probability at least $1-n^{-c}$ for some constant
    $c>0$, in which each node sends $k = o(\log n)$ qubits.
    Then we construct a one-way one-round anonymous quantum algorithm
    $\mathcal{A}'$ as follows:
    \begin{enumerate}
        \item Each node independently samples an identifier uniformly from $[n^{c+2}]$.
        \item Each node runs $\mathcal{A}$ with its sampled identifier.
    \end{enumerate}
    
    The sampled identifiers are pairwise distinct with probability at least $1-1/n^c$.
    Hence in $\mathcal{A}'$, each edge is monochromatic with probability
    at most $\epsilon:=n^{-c} + n^{-c} = 2n^{-c}$, and each node sends only $k$ qubits.
    Applying \Cref{thm:qubit-lower-bound} to $\mathcal{A}'$ gives
    \[
    k = \Omega(\log(1/\epsilon)) = \Omega(\log(n^c/2)) = \Omega(\log n),
    \]
    which contradicts the assumption that $k = o(\log n)$.
\end{proof}

\subsection{Additional Preliminaries}

We first introduce some more discussions on the
energy function~\cite{gur-li-2026-impossibility-of-one-way-one-round-quantum}.
Note that we use the notation $P$ to denote finite-dimensional operators,
instead of $M$ as in the main text, so that $M$ can be used 
for matrices in $\C^{N\times N}$.

\paragraph{Vectorization.} We first recall the vectorization map. 
Let $\mathcal{A}:=\C^{N_1}$ and $\mathcal{B}:=\C^{N_2}$ with standard
bases $\{e_i\}_{i\in[N_1]}$ and $\{e_j\}_{j\in[N_2]}$, respectively.
Let $E_{i,j}\in \C^{N_1\times N_2}$ be the matrix unit with a $1$ in entry $(i,j)$ and $0$ elsewhere. Define
\[
\mathrm{vec} : \C^{N_1\times N_2} \to \mathcal{A} \otimes \mathcal{B} \quad \text{by} \quad \mathrm{vec}(E_{i,j}) := e_i \otimes e_j.
\]
Then $\mathrm{vec}$ is a bijection and preserves the Frobenius inner product, i.e.,
for any two matrices $X,Y\in \C^{N_1\times N_2}$, we have
\[
\Tr(X^\dagger Y) = \mathrm{vec}(X)^\dagger\mathrm{vec}(Y).
\]

\subsubsection{Spectral Formula of Energy} \label{sec:spectral-formula}

The name \emph{energy} for $\Es(P)$ in \Cref{def:energy} comes from the
following spectral formula, which expresses it as a sum of squared norms.

\begin{lemma}[Corollary of {\cite[Lemma 4.2]{gur-li-2026-impossibility-of-one-way-one-round-quantum}}] \label{lem:spectral-unweighted}
    Let $P\in \mathrm{L}(\C^N\otimes\C^N)$ be PSD, and let
    $\{M_k\}\subseteq \C^{N\times N}$ be any finite family of matrices such that
    \[
    P = \sum_k \mathrm{vec}(M_k)\mathrm{vec}(M_k)^\dagger.
    \]
    Then
    \[
    \Es(P) = \sum_{k,\ell} \|M_k M_\ell\|_F^2.
    \]
    In particular, $\{M_k\}$ can be obtained from a spectral decomposition of $P$.
\end{lemma}

In particular, when $P$ is an orthogonal projector, $\{M_k\}$ can
be taken to be an orthonormal basis of the matrix subspace onto which $P$
projects. Then $\Es(P)$ is the \emph{multiplicative energy} of this subspace, which 
measures how strongly the matrices in the subspace couple by multiplication.
The term energy reflects the fact that it can be expressed as a sum of squared norms, and also that it is monotone.

\subsubsection{Weighted Energy} \label[appendix]{appx:weighted-energy}

To characterize one-way one-round algorithms with general quantum messages,
we need the weighted version of the energy.
Here, we define it directly via a spectral formula similar to \Cref{lem:spectral-unweighted}.

\begin{definition}[$\Lambda$-energy, {\cite[Definition 4.1 and Lemma 4.2]{gur-li-2026-impossibility-of-one-way-one-round-quantum}}] \label{def:weighted}
    Given a PSD operator $P\in \mathrm{L}(\C^N\otimes\C^N)$ and a PSD weight matrix
    $\Lambda\in \C^{N\times N}$, define the $\Lambda$-energy of
    $P$ as
    \[
    \mathcal{E}_\Lambda(P) := \sum_{k,\ell} \|\Lambda M_k \Lambda M_\ell \Lambda\|_F^2,
    \]
    where $\{M_k\}\subseteq \C^{N\times N}$ is any finite family of matrices
    satisfying $P = \sum_k \mathrm{vec}(M_k)\mathrm{vec}(M_k)^\dagger$.
\end{definition}

This definition is independent of the choice of $\{M_k\}$.
The following lemma expresses the collision probability of any one-way one-round
anonymous algorithm using $\Lambda$-energy, where the dimension $N$ depends only
on the number of qubits sent.

\begin{lemma}[Corollary of {\cite[Lemma 4.7]{gur-li-2026-impossibility-of-one-way-one-round-quantum}}] \label{lem:weighted}
    Given any one-way one-round anonymous quantum $q$-coloring algorithm
    on a directed $n$-cycle outputting colors $(c_i\in[q])_{i\in\mathbb{Z}_n}$,
    suppose each node sends $k$ qubits,
    and let $N := 2^k$. Then there exist a POVM $\{P_a\}_{a\in[q]}$ on
    $\C^N\otimes\C^N$ and a PSD matrix $\Lambda\in \C^{N\times N}$ with
    $\|\Lambda\|_F^2 = 1$ such that
    \[
    \Pr[c_i=c_{i+1}] = \sum_{a\in[q]} \Es_\Lambda(P_a) \quad\forall i\in\mathbb{Z}_n.
    \]
\end{lemma}

\begin{proof}
    Since nodes can purify their private randomness and intermediate measurements by keeping them coherently in a
    private register, we assume without loss of generality that the algorithm has the following form:
    \begin{enumerate}
    \item Each node $v_i$ prepares a pure state $\ket{\Phi}$ on registers
    $\mathsf{A}_i\otimes\mathsf{B}_i$, where the message register
    $\mathsf{B}_i$ has $k$ qubits and hence dimension $N=2^k$; the private register
    $\mathsf{A}_i$ has arbitrary finite dimension, and we assume $\dim\mathsf{A}_i \geq N$ 
    without loss of generality.
    \item Each node $v_i$ sends $\mathsf{B}_i$ to its successor $v_{i+1}$,
    keeps $\mathsf{A}_i$, and receives $\mathsf{B}_{i-1}$ from $v_{i-1}$.
    \item Each node $v_i$ applies a $q$-outcome POVM $\{Q_a\}_{a\in[q]}$ to
    $\mathsf{B}_{i-1}\otimes\mathsf{A}_i$ and outputs the measurement outcome
    as its color $c_i$.
    \end{enumerate}

    We construct the required $\{P_a\}_{a\in[q]}$ and $\Lambda$ as follows.
    Since $\dim\mathsf{B}_i = N$, the Schmidt decomposition of $\ket{\Phi}$ has
    at most $N$ nonzero terms, and can be written as
    \[
    \ket{\Phi} = \sum_{j=1}^N \lambda_j \ket{u_j}_{\mathsf{A}}\ket{w_j}_{\mathsf{B}},
    \]
    where $\lambda_j\geq 0$ are Schmidt coefficients with $\sum_j \lambda_j^2 = 1$,
    $\{\ket{u_j}\}_{j\in[N]}$ is orthonormal in $\mathsf{A}$, and
    $\{\ket{w_j}\}_{j\in[N]}$ is an orthonormal basis of $\mathsf{B}$.

    Define the weight matrix $\Lambda := \mathrm{diag}(\lambda_1,\ldots,\lambda_N)$
    and the associated canonical quantum state
    \[
    \ket{\Psi} := \mathrm{vec}(\Lambda) = \sum_{j=1}^N \lambda_j \ket{j}\otimes\ket{j}
    \in \C^N\otimes\C^N.
    \]

    Let $U:\C^N\to\mathsf{A}$ be the isometry $\ket{j}\mapsto\ket{u_j}$ and
    $V:\C^N\to\mathsf{B}$ be the unitary $\ket{j}\mapsto\ket{w_j}$, so that
    $\ket{\Phi} = (U\otimes V)\ket{\Psi}$. Then define 
    the new POVM element as 
    \[
    P_a := (V\otimes U)^\dagger\, Q_a\, (V\otimes U) \in \mathrm{L}(\C^N\otimes\C^N)
    \quad\text{for } a\in[q].
    \]

    Then the algorithm which prepares $\ket{\Psi}$ and applies
    $\{P_a\}_{a\in[q]}$ has the same output distribution as the original one.
    Finally, applying \cite[Lemma 4.7]{gur-li-2026-impossibility-of-one-way-one-round-quantum}
    gives
    \[
    \Pr[c_i = c_{i+1}] = \sum_{a\in[q]} \Es_\Lambda(P_a). \qedhere
    \]
\end{proof}

\subsection{Proof of \Cref{thm:qubit-lower-bound}}

The key technical lemma underlying \Cref{thm:qubit-lower-bound} is 
a dimension-dependent lower bound on the total weighted energy of a POVM,
stated as follows.

\begin{lemma} \label{lem:dim-dep-lower-bound}
For any POVM $\{P_a\}_{a\in[q]}$ on $\C^N\otimes\C^N$ and any PSD matrix
$\Lambda\in \C^{N\times N}$ with $\|\Lambda\|_F^2 = 1$, we have
\[
    \sum_{a\in[q]} \Es_\Lambda(P_a) \geq \frac1{q^3 N^{11}}.
\]
\end{lemma}

\begin{proof}
    Since energy is nonnegative, it suffices to lower bound $\Es_\Lambda(P_a)$ for some $a\in[q]$.
    Define a quantum state $\ket{\Psi_\Lambda} := \mathrm{vec}(\Lambda)$.
    Since
    \[
    \sum_{a\in[q]} \braUket{\Psi_\Lambda}{P_a}{\Psi_\Lambda}
    = \braUket{\Psi_\Lambda}{\sum_{a\in[q]} P_a}{\Psi_\Lambda} = \norm{\Lambda}_F^2
    = 1,
    \]
    by averaging, there exists $a\in[q]$ for which $t := \braUket{\Psi_\Lambda}{P_a}{\Psi_\Lambda} \geq 1/q$.
    Consider the spectral decomposition of such a $P_a$ as
    \[
    P_a = \sum_{k} \mathrm{vec}(M_k)\mathrm{vec}(M_k)^\dagger,
    \]
    where $\{M_k\}_{k}\subseteq \C^{N\times N}$ are pairwise orthogonal matrices.
    Then, as $\Lambda$ is Hermitian,
    \[
    \sum_k \abs{ \Tr(\Lambda M_k) }^2
    = \sum_k \bra{\Psi_\Lambda} \mathrm{vec}(M_k)\mathrm{vec}(M_k)^\dagger \ket{\Psi_\Lambda}
    = \bra{\Psi_\Lambda} P_a \ket{\Psi_\Lambda} \geq 1/q.
    \]
    Since $P_a$ has rank at most $N^2$,
    there exists some $k$ for which $\abs{ \Tr(\Lambda M_k) }^2 \geq 1/(qN^2)$.
    By \Cref{def:weighted}, we have
    \[
    \Es_\Lambda(P_a) = \sum_{j,\ell} \|\Lambda M_j \Lambda M_\ell \Lambda\|_F^2
    \geq \|\Lambda M_k \Lambda M_k \Lambda\|_F^2.
    \]
    
    We now remove the last factor $\Lambda$.
    By orthogonality, $P_a \mathrm{vec}(M_k) = \|M_k\|_F^2 \mathrm{vec}(M_k)$,
    so $\|M_k\|_F^2$ is an eigenvalue of $P_a \preceq I$. 
    Since $\|M_k\|_F^2\leq 1$, we have $M_kM_k^\dagger \preceq I$, and thus
    \begin{align*}
    \|\Lambda M_k \Lambda M_k \Lambda\|_F^2
    = \|(\Lambda M_k)^2 \Lambda\|_F^2
    &= \Tr ((\Lambda M_k)^2 \Lambda \cdot \Lambda^\dagger (\Lambda M_k)^{2\dagger}) \\
    &\geq 
    \Tr ((\Lambda M_k)^2 \Lambda M_kM_k^\dagger \Lambda^\dagger (\Lambda M_k)^{2\dagger})
    = \|(\Lambda M_k)^3\|_F^2.
    \end{align*}

    Let $\mu_1, \ldots, \mu_N$ be the eigenvalues of $\Lambda M_k$ counted with multiplicities.
    Then we have
    \[
    \Es_\Lambda(P_a) \geq \|\Lambda M_k \Lambda M_k \Lambda\|_F^2 \geq \|(\Lambda M_k)^3\|_F^2
    \geq \sum_i |\mu_i|^6 \geq \frac{1}{N^5}\left(\sum_i |\mu_i|\right)^6,
    \]
    where the penultimate inequality uses the Schur decomposition of $\Lambda M_k$,
    and the last uses H\"older's inequality $\sum_i |\mu_i| \leq N^{5/6} \left(\sum_i |\mu_i|^6\right)^{1/6}$.
    Recall that $|\Tr(\Lambda M_k)|^2 \geq 1/(qN^2)$ and that the trace is the sum of the eigenvalues,
    so $\sum_i |\mu_i| \geq |\Tr(\Lambda M_k)|$, and thus
    \[
    \Es_\Lambda(P_a) \geq \frac{1}{N^5}|\Tr(\Lambda M_k)|^6 \geq \frac{1}{q^3 N^{11}}. \qedhere
    \]
\end{proof}

\begin{proof}[Proof of \Cref{thm:qubit-lower-bound}]
    Consider any one-way one-round anonymous quantum $q$-coloring algorithm
    in which each node sends $k$ qubits, and define the dimension $N := 2^k$.
    By \Cref{lem:weighted}, there is 
    a POVM $\{P_a\}_{a\in[q]}$ on $\C^N\otimes\C^N$ and a PSD weight matrix
    $\Lambda\in \C^{N\times N}$ with $\|\Lambda\|_F^2 = 1$ such that
    \[
        \sum_{a\in[q]} \Es_\Lambda(P_a) \leq \epsilon.
    \]
    By \Cref{lem:dim-dep-lower-bound}, we have
    \[
    \sum_{a\in[q]} \Es_\Lambda(P_a) \geq \frac1{q^3 N^{11}} = \frac1{q^3 2^{11k}}.
    \]
    Combining the two inequalities gives
    \[
        \epsilon \geq \frac1{q^3 2^{11k}}
        \quad\text{hence}\quad
        k \geq \frac{1}{11}\log\paren*{\frac{1}{q^3 \epsilon}} = \Omega(\log(1/\epsilon)). \qedhere
    \]
\end{proof}

\end{document}